\documentclass[doublecolumn,journal]{IEEEtran}
\IEEEoverridecommandlockouts

\def\BibTeX{{\rm B\kern-.05em{\sc i\kern-.025em b}\kern-.08em
    T\kern-.1667em\lower.7ex\hbox{E}\kern-.125emX}}
\usepackage{easyReview}
\usepackage{textcomp}
\usepackage{xcolor}
\usepackage{tikz}
\usepackage{cite,nicefrac}
\usepackage[ruled]{algorithm2e}
\usepackage{graphicx}
\usepackage{amsmath,amssymb,amsfonts}
\usepackage{bm}
\usepackage[version=4]{mhchem}
\usepackage{siunitx}
\usepackage{longtable,tabularx}
\usepackage{setspace}
\usepackage{subcaption}
\usepackage{array}
\usepackage{pdfpages}
\usepackage{lipsum}
\usepackage{wrapfig}
\usepackage{hyperref}
\usepackage{stmaryrd}
\usepackage{pifont}
\usepackage[stable]{footmisc}
\usepackage{stackengine} 
\usepackage{authblk}
\usepackage{mathtools}
\usepackage[utf8]{inputenc}
\usepackage{soul}
\usepackage{url}
\usepackage{mathtools}
\usepackage{amssymb}
\usepackage{amsmath,amsfonts,bbold}
\usepackage{algorithmic}

\usepackage{textcomp}
\usepackage{stfloats}
\usepackage{verbatim}
\usepackage{bbm}
\usepackage{mdframed}
\usepackage{amsthm}
\usepackage{multirow}
\usepackage{setspace}
\usepackage{enumitem}

\newtheorem{Fact}{Fact}
\newtheorem{lemma}{Lemma}

\newtheorem{theorem}{Theorem}
\newtheorem{Def}{Definition}
\newtheorem{Corollary}{Corollary}

\newtheorem{assumption}{Assumption}

\newtheorem{Remark}{Remark}

\newcommand{\mathbm}[1]{\bm{\mathcal{#1}}}

\DeclareMathOperator*{\minimize}{\textrm{minimize}}

\DeclareMathOperator*{\vecop}{\rm vec}
\DeclareMathOperator*{\var}{\rm Var}

\newcommand{\cmark}{\ding{51}}
\newcommand{\xmark}{\ding{55}}

\definecolor{mygreen}{RGB}{0, 128, 0}
\definecolor{mygreen2}{RGB}{0, 200, 128}

\newcommand{\algocomment}[1]{\item [] \texttt{\color{gray} // #1}}

\allowdisplaybreaks

\begin{document}

\title{Completing Radio Map Tensors via Factor Diffusion Posterior Sampling
}

\author{Le Xu, Rajesh Shrestha, and Xiao Fu
\thanks{
R. Shrestha is with the School of Electrical Engineering and
Computer Science, Oregon State University, OR, USA. 

L. Xu is with University of Granada, Spain. The work was developed when he was with the School of Electrical Engineering and
Computer Science, Oregon State University, OR, USA. 

X. Fu is with the Department of Electrical and Computer Engineering, University of Minnesota - Twin Cities, MN, USA; he was with the School of Electrical Engineering and
Computer Science, Oregon State University, OR, USA 
(email: xfu@umn.edu).

}
\thanks{
The work was supported in part by the National Science Foundation (NSF) under Projects NSF ECCS-2024058 and NSF CCF-2210004.}
}

\maketitle

\begin{abstract}
Radio map estimation (RME), also known as spectrum cartography, aims to recover radio signal strength across multiple domains (e.g., spatial and spectral domains) from sparse measurements, naturally leading to a tensor completion problem. Recent advances show that deep learning–based generative models can provide data-driven structural constraints for RME, yielding substantial performance improvements over traditional model-based approaches. Diffusion models (DMs) represent the current state of the art in generative modeling and hold strong promise for advancing RME. However, existing diffusion-based RME methods often require extensive training data, necessitate retraining when scenarios or measurement formats change, and lack recoverability guarantees.
This paper proposes a latent factor diffusion framework for spatio-spectral radio map tensor completion.
By training two independent DMs on single-emitter's spatial loss field and spectral signature and embedding them into the factors of the radio map tensor, the proposed approach enables diffusion posterior sampling-based recovery with reduced training complexity as well as enhanced measurement format/sensing pattern flexibility.
Recoverability and sample complexity guarantees under the proposed approach are also established. 
Numerical experiments validate the effectiveness of the proposed method, showing strong generalization across city environments unseen in training data.
\end{abstract}


\section{Introduction}
Radio map estimation (RME), sometimes referred to as spectrum cartography, aims to construct radio signal strength maps over spatial regions of interest using 
sparsely deployed sensors \cite{bi2019engineering,romero2022radio,bazerque2009distributed,Bazerque2011Splines}. In many practical settings, radio maps also span additional domains, such as frequency, resulting in high-order tensor representations \cite{Shrestha2022DeepSpectrum,Subash2024Quantized,romero2017learning,chen2025grid,sun2024integrated,Zhang2020SpectrumViaBlockTerm}. Radio maps are a key enabler of wireless environment awareness and support a wide range of wireless resource allocation and networking applications; see discussions in \cite{bi2019engineering,zeng2024tutorial}.

In essence, RME is an inverse problem that aims to recover complete radio map tensors from sub-sampled measurements. As with many inverse problems, effective recovery relies on structural constraints to mitigate ill-posedness. Classical approaches typically impose manually designed structural priors, such as low-rank structures in matrix or tensor completion \cite{Lee2017CGMcartography,Zhang2020SpectrumViaBlockTerm} and spatial smoothness through kernel-based interpolation \cite{Bazerque2011Splines}, or combinations of both \cite{sun2022propagation,sun2024integrated}. While effective in simple propagation environments (e.g., free space or regions with limited shadowing), these methods often degrade in more complex radio-frequency (RF) scenarios, particularly urban or indoor environments with severe shadowing \cite{Shrestha2022DeepSpectrum,Subash2024Quantized}. This degradation is expected, as handcrafted priors such as low rank and sparsity capture macro-scale structures but tend to fail at modeling local oscillations and fine-grained spatial variations present in realistic radio maps.

In recent years, deep learning-based methods have been introduced to provide data-driven structural constraints for RME. In particular, modern deep generative models learn underlying data distributions and can represent complex data characteristics through trained neural parameters \cite{goodfellow2016deep}. This capability makes generative models well suited for capturing intricate RF propagation behaviors using offline data, such as simulated or previously recorded measurements. A variety of generative modeling approaches have therefore been adopted for RME, including autoencoders (AEs) \cite{teganya2021deep,Teganya2020Datadriven,Shrestha2022DeepSpectrum}, variational autoencoders (VAEs) \cite{an2024radio}, generative adversarial networks (GANs) \cite{Subash2024Quantized,zhang2023rme}, untrained networks \cite{timilsina2025domain}, and implicit neural representers \cite{chen2025grid}. These methods have demonstrated promising performance in complex, noisy, and shadowing-intensive RF environments.

Among deep generative models, diffusion models (DMs) \cite{Ho2020DDPM,song2021scorebased,chan2024tutorial} have achieved state-of-the-art performance across a wide range of applications. Compared with earlier generative models such as VAEs and GANs, DMs provide improved data fidelity and sample diversity, and have emerged as a central paradigm in modern artificial intelligence. Motivated by these advances, DMs have recently been introduced to RME and have demonstrated promising reconstruction performance; see, e.g., \cite{gu2025generative,luo2025denoising,wang2025radiodiff,wang2024radiodiff,zhao20253d}. 

\medskip

\noindent
{\bf Challenges.}
Most existing DM–based RME methods (e.g., \cite{luo2025denoising,wang2025radiodiff,wang2024radiodiff,zhao20253d}) adopt a conditional generative modeling paradigm, where a DM is trained to predict complete radio maps in and end-to-end (E2E) manner, where the prediction is done via mapping Gaussian noise and auxiliary information, e.g., spatial measurements, emitter locations, sensor locations, and building maps, to the radio map domain. 
The recent work in \cite{gu2025generative} formulates the RME problem as a more classical linear inverse problem and applies diffusion posterior sampling (DPS) \cite{chung2023diffusion,daras2024survey} to estimate the radio maps from their under-sampled versions.
While these approaches show promising empirical performance, several challenges remain.

First, deployment flexibility is limited.
Conditional E2E DM formulations tightly couple learned score functions to the training scenarios, including sampling patterns and measurement formats. Consequently, changes in the measurement process, such as moving from continuous-valued measurements \cite{Shrestha2022DeepSpectrum,Subash2024Quantized} to quantized observations \cite{Subash2024Quantized}, typically require retraining.

Second, training cost is high.
All the existing methods train DM on radio maps with information aggregated from multiple emitters. 
As noted in \cite{Shrestha2022DeepSpectrum,xu2025radio,Subash2024Quantized}, learning generative models under information-aggregated settings often requires a large number of training samples due to the high-dimensional state space. Aggregated DM based approaches are also sensitive to training–testing mismatches (e.g., changes in the number of emitters). In addition, existing methods are mostly designed for recovering spatial information over a single band, and they
become computationally prohibitive for radio map \textit{tensor} recovery across multiple frequency bands.

Third, theoretical understanding is lacking.
Existing recoverability guarantees for deep learning–based RME rely on deterministic generative models, such as autoencoders and GANs, which map low-dimensional latent variables to ambient radio map spaces \cite{Shrestha2022DeepSpectrum,Subash2024Quantized,timilsina2025domain}. In contrast, diffusion models generate samples via 
iterative stochastic process
and lack such deterministic low-dimensional representations. As a result, recoverability guarantees for DM-based radio map estimation, especially for radio map tensors, are yet to be established.

\medskip

\noindent
{\bf Contributions.} This work develops a DM-based framework for estimating spatio-spectral radio map tensors from sparsely acquired samples. We address challenges in both methodological development and theoretical analysis for DM-based RME. The main contributions are summarized as follows:

\smallskip

\noindent
$\blacktriangleright$~{\it Latent Factor Diffusion Framework}:
Unlike existing DM-based RME approaches that rely on models trained on aggregated radio maps \cite{luo2025denoising,wang2025radiodiff,wang2024radiodiff,gu2025generative,zhao20253d}, we treat RME as a tensor inverse problem and impose diffusion priors in the latent domain. Motivated by the observation that learning generative models for single-emitter spatial loss fields (SLFs) requires significantly fewer training samples \cite{Shrestha2022DeepSpectrum,Subash2024Quantized}, and leveraging the spatio–spectral 
decomposition
of narrowband radio maps \cite{Bazerque2011Splines,romero2017learning,Lee2017CGMcartography,Zhang2020SpectrumViaBlockTerm}, we incorporate {\it two} independently trained DMs for single-emitter's spatial and spectral characteristics into the latent factors of the radio map tensor. This yields a latent factor diffusion posterior sampling algorithm that accommodates diverse measurement formats and auxiliary information.

\smallskip

\noindent
$\blacktriangleright$~{\it Recoverability Analysis}:
We establish recoverability guarantees for radio map tensor completion based on latent factor posterior sampling. Specifically, assuming that the diffusion-based sampler provides a sufficiently accurate approximation of the true posterior and that the learned diffusion models capture the spatial and spectral data distributions with reasonable fidelity, we show that accurate recovery is guaranteed given a sufficient number of measurements. The analysis extends the instance-optimal compressed sensing framework in \cite{jalal2021instance} to tensor-structured signals under non-Gaussian sensing models, and provides theoretical insights for DM-based radio map recovery.

\medskip
We validate the proposed methods through extensive simulations, demonstrating their effectiveness and efficiency under a variety of observation models.

\medskip
\noindent
\textbf{Notation:} In this paper, the notations $x \in \mathbb R$, $\bm x\in \mathbb R^I$, $\bm X \in \mathbb R^{I\times J}$ and $\mathbm X \in \mathbb R^{I\times J\times K}$ denote a scalar, a vector, a matrix and a tensor, respectively. The notation $\bm 1_K$ denotes the vector of all ones of length $K$ and $\bm I_K$ denotes $K \times K$ identity matrix.
The set $[I]$ denotes $\{1,\ldots,I\}$ for a positive integer $I$. The operators $\circ$, $\oast$ and $\otimes$ denote the outer product, Hadamard product, and Kronecker product, respectively. 
Probability density functions (pdfs), probabilities of events, and expectations are denoted by $p(\cdot)$, ${\rm Pr}(\cdot)$, and $\mathbb E[\cdot]$, respectively.
The indicator function of an event $A$ is denoted by $\mathbb{1}(A)$, which equals $1$ if $A$ occurs and $0$ otherwise.
The operator $\mathcal{Q}(\cdot)$ denotes the element-wise quantization, $\mathcal U(a,b)$ denotes the uniform distribution on the interval $[a,b]$, and $\Phi(\cdot)$ denotes the cumulative distribution function (CDF) of the standard Gaussian distribution.

\section{Background}
\subsection{Problem Statement}
We consider the classical spatio-spectral RME problem (see, e.g., \cite{romero2017learning,Zhang2020SpectrumViaBlockTerm,Shrestha2022DeepSpectrum,Bazerque2011Splines,Lee2017CGMcartography}) over a rectangular region with $I\times J$ grids and $K$ frequency bands. This spatio-spectral radio map is represented as a tensor $\mathbm X \in \mathbb{R}^{I \times J \times K}$, where $\mathbm X(i,j,k)$ denotes the RF signal strength at the $(i,j)$th spatial cell and the $k$th frequency band. Depending on applications and measure data, such signal strength can represent interference \cite{Bazerque2011Splines}, channel gain \cite{Lee2017CGMcartography}, or beamforming gain \cite{wang2025beamckm}.
To simplify the discussion, we follow \cite{Bazerque2011Splines,romero2017learning,Zhang2020SpectrumViaBlockTerm,Shrestha2022DeepSpectrum,Subash2024Quantized,xu2025radio,sun2024integrated} to consider the power interference propagation map, The developed technique, however, also applies to other types of maps.
In the considered case, $\mathbm X(i,j,:)\in\mathbb{R}^K$ represents the power spectral density (PSD) of the signal received at location $(i,j)$.
Note that the signal at every $\mathbm X(i,j,k)$ is often contributed by multiple emitters in the region.
Within the region of interest, measurements are collected at a subset of grid locations indexed by $\bm\varOmega = \{(i_n,j_n)\}_{n=1}^N$, where $|\bm\varOmega| = N$ denotes the number of measurements. Accordingly, the PSD vector $\mathbm X(i,j,:)$ is observed if $(i,j) \in \bm\varOmega$. That is, the sensor can scan across all the $K$ frequency bands at its location $(i,j)$. The scenario is illustrated in Fig.~\ref{fig:RME}.

Note that sampling such PSD ``fibers'' from the tensor $
\mathbm{X}$ can be represented as a classical linear sensing problem
\begin{align}\label{eq:classicIP}
    \bm y=\bm A \bm x + \bm \xi,
\end{align}
where $\bm \xi$ denotes the noise.
To see this, define the spatial selection matrix $\bm M_{\bm \varOmega}\in\{0,1\}^{N\times IJ}$ so that its $n$th row selects the spatial location $(i_n,j_n)\in \bm \varOmega$. Using the spatial index mapping $\ell(i,j)=(j-1)I+i$, we define
\begin{equation}
(\bm M_{\bm \varOmega})_{n,\,\ell(i,j)}
=
\begin{cases}
1, & {\text{if } (i,j)=(i_n,j_n),}\\
0, & \text{otherwise},
\end{cases}
\label{eq:Pomega}
\end{equation}
for $n\in[N]$, $i\in[I]$, and $j\in[J]$. 
In addition, let $\bm X_{\rm spec}\in \mathbb R^{K \times IJ}$ denote the matrix formed by stacking the spectral vectors of \(\mathbm X\) as columns; i.e., $\bm X_{\rm spec}\triangleq[\mathbm X(1,1,:),\mathbm{X}(2,1,:),\ldots,\mathbm X(I,J,:)]\in \mathbb{R}^{K\times IJ}$. Then the spectra observed at the sampled locations $\bm \varOmega$ can be written as
\begin{equation}
    \bm Y =  \bm X_{\rm spec} \bm M_{\bm \varOmega}^\top + \bm W \in \mathbb R^{K \times N},
\end{equation}
where $\bm W \in \mathbb{R}^{K \times N}$ denotes measurement noise. In particular, the $n$th column of $\bm Y\in \mathbb R^{K\times N}$ is given by $\mathbm X(i_n,j_n,:)+\bm W(:,n)$. Let $\bm x= \vecop (\bm X_{\rm spec}) \in \mathbb R^{IJK}$, $\bm y= \vecop(\bm Y)\in \mathbb R^{NK}$ and $\bm w = \vecop ( \bm W )\in\mathbb R^{NK}$. The measurement vector is then expressed as follows
\begin{equation}
    \bm y = (\bm M_{\bm \varOmega} \otimes \bm I_K  ) \bm x + \bm w, \label{eq: IP for RM}
\end{equation}
where we used the vectorization identity $\vecop (\bm B \bm C \bm D)=(\bm D^T \otimes \bm B) \vecop(\bm C)$.

\begin{figure}[!t]
    \centering
    \includegraphics[width=0.85\linewidth]{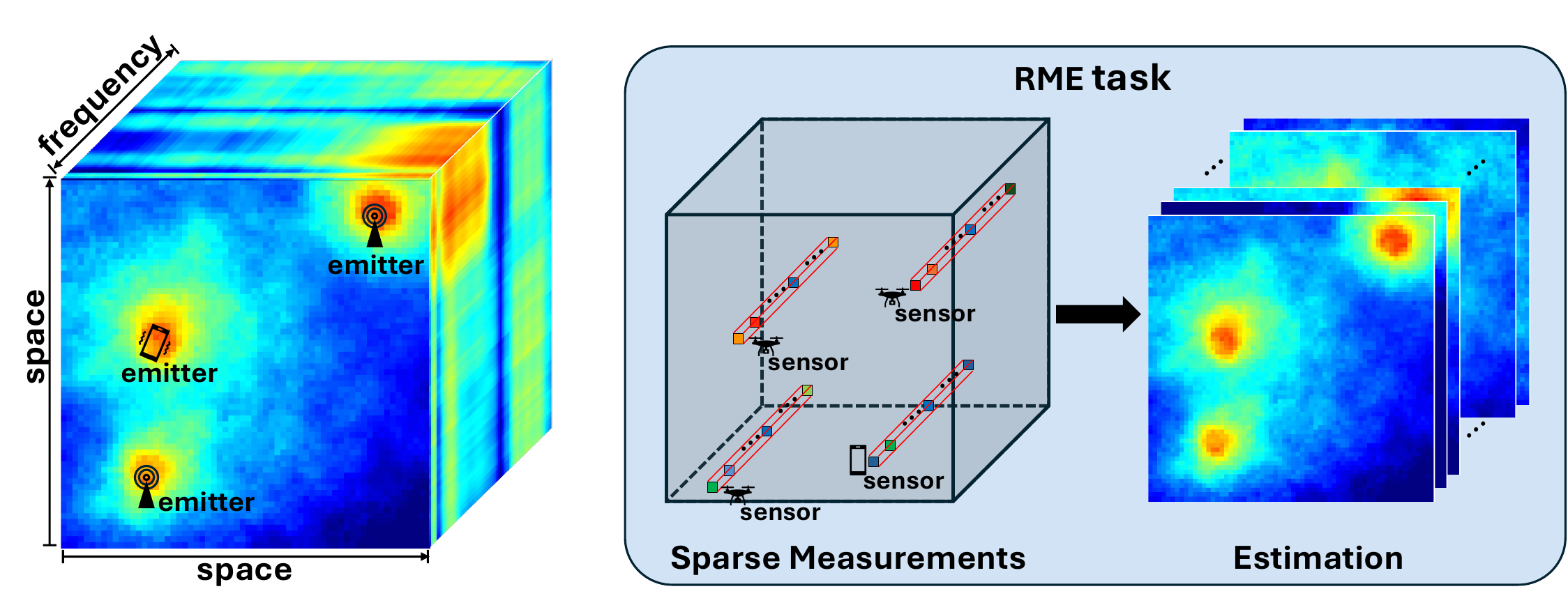}
    \caption{Left: scenario of interest. Right: RME as tensor sampling and recovery; figure adapted from \cite{Shrestha2022DeepSpectrum,Subash2024Quantized,xu2025radio}.}
    \label{fig:RME}
\end{figure}

\subsection{RME: From Model-based to Data-driven Paradigms}
Clearly, the RME task is an ill-posed, under-determined inverse problem, as $|\bm \varOmega|$ is often much smaller than $IJ$. Classical matrix/tensor completion approaches typically rely on model-based structural constraints to avoid spurious solutions. These include local smoothness \cite{Bazerque2011Splines,romero2017learning}, low-rank matrix or tensor models \cite{Lee2017CGMcartography,Zhang2020SpectrumViaBlockTerm} (on spatial slabs or the entire $\mathbm X$), and the integration of low-rank and interpolation models \cite{sun2022propagation,sun2024integrated}.
These methods exhibit strong performance in relatively simple RF settings (e.g., rural areas with minimal blockage), but are less competitive in handling severe shadowing effects arising from complex propagation conditions; see \cite{Shrestha2022DeepSpectrum,Subash2024Quantized}.

In recent years, deep learning approaches, especially generative models, were adopted
for RME (see, e.g., \cite{Levie2021radiounet,teganya2021deep,Teganya2020Datadriven,Shrestha2022DeepSpectrum,Subash2024Quantized}). Deep models are considered more capable of capturing complex data characteristics and thus more promising in complex RF environments.
Pioneering works \cite{teganya2021deep,Teganya2020Datadriven} employ an end-to-end autoencoder to recover $\mathbm X$ from $\mathbm Y$ (with the assistance of some auxiliary information). These methods are considered computationally heavy in multi-band settings, as they essentially have to model the distribution of the high-dimensional data $\mathbm X$, contributed by multiple emitters' activities.

\begin{figure}
    \centering
    \includegraphics[width=0.85\linewidth]{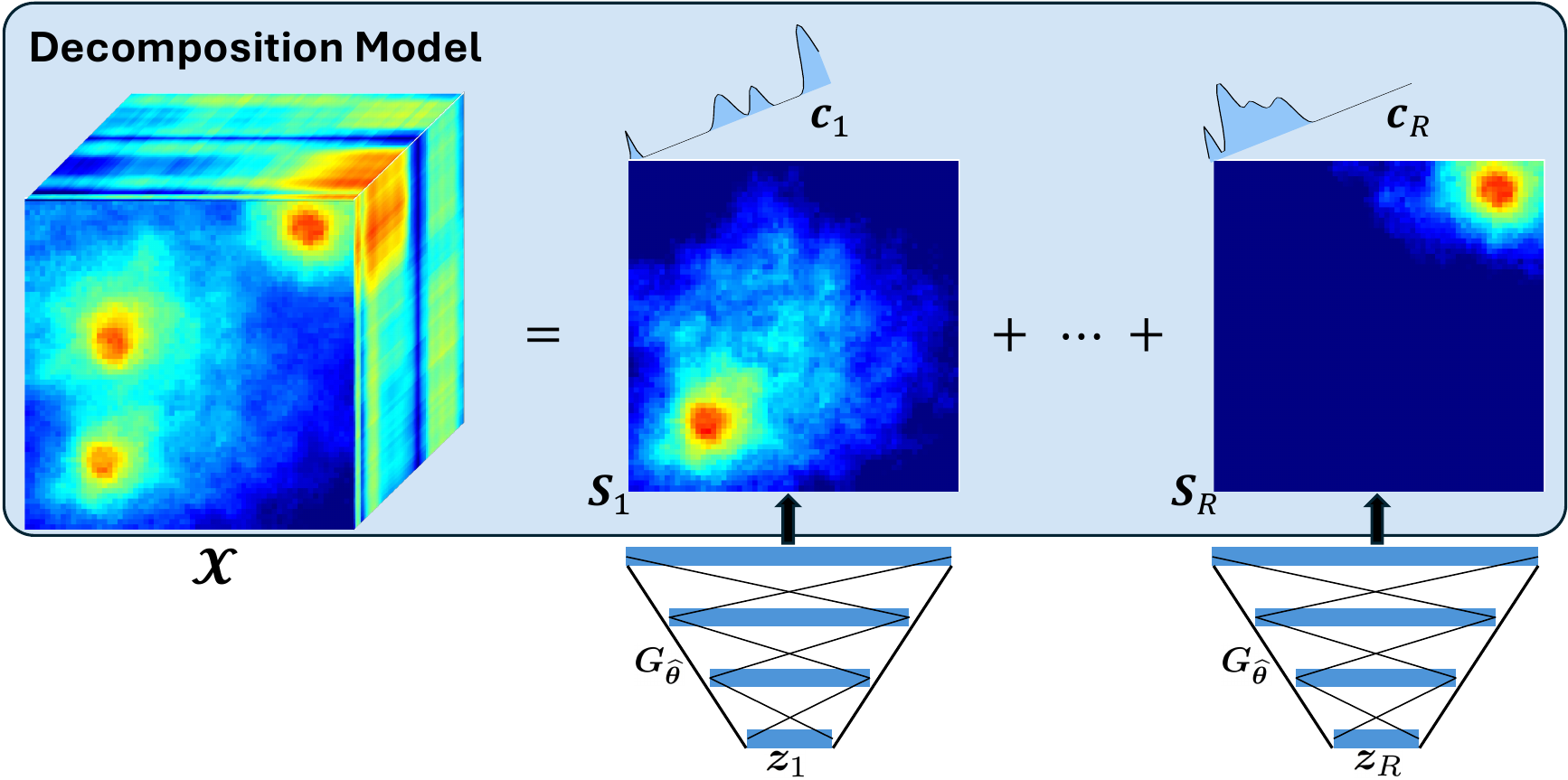}
    \caption{The spatio-spectral emitter decomposition and the single-emitter SFL generative model in \cite{Shrestha2022DeepSpectrum,Subash2024Quantized}; figure adapted from \cite{Shrestha2022DeepSpectrum,Subash2024Quantized}.}
    \label{fig:dowjons}
\end{figure}

To reduce learning complexity, the works in \cite{Shrestha2022DeepSpectrum,Subash2024Quantized,xu2025radio} adopt a classical spatio–spectral decomposition of radio maps (see, e.g., \cite{bazerque2009distributed,Bazerque2011Splines,romero2017learning,Zhang2020SpectrumViaBlockTerm}):
\begin{align}\label{eq:X=SC}
\mathbm X = \sum_{r=1}^R \bm S_r \circ \bm c_r,
\end{align}
where $R$ denotes the number of emitters, and $\bm S_r$ and $\bm c_r$ represent the \textit{spatial loss field} (SLF) and PSD of the $r$th emitter, respectively. Under this model, a generator $\bm G_{\bm\theta}:\mathbb{R}^d \rightarrow \mathbb{R}^{I\times J}$ is trained via GANs \cite{Subash2024Quantized} or AEs \cite{Shrestha2022DeepSpectrum} to model single-emitter SLFs, rather than aggregated radio maps. This substantially simplifies learning by avoiding the combinatorial complexity induced by multi-emitter superposition; see Fig.~\ref{fig:dowjons}.

Once trained, each $\bm S_r$ admits the representation $\bm S_r=\bm G_{{\bm\theta}}(\bm z_r)$ for a certain latent code $\bm z_r\in\mathbb{R}^d$ with $d\ll IJ$. This way, recovering $\mathbm X$ boils down to estimating $\{\bm z_r,\bm c_r\}$, e.g., via
\begin{align}\label{eq:opt-based}
    \minimize_{\{\bm z_r,\bm c_r\}}  \left\| \bm y  - \bm A {\rm vec}\left( 
             \sum_{r=1}^R \bm G_{{\bm \theta}}(\bm z_r) \circ \bm c_r \right) \right\|_{\rm 2}^2;
\end{align}
also see Fig.~\ref{fig:dowjons}.
Recoverability guarantees have been established for this framework \cite{Shrestha2022DeepSpectrum,Subash2024Quantized,timilsina2025domain}, showing that accurate recovery is ensured with high probability when the number of sensor locations reaches the order of $(d+K)R$.

\subsection{Existing Diffusion-Based RME and Challenges}

Classical generative models, e.g., AE, VAE, and GANs, have their respective limitations. In particular, AEs/VAEs are prone to over-smoothing, while GANs often suffer from min-max optimization instability, mode collapse, and limited sample diversity \cite{goodfellow2016deep}.
In recent years, DM emerged as a high-accuracy generative model with a manageable training process \cite{chan2024tutorial}. These properties make DMs a compelling choice for incorporation into RME.

\subsubsection{Diffusion Model Preliminaries}
DMs have a forward process and a reverse process. In the forward process, a clean sample $\bm x_0\sim p_{\rm data}$ is gradually corrupted by Gaussian noise \cite{chan2024tutorial}:
\begin{align}
    \bm x_t
    =
    \sqrt{\alpha_t}\,\bm x_{t-1}
    +
    \sqrt{1-\alpha_t}\,\bm \varepsilon_t,
    \quad
    \bm \varepsilon_t \sim \mathcal{N}(\bm 0,\bm I),
\end{align}
for $t=1,\ldots,T$, where $\{\alpha_t \in (0,1)\}_{t=1}^T$ is a predefined noise schedule; see choices in \cite{Ho2020DDPM,nichol2021improved,chan2024tutorial}. 
Denote the distribution of $\bm x_t$ as $p_t(\bm x_t)$.
When $T\rightarrow\infty$, it is shown that $p_T(\bm x_T)$ becomes standard Gaussian under proper noise scheduling \cite{song2021scorebased}.
A neural network $\bm s_{\bm \theta}(\bm x_t,t)$ is trained to predict the noise component from $\bm x_t$: 
\begin{equation}\label{eq:dsm_uncond}
\minimize_{\bm\theta}~
\mathbb{E}_{t,\bm x_0,{\bm\varepsilon_t}}\left[
\left\|
\bm s_{\bm\theta}(\bm x_t,t)
+\frac{1}{\sigma_t}\bm\varepsilon_t
\right\|_2^2
\right],
\end{equation}
where $t$ is sampled uniformly from $\{1,\ldots, T\}$, and $\sigma_t = \sqrt{1-\bar\alpha_t}$ with $\bar\alpha_t=\prod_{i=1}^t \alpha_i$.
It is shown that the learned neural network approximates the score function of noisy data distribution in the following way:
\begin{align}\label{eq:score_approx}
    \bm s_{\bm \theta}(\bm x_t,t)
    \approx
    \nabla_{\bm x_t} \log p_t(\bm x_t).
\end{align}
At inference time, DMs generates $\bm x_0\sim p_0(\bm x) = p(\bm x)$ using a backward process starting from Gaussian noise:
\begin{equation}\label{eq:reverseSDE}
\bm x_{t-1}
=\frac{1}{\sqrt{\alpha_t}}
    \left(
        \bm x_t + \beta_t\bm s_{\bm \theta}(\bm x_t,t)
    \right)
+ \sqrt{\tilde{\beta}_t} \bm z_t,
\end{equation}
where $\beta_t = 1-\alpha_t$, $\tilde\beta_t=\nicefrac{(1-\alpha_t)(1-\bar\alpha_{t-1})}{(1-\bar\alpha_t)}$ and $ \bm z_t \sim {\cal N}(\bm 0,\bm I)$.

\subsubsection{Diffusion for Inverse Problem Solving in RME}
\label{subsec:dps}
Consider applying score-based diffusion methods to RME under the linear inverse model in \eqref{eq:classicIP} with Gaussian noise, i.e., $\bm y=\bm A\bm x+\bm w$. The objective is to sample from the posterior distribution $p(\bm x | \bm y)$. This differs from standard diffusion models, which sample from the unconditional prior $p(\bm x)$. Existing diffusion-based approaches to RME can be broadly categorized into two routes.

The first route trains an E2E \emph{conditional} score model that directly approximates
$\bm s_{\bm\theta}(\bm x_t,t,\bm y)\approx \nabla_{\bm x_t}\log p(\bm x_t|\bm y)$, with the score network conditioned on the observation $\bm y$ via
\begin{equation}\label{eq:dsm_cond}
\minimize_{\bm\theta}~
\mathbb{E}_{t,\bm x_0,\bm\varepsilon_t,\bm y}\left[
\left\|
\bm s_{\bm\theta}(\bm x_t,t,\bm y)
+\frac{1}{\sigma_t}\bm\varepsilon_t
\right\|_2^2
\right],
\end{equation}
Many existing DM-based RME methods follow this paradigm; see, e.g., \cite{luo2025denoising,wang2025radiodiff,wang2024radiodiff}. In these approaches, the incomplete radio map together with auxiliary information (e.g., sensor locations and building or blockage maps) is treated as a conditioning variable, and a conditional score model of $p(\bm x | \bm y)$ is learned. While effective in certain scenarios, this strategy incurs high learning complexity and limited flexibility, as changes in sampling patterns or measurement formats (i.e., changes in $\bm y$) typically require retraining.

The second route makes use of the measurement acquisition models, e.g., $\bm y=\bm A\bm x+\bm w$, and solves the associated inverse problems to recover $\bm x$.
The solution only uses the score of $\bm x_t$ and approximates $\nabla_{\bm x_t} \log p(\bm x_t|\bm y)$ via Bayes’ rule,
\begin{align*} 
\nabla_{\bm x_t} \log p(\bm x_t \mid \bm y)
=
\nabla_{\bm x_t} \log p_{t}(\bm x_t)
+
\nabla_{\bm x_t} \log p(\bm y \mid \bm x_t),
\end{align*}
where the first term is the prior score. The likelihood score in the second term is generally intractable but can be approximated under known signal sensing models like $\bm y=\bm A\bm x+\bm w$ \cite{daras2024survey}. For example \textit{diffusion posterior sampling} (DPS) \cite{chung2023diffusion} approximates it via 
\begin{align} \label{eq:likelihood_score_approx}
\nabla_{\bm x_t} \log  p(\bm y | \bm x_t) \approx \nabla_{\bm x_t}  \log  p(\bm y | \widehat{\bm x}_0(\bm x_t)),
\end{align}
where $\widehat{\bm x}_0(\bm x_t)$ is the estimated noise-free version of $\bm x_t$ derived via Tweedie's formula \cite{efron2011tweedie,chung2023diffusion}:
\begin{align} \label{eq:tweedie}
\widehat{\bm x}_0(\bm x_t) \triangleq {\mathbb E[\bm x_0 |\bm x_t ] = } \frac{1}{\sqrt{\bar\alpha_t}}\left( \bm x_t + (1-\bar\alpha_t)\bm s_{\bm \theta}(\bm x_t, t) \right).
\end{align}
Ideally, $\widehat{\bm x}_0\sim p(\bm x)$ and $p(\bm y | \widehat{\bm x}_0) $ often has a closed form given the distribution of $\bm w$.
Recent RME works such as \cite{gu2025generative} adopt this route, but train diffusion models directly on aggregated radio maps, thereby inheriting similar learning complexity and generalization challenges noted in \cite{Shrestha2022DeepSpectrum,Subash2024Quantized}. Moreover, recoverability guarantees for diffusion-based RME under either route remain largely unexplored.

\section{Proposed Approach}
In this work, our interest lies in pre-trained priors (without involving $\bm y$) to recover $\mathbm X$. Nonetheless, as DM of the aggregate radio map tensors $\mathbm X$ is costly to train, we propose to train diffusion models of the latent factors of $\mathbm X=\sum_{r=1}^R \bm S_r \circ \bm c_r$ in \eqref{eq:X=SC}.  As mentioned in \cite{Subash2024Quantized,Shrestha2022DeepSpectrum}, the hardness of modeling radio maps lies in capturing the characteristics of $\bm S_r$, i.e., the SLF of a single emitter. Learning a DM of $\bm S_r$ is arguably much easier than that of $\mathbm X$, as $\bm S_r$ has a much smaller state space \cite{Shrestha2022DeepSpectrum}.

Under the sensing model in \eqref{eq:classicIP}, our goal is to draw samples from $p(\bm x|\bm y)$ (recall that $\bm x=\vecop(\mathbm X_{\rm spec})$ in \eqref{eq: IP for RM}). 
However, if one only knows the DM of $\bm S_r$ and has no access to either the conditional score $\nabla_{\bm x_t} p(\bm x_t|\bm y)$ or the prior score $\nabla_{\bm x_t} p_{t}(\bm x_t)$, directly using both routes in Sec.~\ref{subsec:dps} for DM-based RME is infeasible. This motivates us to propose a new posterior sampling paradigm.

\begin{figure}
    \centering
    \includegraphics[width=0.385\linewidth]{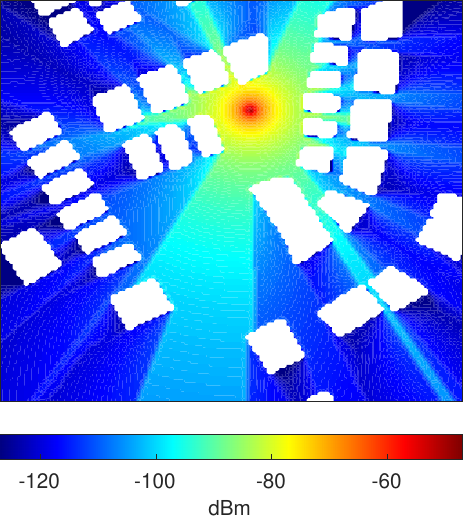} 
    \hspace{4mm}
    \includegraphics[width=0.36\linewidth]{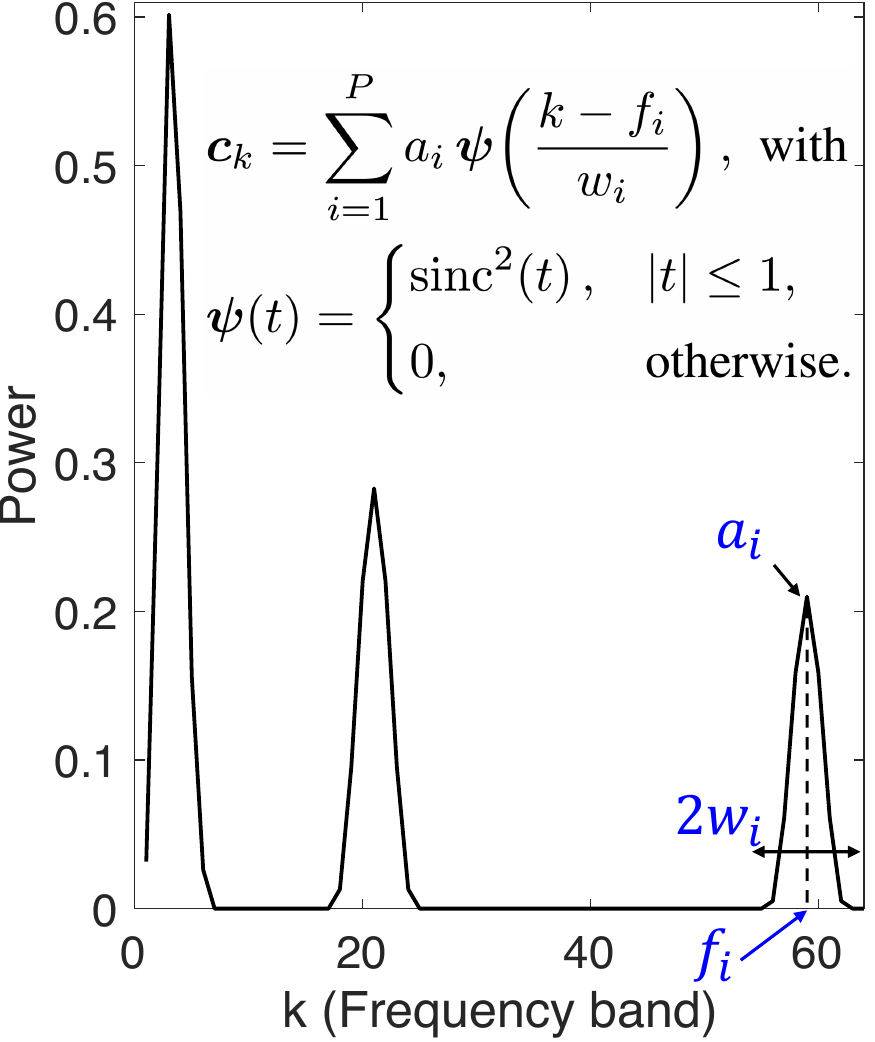}
    \caption{{Typical SLF (left) and PSD (right) of a single emitter over frequency band of interest \cite{Shrestha2022DeepSpectrum,mehanna2013frugal}.}}
    \label{fig:PSD}
\end{figure}

\subsection{Proposed Algorithm: FactorDPS} \label{subsec:FactorDPS}
We note that the hard to model part in radio maps is the SLF \cite{Shrestha2022DeepSpectrum,Subash2024Quantized,xu2025radio}. Hence, we train a DM for $\bm S_r$.
In addition, to facilitate a simple algorithm,
we also train a separate DM of the single-emitter PSD, i.e., $\bm c_r \in \mathbb{R}^K$. Note that both the SLF and PSD have mature numerical simulators.
The SLF can be generated by the statistical simulator in \cite{goldsmith_2005} and ray-tracing based simulators such as WinProp \cite{Hoppe2017wave,DatasetPaper} and SIONNA \cite{hoydis2023sionna}. Emitter PSDs can be well simulated by truncated squared-sinc signals as in \cite{mehanna2013frugal,fu2015factor,fu2016power}; also see Fig.~\ref{fig:PSD}.

\subsubsection{Spatial and Spectral DM Training} 
\label{subsec:LatentTraining}
Let $\bm s \in \mathbb{R}^{IJ}$ denote a vectorized single-emitter SLF sample
(i.e., $\bm s=\mathrm{vec}(\bm S_r)$ for any $r$). It is reasonable to assume that all emitters share the same generative model for their SLFs; see \cite{Shrestha2022DeepSpectrum,Subash2024Quantized,timilsina2025domain}. We train a score network
$f_{\mathrm{SLF}}(\bm s_t,t; \bm \theta_{\rm SLF})\approx \nabla_{\bm s_t}\log p_t(\bm s_t)$
using denoising score matching:
\begin{equation}\label{eq:train_slf}
\min_{\bm\theta_{\mathrm{SLF}}}\;
\mathbb{E}_{t,\bm s_0,\bm\varepsilon_t}\!\left[
\left\|
f_{\mathrm{SLF}}\!\left(\bm s_t,t;\bm\theta_{\mathrm{SLF}}\right)
+\frac{1}{\sigma_t}\bm\varepsilon_t
\right\|_2^2
\right].
\end{equation}
Here, $\bm s_t$ is generated from the clean sample $\bm s_0$ via the forward diffusion process $\bm s_t = \sqrt{\bar{\alpha}_t} \bm s_0 + \sqrt{1-\bar{\alpha}_t} \bm \varepsilon_t$, where $\bm \varepsilon_t \sim \mathcal N(\bm 0, \bm I_{IJ})$. The training samples $\bm s_0$ can be generated using statistical \cite{goldsmith_2005} or ray-tracing based simulators \cite{Hoppe2017wave,hoydis2023sionna}, depending on the wireless propagation scenario under consideration. In this work, we will use the ray-tracing based dataset in \cite{DatasetPaper}; see details in Sec. \ref{subsec:training}. Similarly, let $\bm c \in \mathbb{R}^{K}$ denote a PSD (spectral signature) sample. We train a score network
$f_{\mathrm{PSD}}(\bm c_t,t)\approx \nabla_{\bm c_t}\log p_t(\bm c_t)$
as follows:
\begin{equation}\label{eq:train_psd}
\min_{\bm\theta_{\mathrm{PSD}}}\;
\mathbb{E}_{t,\bm c_0,\bm\varepsilon_t}\!\left[
\left\|
f_{\mathrm{PSD}}\!\left(\bm c_t,t;\bm\theta_{\mathrm{PSD}}\right)
+\frac{1}{\sigma_t}\bm\varepsilon_t
\right\|_2^2
\right],
\end{equation}
where $\bm c_t=\sqrt{\bar\alpha_t}\bm c_0+\sqrt{1-\bar{\alpha}_t}{\bm\varepsilon_t},\;{
\bm\varepsilon_t}\sim\mathcal N(\bm 0,\bm I_K)$ and $\bm c_0$ is generated by the PSD simulators in \cite{mehanna2013frugal,fu2015factor,fu2016power}.

Consequently, after the training stage, we have two DMs $\bm f_{\rm SLF}(\cdot,t;\bm \theta_{\rm SLF})$ and $\bm f_{\rm PSD}(\cdot,t;\bm \theta_{\rm PSD})$ to represent the scores $\nabla \log p(\bm S_r)$ and $\nabla \log p(\bm c_r)$, respectively.

\subsubsection{FactorDPS}
Under the signal model, our goal boils down to sampling from the conditional joint distribution of the SLFs and the PSDs, i.e., \[ \{\bm S_r, \bm c_r\}_{r=1}^R \sim p(\{\bm S_r, \bm c_r\}_{r=1}^R|\bm{y}),\]
as recovering $\{\bm S_r, \bm c_r\}_{r=1}^R$ would recover $\mathbm X$. To proceed, we assume the following:

\begin{assumption}[Latent Factor Independence]\label{as:indep}
For any radio map $\mathbm X=\sum_{r=1}^R \bm S_r \circ \bm c_r$, the latent factors $\{ \bm S_r, \bm c_r \}_{r=1}^R$ are mutually statistically independent, i.e.,
    \begin{align}
        p(\{\bm S_r, \bm c_r\}_{r=1}^R)=\prod_{r=1}^R p(\bm S_r) p(\bm c_r).
    \end{align}
\end{assumption}
Assumption~\ref{as:indep} is a reasonable approximation for reality. First, the spatial occurrence and spectral occupancy of an emitter are often not related to each other, as the latter is governed by purely spectral considerations, e.g., frequency license.
Second, as long as signal strength propagation and spectral occupancy of different emitters are not coordinated, cross-emitter independence also makes sense.

Under Assumption~\ref{as:indep}, Bayes' rule leads to the following factored posterior:
\begin{align}\label{eq:indepstep}
p(\{\bm S_r,\bm c_r\}_{r=1}^R | \bm y)
\propto
p(\bm y | \{\bm S_r,\bm c_r\}_{r=1}^R)
\prod_{r=1}^R  p(\bm S_r)p(\bm c_r).
\end{align}

The above allows us to apply the reverse process in \eqref{eq:reverseSDE} to generate the latent factors $$\bm S\triangleq(\bm S_1,\ldots,\bm S_R),~~\bm C\triangleq(\bm c_1,\ldots,\bm c_R)$$ conditioned on $\bm y$.
To be specific, in each sampling step $t$, the following holds:
\begin{align} \label{eq:pysc}
    & \nabla_{\bm S^{(t)},\bm C^{(t)}} \log p(\bm S^{(t)},\bm C^{(t)}  | {\bm y})  
    \nonumber \\
    &=  \nabla_{\bm{S}^{(t)},\bm C^{(t)}} \log p(\bm S^{(t)},\bm C^{(t)} ) + {\nabla_{\bm{S}^{(t)},\bm C^{(t)}} \log {p( \bm y|\bm S^{(t)},\bm C^{(t)} )}} \nonumber\\
      &= \operatorname{cat}\left(\left\{\nabla_{\bm S_r^{(t)}}\log p_t(\bm S_r^{(t)})\right\}_{r=1}^R,\left\{\nabla_{\bm c_r^{(t)}}\log p_t(\bm c_r^{(t)})\right\}_{r=1}^R\right) \nonumber \\
      & \qquad \quad   +  \nabla_{\bm{S}^{(t)},\bm C^{(t)}} \log {p( \bm y|\bm S^{(t)},\bm C^{(t)} )} \nonumber \\
      & \approx \operatorname{cat}\left(\left\{\bm f_{\rm SLF}(\bm S_r^{(t)},t)\right\}_{r=1}^R,\left\{\bm f_{\rm PSD}(\bm c_r^{(t)},t)\right\}_{r=1}^R\right) \nonumber \\
      & \qquad \quad +  \nabla_{\bm{S}^{(t)},\bm C^{(t)}} \log {p( \bm y|\bm S^{(t)},\bm C^{(t)} )},
\end{align} 
where we have used \eqref{eq:indepstep} in the second equality. The operator $\operatorname{cat}(\cdot)$ denotes concatenation of its arguments after vectorization. 
In this decomposition, the only term that requires approximation is the likelihood score $\nabla_{\bm S^{t}, \bm C^{(t)}} \log p(\bm y | \bm S^{(t)}, \bm C^{(t)})$, since it involves the generally intractable conditional distribution relating the measurement process. This term can be approximated using off the-shelf intermediate likelihood-score approximators from the literature \cite{daras2024survey}. 
In particular, recall that DPS \cite{chung2023diffusion} approximates the likelihood term using \eqref{eq:likelihood_score_approx}.
Applying the same trick, we have
\begin{align} \label{eq:S-score-conditional}
    &{\nabla_{\bm S^{(t)}, \bm C^{(t)}} \log p(\bm y | \bm S^{(t)},\bm C^{(t)})} \\
    &\approx  {\nabla_{\bm S^{(t)}, \bm C^{(t)}} \log p(\bm y | \widehat{\bm S}^{(0)},\widehat{\bm C}^{(0)})},  \nonumber
\end{align}
where $\widehat{\bm S}^{(0)}$ and $\widehat{\bm C}^{(0)}$ are the Tweedie-denoised versions of $\bm S^{(t)}$ and $\bm C^{(t)}$, respectively; see \eqref{eq:tweedie}.

Further note that when the measurement noise $\bm w$ is white Gaussian, the log-likelihood $\log p(\bm y |{\bm S}, \bm C)$ is given by:
\begin{align} \label{eq:likelihood_noQuan}
- \frac{1}{2 \sigma_n^2}
\left\|
\bm y -\bm A
{\rm vec}\left(\sum_{r=1}^R \bm S_r \circ \bm c_r\right)
\right\|_F^2
+ {\rm const},
\end{align}
where $\sigma_n^2 = \var (\bm w)$---making the reverse process in \eqref{eq:reverseSDE} fully executable. 

\begin{algorithm}[t]
\footnotesize
\algsetup{indent=0.8em}
\caption{Factor Diffusion Posterior Sampling (FactorDPS)}
\label{alg:lfdps}
\begin{algorithmic}[1]

\REQUIRE Observation $\bm y$, diffusion steps $T$, initial latents $\{\bm S_r^{(T)},\bm c_r^{(T)}\}_{r=1}^R$
\ENSURE Estimated latent factors $\{\bm S_r^{(0)},\bm c_r^{(0)}\}_{r=1}^R$

\FOR{$t = T, \ldots, 1$}

\vspace{0.5ex}
\algocomment{Tweedie Lemma Denoising}
\FOR{$r=1$ \TO $R$}
\vspace{0.5ex}
\STATE $\widehat{\bm S}_r^{(0)} \gets \dfrac{1}{\sqrt{\bar\alpha_t}}
\left(\bm S_r^{(t)} + (1-\bar\alpha_t)\, \bm f_{\rm SLF}(\bm S_r^{(t)},t)\right)$
\STATE $\widehat{\bm c}_r^{(0)} \gets \dfrac{1}{\sqrt{\bar\alpha_t}}
\left(\bm c_r^{(t)} + (1-\bar\alpha_t)\, \bm f_{\rm PSD}(\bm c_r^{(t)},t)\right)$
\ENDFOR

\vspace{0.5ex}
\algocomment{Conditional Score Construction}
\FOR{$r=1$ \TO $R$}
\STATE $\bm f_{\rm SLF}(\bm S_r^{(t)},t|\bm y) \gets 
\bm f_{\rm SLF}(\bm S_r^{(t)},t)
+ \nabla_{\bm S_r^{(t)}}
\log p\!\left(\bm y |\widehat{\bm S}^{(0)},\widehat{\bm C}^{(0)}\right)$
\STATE $\bm f_{\rm PSD}(\bm c_r^{(t)},t|\bm y) \gets 
\bm f_{\rm PSD}(\bm c_r^{(t)},t)
+\nabla_{\bm c_r^{(t)}}
\log p\!\left(\bm y | \widehat{\bm S}^{(0)},\widehat{\bm C}^{(0)} \right)$
\ENDFOR

\vspace{0.5ex}
\algocomment{Reverse Update}
\FOR{$r=1$ \TO $R$}
\STATE Sample $\bm z_{S_r}\sim\mathcal N(\bm 0,\bm I_{IJ})$,
$\bm z_{c_r}\sim\mathcal N(\bm 0,\bm I_K)$
\STATE $\bm S_r^{(t-1)} \gets
\dfrac{1}{\sqrt{\alpha_t}}
\left(\bm S_r^{(t)} + \beta_t \bm f_{\rm SLF}(\bm S_r^{(t)},t\mid\bm y)\right)
+ \sqrt{\tilde\beta_t}\,\bm z_{S_r}$
\STATE $\bm c_r^{(t-1)} \gets
\dfrac{1}{\sqrt{\alpha_t}}
\left(\bm c_r^{(t)} + \beta_t \bm f_{\rm PSD}(\bm c_r^{(t)},t\mid\bm y)\right)
+ \sqrt{\tilde\beta_t}\,\bm z_{c_r}$
\ENDFOR

\ENDFOR

{\bf Output}:
$\{\bm S_r^{(0)}, \bm c_r^{(0)}\}_{r=1}^R$ and $\widehat{\mathbm X} = \sum_r \bm S_r^{(0)}\circ\bm c_r^{(0)}$.

\end{algorithmic}
\end{algorithm}

Algorithm \ref{alg:lfdps} summarizes our proposed method. Since this method uses a decomposed signal model and pre-trained DMs in the latent space, we refer to it as \textit{Latent Factor Diffusion Posterior Sampling} (\texttt{FactorDPS}).
The proposed approach enjoys a lighter training load compared to those training a DM for $p(\bm x)$, e.g., \cite{gu2025generative}, and also avoids the potential pitfalls in directly training a DM for $p(\bm x| \bm y)$ as in \cite{luo2025denoising,wang2025radiodiff,wang2024radiodiff}, as will be shown in the experiments.

\subsection{Extensions}
The proposed method is flexible to handle various scenarios (e.g., single-band cases), incorporate auxiliary information (e.g., building maps), and deal with non-continuous measurements.

\subsubsection{Single-band case}\label{subsec:single-band-case}:
Although the algorithm is designed for spatio-spectral radio map tensors,
it can also be used for single-band radio maps as in \cite{romero2022radio,bi2019engineering}.
For single-band cases, the model becomes
\begin{align}\label{eq:multi-R,single-K}
 \bm X =\sum_{r=1}^R \bm S_r. 
\end{align}
The \texttt{FactorDPS} algorithm applies with $\bm f_{\rm SLF}(\cdot,t)$ without updating $\bm c_r$.

\subsubsection{Quantized Measurements} Assume that the measurements are quantized using a Gaussian dithering based quantizer at the sensors for saving signaling bandwidth \cite{Subash2024Quantized,timilsina2025domain}. That is, we have
\begin{align} \label{eq:quantize_Y}
    \bm y = {\cal Q} \left( \bm A \bm x +{ \bm w} \right), 
\end{align}
where $\bm y \in [Q]^{NK}$ is the quantized radio map and $\mathcal{Q}(\cdot): \mathbb{R} \to [Q]$ denotes a quantization operator. Specifically, for quantization boundaries $b_0 < b_1 < \cdots < b_Q$,
\begin{equation}
    \mathcal{Q}(x) = q \quad \text{if } x \in (b_{q-1}, b_q].
\end{equation}
The probability of observing level $q$ is the integral of the Gaussian noise density over this bin. The log-likelihood is:
\begin{align} \label{eq:likelihood_Quantized}
\log\; p(\bm y|\bm S,\bm C)
&= \sum_{(i,j)\in\bm\varOmega}
\sum_{k=1}^K  \sum_{q=1}^Q \mathbb 1({\mathbm Y(i,j,k)=q})  \nonumber\\
&
\log
\Big( \Phi_{\sigma_n}\big(b_q-\sum_{r=1}^R \bm S_r(i,j)\bm c_r(k)\big) \nonumber
 \\
&-
\Phi_{\sigma_n}\big(b_{q-1}-\sum_{r=1}^R \bm S_r(i,j)\bm c_r(k)\big)
\Big),   
\end{align}
in which $\Phi_{\sigma_n}(\cdot)$ denotes the CDF of a zero-mean Gaussian with variance $\sigma_n^2$.
With this expression, the \texttt{FactorDPS} algorithm can handle quantized measurements by replacing \eqref{eq:likelihood_noQuan} with \eqref{eq:likelihood_Quantized} in the updates.

\subsubsection{SLF Score Conditioned on Building Maps} \label{subsec:condition}
The single-emitter SLF's DM also allows one to incorporate auxiliary information such as buidling maps in a relatively lightweight way (as the signal interacts with building maps individually). That is, instead of training $\bm f_{\rm SLF}(\bm S_{r}^{(t)},t)\approx \nabla_{\bm S_r^{(t)}} \log p_{t}(\bm S_r^{(t)})$, we use conditional score \cite{song2021scorebased} to represent 
\begin{align}\label{eq:conditionalscore}
    \bm f_{\rm SLF, aux}(\bm S_r^{(t)}, t, \bm m) \approx \nabla_{\bm S_r^{(t)}} \log p_{t}\left( \bm S_{r}^{(t)} | \bm m  \right).
\end{align}
The training of such scores can be done by inputting $\bm m$ into the score based denoiser; see \cite{tashiro2021csdi} and the illustration in Fig.~\ref{fig:slf_conditional}.
 
An important remark is that including $\bm m$ almost does not change the algorithm.
To be specific, by the same spirit of Assumption~\ref{as:indep}, it is reasonable to assume $p(\{\bm S_r,\bm c_r\}_{r=1}^R|\bm m)
=\prod_{r=1}^R p(\bm S_r|\bm m)p(\bm c_r)$, where we have used an additional condition that $\bm c$ is independent from the building map $\bm m$.
As a consequence, during inference over a given city map $\bm m$, one only needs to replace $f_{\rm SLF}(\bm S_r^{(t)},t)$ in Algorithm~\ref{alg:lfdps} by $f_{\rm SLF,aux}(\bm S_r^{(t)},t,\bm m)$, i.e., the learned score conditioned on $\bm m$.  We use such city map-conditioned training and inference methods in our experiments; see Sec.~\ref{sec:exp}.

\begin{figure}
    \centering
    \includegraphics[width=0.85\linewidth]{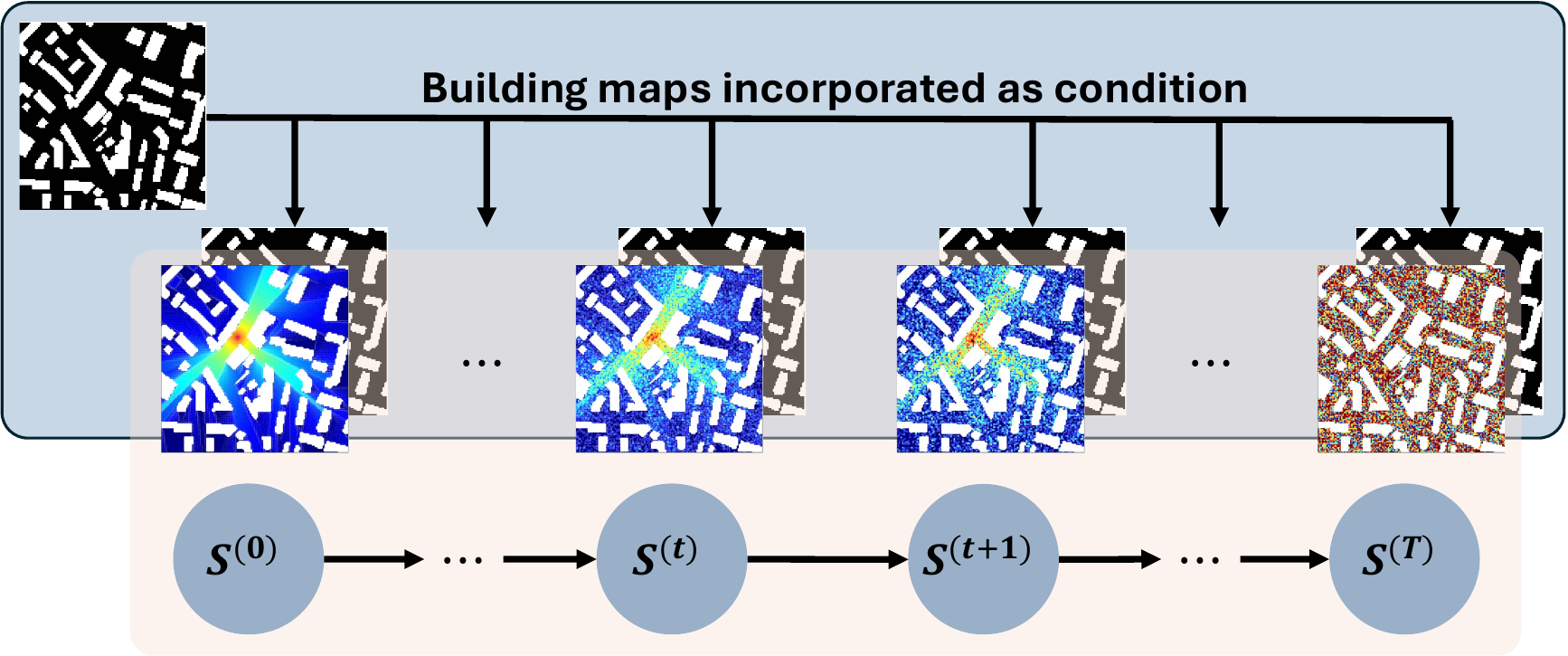}
    \caption{Conditional score model for incorporating auxiliary information, where the building map is adopted as the network input.}
    \label{fig:slf_conditional}
\end{figure}

\section{Recoverability of Posterior Sampling} \label{sec:recoverability}
In this section, we show that posterior sampling, i.e., $\{\widehat{\bm S}_r,\widehat{\bm c}_r\}\sim p(\bm S,\bm C|\bm y)$, recovers the ground-truth radio map, i.e., $\mathbm{X}=\sum_{r=1}^R \bm S_r\circ\bm c_r$, up to bounded errors. 
We also show that this process is robust: even when the learned scores and approximated likelihood present a certain gap from the ground-truth $p(\bm S,\bm C|\bm y)$, the estimated radio map is still reasonably accurate. For simplicity of notation, we present the results for unconditional score models $\bm f_{\rm SLF}(\bm S_r^{(t)},t)$.
Nonetheless, extension to settings with building maps and conditional scores $f_{\rm SLF,aux}(\bm S_r^{(t)},t,\bm m)$ in Sec.~\ref{subsec:condition} is within reach.

To move forward, we impose the following assumptions on the latent factors:
\begin{assumption} \label{assump:structured_latent}
Both the ground-truth and estimated radio maps satisfy the following:
\begin{align}
\|\bm S_r\|_{\rm F} \le \gamma_{\rm s},
\ 
\|\bm c_r\|_2 \le \gamma_{\rm c},
\qquad
\forall r\in[R],
\label{eq:c_bound}
\end{align}
where $\gamma_s$ and $\gamma_c$ are positive constants.
\end{assumption}

Let $\bm z\in\mathbb{R}^{(IJ+K)R}$ be the parameter vector collecting all the unknown parameters in our model. That is, we have $\bm z \triangleq {\rm vec}(\bm Z) \in \mathbb R^{(IJ+K)R}$ where
\begin{align}\label{eq:z}
\bm Z \triangleq
\begin{bmatrix}
\mathrm{vec}(\bm S_1) & \cdots & \mathrm{vec}(\bm S_R)\\
\bm c_1 & \cdots & \bm c_R
\end{bmatrix}
\in \mathbb R^{(IJ+K)\times R}.
\end{align}
In addition, we define a map $\bm h(\cdot)$ applied onto $\bm z$, i.e.,
\begin{align}\label{eq:h(z)}
    \bm x = \bm h(\bm z)
    \triangleq
    \mathrm{vec}\!\Big(\sum_{r=1}^R \bm S_r \circ \bm c_r\Big)
    \in \mathbb R^{IJK}.
\end{align}
Let $p_{\bm z}$ denote the probability density function of $\bm z$, and $p_{\bm x}$ denote the distribution induced by $\bm x=\bm h(\bm z)$.
The support of $p_{\bm x}$ is therefore \textit{included} in the following set:
\begin{align}\label{eq:Xset}
\mathbb X
\triangleq
\{\bm x \in &\mathbb R^{IJK} \mid \bm x = \bm h(\bm z),\ \bm z ={\rm vec}(\bm Z), \nonumber \\
&\|\bm S_r\|_{\rm F} \le \gamma_{\rm s},\ 
\|\bm c_r\|_2 \le \gamma_{\rm c},\
\forall r\in[R]\}.
\end{align}

Our analysis uses the notion of $\epsilon$-net and its covering number \cite{shalev2014understanding,mohri2018foundations}, whose definitions are as follows:
\begin{Def}[$\epsilon$-Net and Covering Number]
    Let $\mathbb A \subseteq \mathbb R^d$ and $\epsilon > 0$. A finite set $\mathbb M_{\epsilon} \triangleq \{ \underline{\bm a}_1,\ldots,\underline{\bm a}_M \} \subseteq \mathbb R^d$ is called an $\epsilon$-net of $\mathbb A$ with respect to the $\ell_2$ norm if, for every $\bm a \in \mathbb A$, there exists some $\underline{\bm a}_i \in \mathbb M_{\epsilon}$ such that
    \begin{align}
        \|\underline{\bm a}_i - \bm a \|_2 \leq \epsilon. \nonumber
    \end{align}
    The covering number of $\mathbb A$ is then defined as 
    \begin{align}
        \mathrm{N}(\mathbb A,\epsilon) \triangleq \min \{ |\mathbb M_{\epsilon}| : \mathbb M_{\epsilon} \text{ is an } \epsilon\text{-net of } \mathbb A \}, \nonumber
    \end{align}
    where $|\mathbb M_\epsilon|$ denotes the cardinality of $\mathbb M_\epsilon$.
\end{Def}
Under the above definitions,
the set in \eqref{eq:Xset} satisfies the following characterization:
\begin{Fact}\label{fact:covering_X} Let $\mathbb X$ be defined as in \eqref{eq:Xset}. Then, every $\bm x \in \mathbb X$ satisfies
\[
\max_{\ell} |\bm x_\ell|
\le
{\tau_{\rm x}},
\quad
\tau_{\rm x} \triangleq R\gamma_{\rm s} \gamma_{\rm c}.
\]
Moreover, the covering number of $\epsilon$-net of $\mathbb X$ satisfies
\begin{align}
{\rm N}(\mathbb X, \epsilon)
\le
\left(\frac{3R(\gamma_{\rm s}+\gamma_{\rm c})}{\epsilon}\right)^{(IJ+K)R}
\gamma_{\rm s}^{IJR}
\gamma_{\rm c}^{KR}.
\label{eq:cover_X}
\end{align}
\end{Fact}

\begin{IEEEproof}
The bound follows from a standard covering-number derivation for matrix decomposition/completion models \cite{Shrestha2022DeepSpectrum,fornasier2021robust,xu2025radio,wang2012stability}; see a derivation in the supplementary materials for completeness.
\end{IEEEproof}

\subsection{Perfect Sampler Case}
In this subsection, we first characterize the performance of \texttt{FactorDPS} under an ideal assumption that the algorithm produces samples exactly from $p(\bm S,\bm C|\bm y)$, which is the exact objective of the posterior sampling process.
We show the following theorem:

\begin{theorem}
\label{thm:perfect_case}
Denote the ground-truth latent parameters and radio map as $\bm z^\natural \sim p_{\bm z}$ and $\bm x^\natural = \bm h(\bm z^\natural)$, respectively. 
Under the data acquisition model in \eqref{eq:classicIP}-\eqref{eq:Pomega}, assume that the sensor locations are uniformly sampled from the $I\times J$ grids, and that \texttt{FactorDPS} produces $\hat{\bm z}\sim p(\bm z\mid \bm y)$ and $\hat{\bm x}=\bm h(\hat{\bm z}) \in \mathbb{X} $. 
For any $\epsilon\in(0,\nicefrac{4\tau_x}{3})$
and $\delta\in(0,1)$, assume that
\begin{align}\label{eq:N_bound}
    &N \geq { C(\tau_{\rm x},\epsilon)} \bigg(  (IJ+K)R\log(\frac{12R(\gamma_{\rm s}+\gamma_{\rm c})}{\sqrt{IJK}\epsilon}) \nonumber\\
    &\qquad \quad + IJR\log\gamma_{\rm s} + KR \log \gamma_{\rm c} + \log \frac{8}{\delta}  \bigg), 
\end{align}
where
\begin{align}\label{eq:C_epsilon}
{C(\tau_{\rm x},\epsilon)} \triangleq \max\left\{ \frac{512\tau_{\rm x}^4}{\epsilon^4},\frac{32\sigma^2(\epsilon^2+8\sigma^2+{64\tau_{\rm x}^2})}{K\epsilon^4} \right\}.   
\end{align}
Then, with probability at least $ 1-\delta$, the following holds
\[
\frac{\|\bm x^\natural - \hat{\bm x}\|_2}{\sqrt{IJK}} < \epsilon.
\]
\end{theorem}

\begin{IEEEproof}
    See the supplementary materials.
\end{IEEEproof}
Theorem~\ref{thm:perfect_case} shows that posterior sampling under the radio map's data structure in \eqref{eq:X=SC} and the data acquisition process in \eqref{eq:classicIP}-\eqref{eq:Pomega} is {\it instance-optimal} \cite{jalal2021instance,jalal2021robust}. That is, any solution $\widehat{\bm x}=\bm h(\widehat{\bm z})$ resulted from a posterior draw $\widehat{\bm z}\sim p(\bm z|\bm y)$ falls in the vicinity of $\bm x^\natural$. 

A remark is that the sample complexity bound for reaching $\epsilon$-accuracy in \eqref{eq:N_bound}
scales with the effective dimension of the radio map (i.e., $O(IJR+KR)$ when $R\leq \sqrt{IJK}$){, consistent} with classical sample complexity bounds of inverse problems. 

\subsection{Imperfect Sampler Case}

Theorem~\ref{thm:perfect_case}'s premise is that one can sample exactly from the posterior $p(\bm z|\bm y)$.
In practice, this condition is hard to meet due to various reasons, e.g., that the learned scores $\bm f_{\bm SLF}$ and $\bm f_{\rm PSD}$ are only approximations for the true scores and that the DPS trick [cf. Eq.~\eqref{eq:S-score-conditional}] does not exactly recover the likelihoods in \eqref{eq:likelihood_score_approx}.

In this subsection, we consider a more realistic setting where our \texttt{FactorDPS} does not exactly sample from $p(\bm z|\bm y)$ but only approximately does so. More specifically, we assume that \texttt{FactorDPS}'s draws end up in a proxy distribution $p'(\bm z|\bm y)\neq p(\bm z|\bm y)$.
To formalize this notion, we first quantify the mismatch between the ground-truth and empirical distributions using the \textit{total variation} (TV) \cite{takezawa2005introduction}:
\begin{Def}[TV-distance]\label{def:tv-distance}
    For two probability distributions $p'$ and $p$ defined on the same space, the TV distance is defined as:
    \begin{align}\label{eq:TV}
        \mathrm{TV}(p',p)
        \triangleq
        \sup_{\mathbb A} |P'(\mathbb A)-P(\mathbb A)|,
    \end{align}
    where $\mathbb A$ is any measurable set within the domain of $p$ and $p'$.
\end{Def}
Under this definition, one can see that
there always exists $\varepsilon_{\rm TV}\in[0,1]$ such that for every measurement
$\bm y\in\mathbb R^{NK}$,
\begin{equation}
\mathrm{TV}\!\Big(p'(\bm z\mid \bm y),\,p(\bm z\mid \bm y)\Big)
\le \varepsilon_{\rm TV},
\label{eq:post_TV_mismatch_final}
\end{equation}
where $p(\bm z\mid \bm y)$ is the ground-truth posterior induced by $p_{\bm z}$
and the likelihood of $\bm y=\bm A\bm h(\bm z)+\bm\xi$, while $p'(\bm z\mid \bm y)$ denotes
the stationary distribution of the \texttt{FactorDPS} process.

Under such imperfect posterior sampling, the recovered estimate $\hat{\bm x}$ may not necessarily lie in the ground truth radio map set $\mathbb X$. To account for this mismatch, we introduce the following {\it extended} support:
\begin{align} \label{eq:G(X)}
    \mathbb G(\mathbb X,{\kappa}) \triangleq \left\{ \hat{\bm x}\in\mathbb R^{IJK}:\; 
    \min_{\bm x\in\mathbb X}\|\hat{\bm x}-\bm x\|_2\le {\kappa}
    \right\}.
\end{align}
It is intuitive to assume that $\widehat{\bm x}$ ends up within $\mathbb G(\mathbb X,{\kappa})$ with an acceptable ${\kappa}$ if $\varepsilon_{\rm TV}$ is not large.
Using this notion, we show the following:

\begin{theorem}
\label{thm:imperfect_TV}
Under the settings of Theorem~\ref{thm:perfect_case}, let $\hat{\bm z}\sim p'(\bm z\mid \bm y)$ be a draw from 
\texttt{FactorDPS}. Assume $\hat{\bm x} = \bm h(\hat{\bm z})\in\mathbb G(\mathbb X,{\kappa})$ for some ${\kappa} \geq 0$.
Then for any $\epsilon\in(0,\nicefrac{4(\tau_x+\kappa)}{3})$ and $\delta \in (0,1)$, if
\begin{align}\label{eq:N_bound_imperf}
    &N \geq {C({ \tau_{\rm x}+\kappa},\epsilon)} \bigg(  (IJ+K)R\log(\frac{12R(\gamma_{\rm s}+\gamma_{\rm c})}{\sqrt{IJK}\epsilon}) \nonumber\\
    &\qquad \quad + IJR\log\gamma_{\rm s} + KR \log \gamma_{\rm c} + \log \frac{8}{\delta}  \bigg), 
\end{align}
where { $C(\cdot,\cdot)$} is defined in \eqref{eq:C_epsilon},
then
\[
\frac{\|\bm x^\natural - \hat{\bm x}\|_2}{\sqrt{IJK}} < \epsilon
\]
with probability at least $ 1-\varepsilon_{\rm TV}-\delta$.
\end{theorem}

\begin{IEEEproof}
    See the supplementary materials.
\end{IEEEproof}

Theorem~\ref{thm:imperfect_TV} asserts that, even if the posterior sampling process is imperfect, as long as the stationary distribution of the \texttt{FactorDPS} approach is a reasonable approximation of $p(\bm z|\bm y)$ (in terms of bounded TV and the solution support), the sampled $\widehat{\bm x}$ is still within a neighborhood of the ground truth, serving as an acceptable estimate.

\subsection{Sharpening Sample Complexity}\label{subsec:sample_complexity}

In this section, we will show that the sample complexity can be reduced even further. By adopting widely assumed structures on the spatial factors, one can effectively restrict the feasible set and thereby sharpen the theoretical error bound.

\begin{Corollary}
\label{cor:cover_lowrank}
In addition to Assumption~\ref{assump:structured_latent}, suppose that each spatial factor
$\bm S_r\in\mathbb R^{I\times J}$ satisfies ${\rm rank}(\bm S_r)\leq r_{\rm s}$. 
Define the corresponding low-rank radio map set as
\begin{align}
    &\mathbb X_{\rm lr} \triangleq \{\bm x \in \mathbb X \mid  \mathrm{rank}(\bm S_r)\le r_{\rm s},\ 
    \forall r\in[R]\}.
\end{align}
Let $\hat{\bm z}$ be drawn from \texttt{FactorDPS}, and assume that $\hat{\bm x}=\bm h(\hat{\bm z})\in\mathbb G(\mathbb X_{\rm lr},{\kappa})$ for some ${ \kappa} \geq 0$.
Then, for any $\epsilon\in(0,\nicefrac{4(\tau_x+\kappa)}{3})$, the guarantee in Theorem \ref{thm:imperfect_TV} continues to hold provided that
\begin{align}
    &N \geq C({ {\tau_{\rm x}+\kappa}},\epsilon) \Big(  R\big((I + J+1)r_{\rm s}+K\big)\log(\frac{12R(\gamma_{\rm s}+\gamma_{\rm c})}{\sqrt{IJK} \epsilon}) \nonumber\\
    & \quad + Rr_{\rm s}(I+J+{ 1})\log(3\gamma_{\rm s}) + RK \log \gamma_{\rm c} + \log \frac{8}{\delta}  \Big),
\end{align}
where $C({\cdot,\cdot})$ is defined in \eqref{eq:C_epsilon}.
\end{Corollary}
\begin{IEEEproof}
    See the supplementary materials.
\end{IEEEproof}

Corollary~\ref{cor:cover_lowrank} indicates that, if we further restrict $\mathbb X$ in \eqref{eq:Xset} to radio maps with low-rank SLFs, then accurate recovery of $\bm x^\natural$ can be achieved with substantially reduced sample complexity compared with Theorems~\ref{thm:perfect_case} and \ref{thm:imperfect_TV}.

\begin{Remark}
A remark is that the low-rank assumption is often adopted in radio map sensing frameworks due to its simplicity \cite{Zhang2020SpectrumViaBlockTerm, sun2024integrated, sun2022propagation}; also see Fig.~\ref{fig:slf_lowrank} for empirical evidence.
However, we emphasize that the low-rank assumption alone has been shown not sufficient to capture the complex nature of SLFs \cite{Shrestha2022DeepSpectrum,Subash2024Quantized,timilsina2025domain}.
Here, our computational framework uses the learned DM of the SLFs, which likely reflects both low rank and other more complex structures.
Therefore, the actual radio map's support set is expected to be a much smaller subset of $ \mathbb X_{\rm lr}$. That is, Corollary~\ref{cor:cover_lowrank} still presents a pessimistic bound. Nonetheless, the bound is much tighter than those in Theorems~\ref{thm:perfect_case}-\ref{thm:imperfect_TV}. Overall, these bounds are intended to provide theoretical support for the effectiveness of \texttt{FactorDPS}, rather than to precisely characterize the number of samples required in practice.
\end{Remark}

\begin{figure}[!t]
    \centering
    \includegraphics[width=0.45\linewidth]{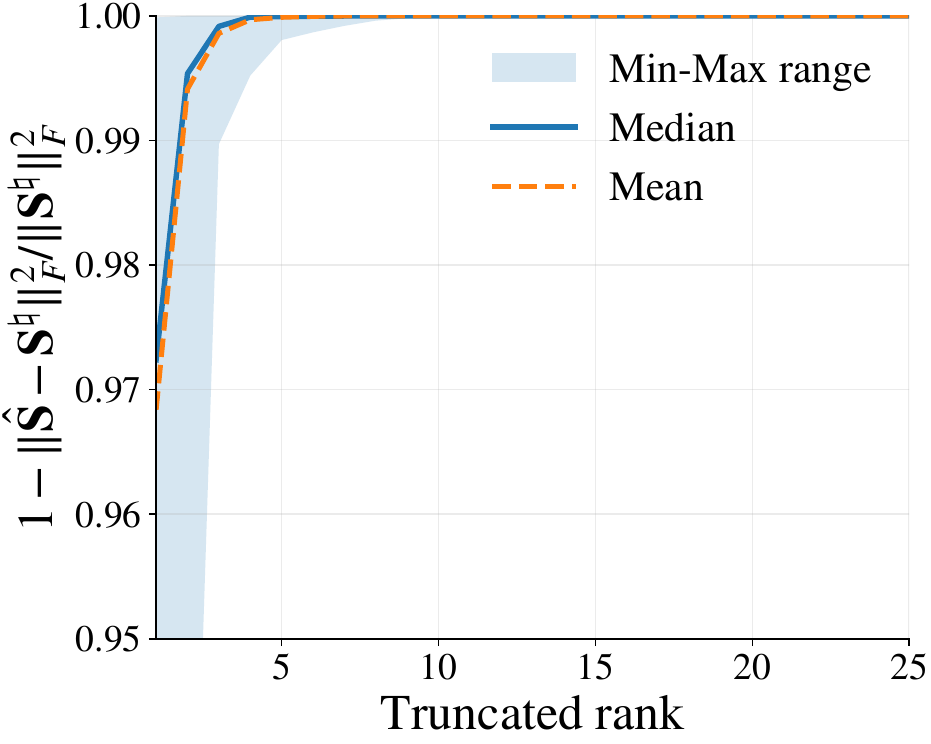}
    \hspace{3mm}
    \includegraphics[width=0.45\linewidth]{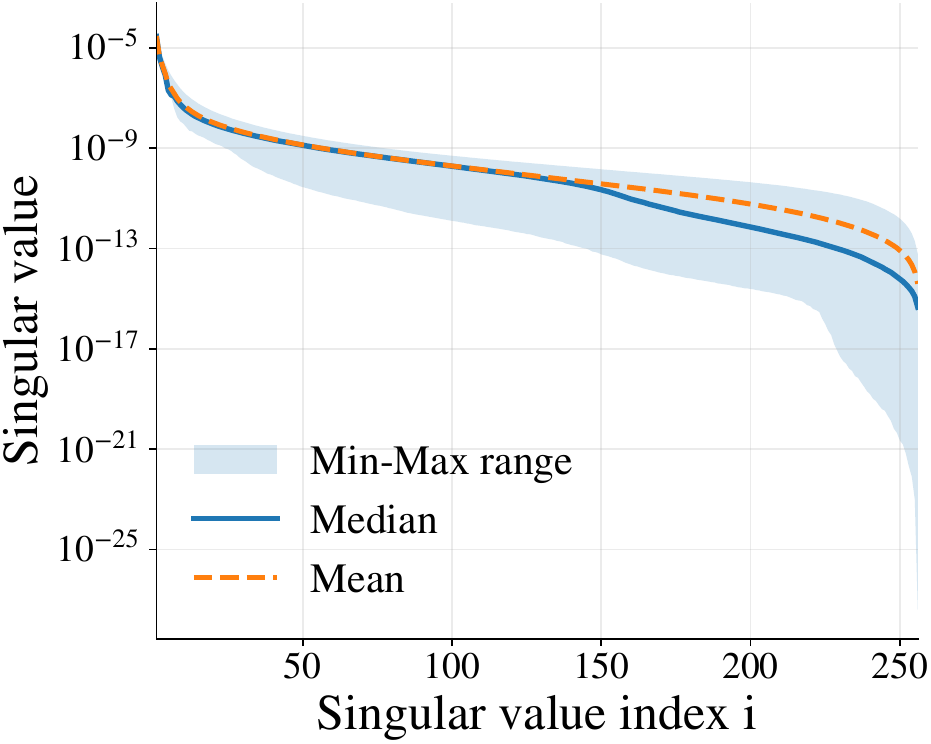}
    \caption{
    {Left:} reconstruction accuracy versus truncated rank;
    {Right:} singular values in a decreasing order.
    Results are computed on $10,000$ SLF samples from the RadioMapSeer dataset~\cite{DatasetPaper}.}
    \label{fig:slf_lowrank}
\end{figure}

\begin{Remark}\label{remark:ourproof}
As mentioned, Theorems~\ref{thm:perfect_case}--\ref{thm:imperfect_TV} and 
Corollaries~\ref{cor:cover_lowrank} establish 
instance-optimal recovery guarantees for posterior-sampling-based radio map estimation. 
Instance-optimality analysis was studied in 
\cite{jalal2021instance,jalal2021robust} for compressive sensing settings, 
where the goal is to recover a structured vector $\bm x$ from Gaussian linear measurements. 
Our work extends this line of research in two important directions. 
First, we consider models in which the prior is specified in a latent space 
(i.e., on $\bm z$) rather than directly on $\bm x$. 
Second, the sensing operator in our problem is not Gaussian but arises from 
tensor fiber subsampling. 
In addition, our analysis employs a more direct, TV-based characterization 
of the accuracy of the posterior sampler (cf. Theorem~\ref{thm:imperfect_TV}).
Compared to the Wasserstein distance based characterizations in \cite{jalal2021instance},
this TV-based characterization provides a lighter proof route for establishing 
recoverability guarantees under imperfect samplers. 
\end{Remark}

\begin{Remark}\label{remark:gap_from_alg}

We also remark on several limitations. The recoverability results in Theorems~\ref{thm:perfect_case}–\ref{thm:imperfect_TV} and Corollary~\ref{cor:cover_lowrank} are established under two important conditions.
First, these results assume that \texttt{FactorDPS} produces samples from $p(\bm x|\bm y)$ or $p’(\bm x|\bm y)$. 
In practice, exact posterior sampling may require infinitely many sampling steps and an accurate likelihood representation beyond approximations such as \ref{eq:S-score-conditional}. Therefore, it remains to establish theoretical guarantees on how closely the final step distribution produced by Algorithm~\ref{alg:lfdps} approximates the ground-truth posterior.
Second, our analysis assumes uniformly sampled spatial locations. In practice, \texttt{FactorDPS} can also accommodate structured missing regions, as demonstrated in Sec.~\ref{subsec:multi-band-cases}. Establishing analogous theoretical guarantees under such structured missing patterns may require additional analysis.
\end{Remark}

\section{Numerical Results}
\label{sec:exp}
\subsection{Simulation Setup}

\subsubsection{Dataset} \label{subsec:Dataset}
We evaluate the proposed method on the RadioMapSeer dataset \cite{DatasetPaper},
which contains radio maps generated via \textit{ray-tracing} (RT) across 701 cities. In each city, 80 independent simulations are conducted by placing a single emitter on a randomly picked building rooftop, and the received signal strength is computed through the \textit{dominant path model} (DPM)  at a horizontal plane at the height of 1.5m, yielding $80$ emitter locations per city. Each map covers a geographical area with a size of $256 \times 256$m$^2$, represented by $256 \times 256$ spatial grids.

In our simulations, we treat the RadioMapSeer-provided single-emitter maps at a single carrier frequency (5.9 GHz) as SLFs (i.e., $\bm S_r$).
We denote by $\bm S_{\ell,n}$ the SLF corresponding to the $\ell$-th transmitter in the $n$-th city, where $\ell \in [80]$ and $n\in[701]$. 
We use disjoint environments for training and testing. Radio maps from city scenarios $\{1,\ldots,625\}$ are used for training, while those indexed by $\{626,\ldots,701\}$ are used for testing, i.e., we have
$
\mathbb S_{\rm tr}^\natural
\triangleq
\left\{
\bm S_{\ell,n} \in \mathbb R^{256\times256}:
\ell\in[80],~n\in[625]
\right\}$ and $
\mathbb S_{\rm te}^\natural
\triangleq
\left\{
\bm S_{\ell,n} \in \mathbb R^{256\times256}:
\ell\in[80],~n\in\{626,\ldots,701\}
\right\}$ as the training and testing SLF pools, respectively.
At test time, $R$ SLFs from the same city are selected from $\mathbb S_{\rm te}^\natural$. For single-band cases, radio maps are constructed according to \eqref{eq:multi-R,single-K}, where $c_r=1$ for all $r\in[R]$. For multi-band cases, $R$ PSDs are generated as in \cite{Shrestha2022DeepSpectrum} (also see Fig.~\ref{fig:PSD}), and the radio maps are synthesized using \eqref{eq:X=SC}.

\subsubsection{Baselines}

\begin{table}[t]
\centering
\caption{Methodological comparison with baseline methods.}
\label{tab:method_comparison}
\footnotesize
\resizebox{\columnwidth}{!}{
\begin{tabular}{lccc}
\hline
Method
& Learned model
& \shortstack{Measurement-\\independent prior}
& \shortstack{Spatio-spectral\\factorization} \\
\hline
\texttt{DeepComp}~\cite{teganya2021deep}
& AE & \xmark & \xmark \\

\texttt{RadioUNet}~\cite{Levie2021radiounet}
& U-Net & \xmark & \xmark \\

\texttt{DowJons}~\cite{Shrestha2022DeepSpectrum}
& AE & \cmark & \cmark \\

\texttt{RM-Gen}~\cite{luo2025denoising}
& DM & \xmark & \xmark \\

\texttt{DPS}~\cite{chung2023diffusion}
& DM & \cmark & \xmark \\

\texttt{FactorDPS}
& DM & \cmark & \cmark \\
\hline
\end{tabular}
}
\end{table}

We compare the proposed methods with several representative baselines. 
Three end-to-end learning-based approaches are considered: 
i) {\tt DeepComp} \cite{teganya2021deep}, which employs an autoencoder architecture; 
ii) {\tt RadioUNet} \cite{Levie2021radiounet}, which adopts a W-Net formed by concatenating two U-Nets for enhanced feature extraction; and 
iii) {\tt RM-Gen} \cite{luo2025denoising}, a diffusion-based approach that performs conditional generation given the measurements. Unlike {\tt FactorDPS}, {\tt RM-Gen} couples the diffusion model with the sensor measurements during training, whereas {\tt FactorDPS} learns a measurement-independent prior and enforces measurement consistency only during inference.
We also include iv) {\tt DowJons} \cite{Shrestha2022DeepSpectrum}, an optimization-enhanced method that exploits the same spatio-spectral factorization as {\tt FactorDPS}, but uses a pretrained AE prior for the SLFs instead of a DM prior.
Table~\ref{tab:method_comparison}
summarizes the key methodological characteristics of {\tt FactorDPS} and
representative deep prior-based methods.

\subsubsection{Algorithm Settings} \label{subsec:training}

The training data are generated following the same procedure as the testing data in Sec.~\ref{subsec:Dataset}, but using $\mathbb S_{\rm tr}^\natural$ instead of $\mathbb S_{\rm te}^\natural$.
The training procedure is summarized as follows:

\begin{itemize}[leftmargin=4mm]

\item \textbf{End-to-end baselines.}
For {\tt DeepComp}, {\tt RadioUNet}, and {\tt RM-Gen}, the measurements and building masks are provided as inputs. During training, the sampling ratio $\tau$ is drawn from $\mathcal{U}(0.01,0.1)$. For each method, three models are trained to match the corresponding testing settings:
i) for the single-emitter, single band setting ($R=1,K=1$), the radio maps in $\mathbb S_{\rm tr}^\natural$ are directly used for training;
ii) for the multi-emitter, single-band setting ($R>1, K=1$), the methods are re-trained using radio maps constructed according to \eqref{eq:multi-R,single-K}, where $R$ is drawn uniformly from $\{1,2,3,4\}$;
iii) for the multi-emitter, multi-band setting ($R>1, K>1$), the methods are trained on per-band radio maps generated as
$\bm X=\sum_{r=1}^{R}c_r\bm S_{\ell_r,n}$,
where $R$ is drawn uniformly from $\{1,2,3,4\}$ and
$c_r\sim\mathcal{U}(0,1)$.
In each epoch, 50{,}000 training samples are used.

\item \textbf{Decomposition based methods.}
{\tt DowJons} and the proposed {\tt FactorDPS} exploit the spatio-spectral decomposition in \eqref{eq:X=SC} and use generative models of single-emitter
SLFs, which are trained directly using $\mathbb S_{\rm tr}^\natural$.
In particular, for {\tt FactorDPS}, the SLF diffusion network takes a noisy
SLF, the building map, and timestep $t$ as inputs and predicts the score
$\bm f_{\rm SLF}$ used in Algorithm~\ref{alg:lfdps}.
For multi-band settings, an additional diffusion network is trained for the
PSD's score function $\bm f_{\rm PSD}$.
A total of 50{,}000 synthetic PSD samples are used for training; see Sec.~\ref{subsec:LatentTraining}.
\end{itemize}

The network architectures of the SLF and PSD DMs in \texttt{FactorDPS} are summarized in the supplementary materials. Both DMs use 2,000 reverse diffusion steps. Other baselines follow their default network setups, with input dimensions adjusted to the $256\times256$ spatial size. The only exception is {\tt RM-Gen}, for which we adopt essentially the same diffusion settings and network architecture as the SLF DM in {\tt FactorDPS}, so as to enable a controlled comparison. The only architectural difference is that {\tt RM-Gen} includes an additional input channel for the sampled measurements.

At the inference stage, for single-band cases, all methods can be directly applied. In particular, {\tt FactorDPS} uses only the learned SLF score model with $\bm c_r=1$ for all $r\in[R]$; see Sec.~\ref{subsec:single-band-case}. For multi-band cases, {\tt DowJons} and {\tt FactorDPS} are directly applicable. In contrast, {\tt DeepComp}, {\tt RadioUNet}, and {\tt RM-Gen} are applied independently to each band. The reconstructed bands are then stacked to form the full radio map tensor.
For all methods that require the knowledge of the number of emitters $R$, unless otherwise specified, we assume that $R$ is known or has been estimated prior to inference. In practice, such information may be obtained from transmitter deployment records maintained by the network operator, or estimated from the measurements using model order estimation methods 
\cite{stoica2005spectral,fu2015self}.

\begin{figure}[t!]
    \centering

    \begin{subfigure}{0.98\linewidth}
        \centering
        \includegraphics[width=\linewidth]{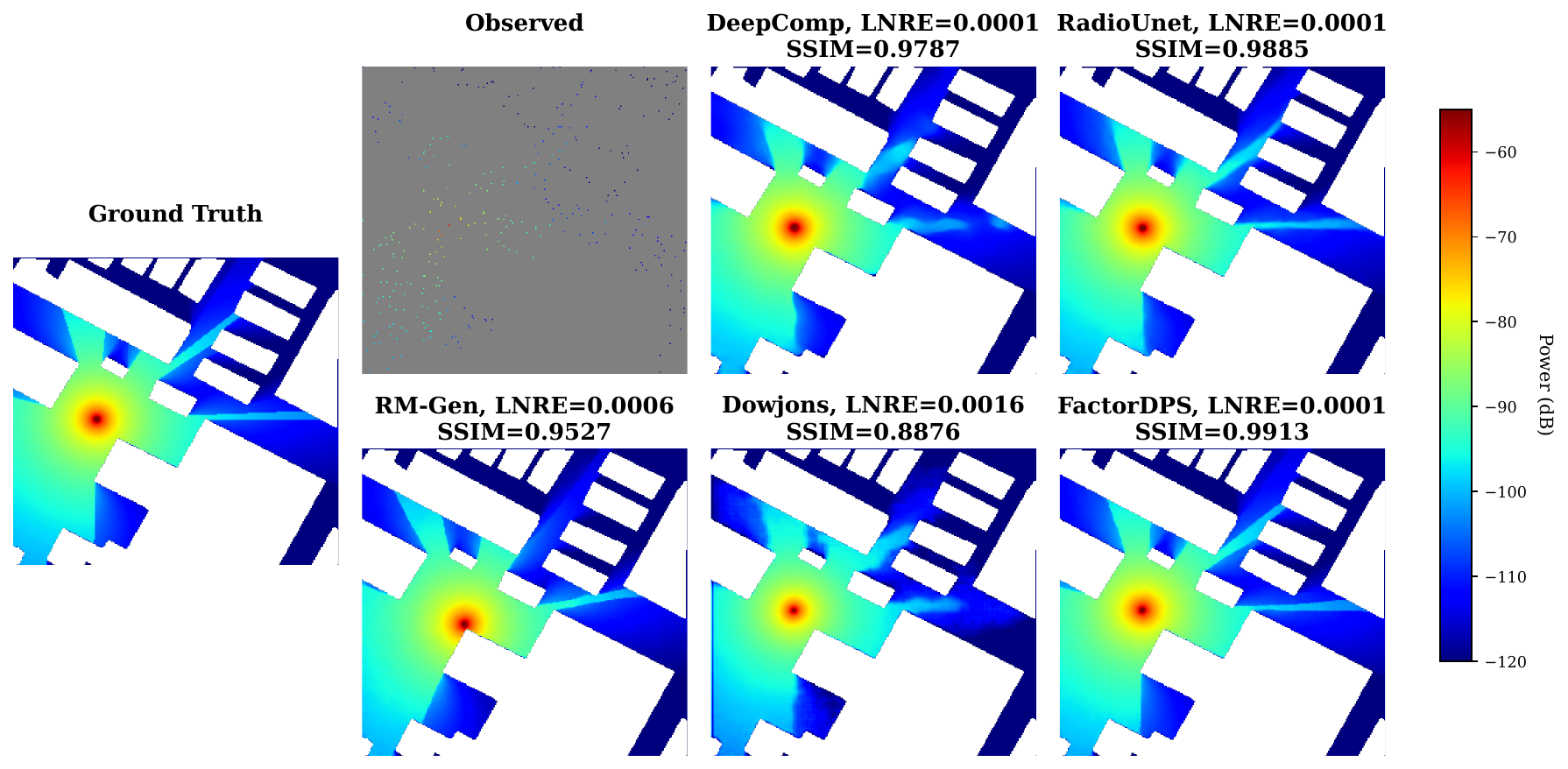}
        \caption{$R=1$}
        \label{subfig:K=1R=1}
    \end{subfigure}
    \hfill
    \begin{subfigure}{0.98\linewidth}
        \centering
        \includegraphics[width=\linewidth]{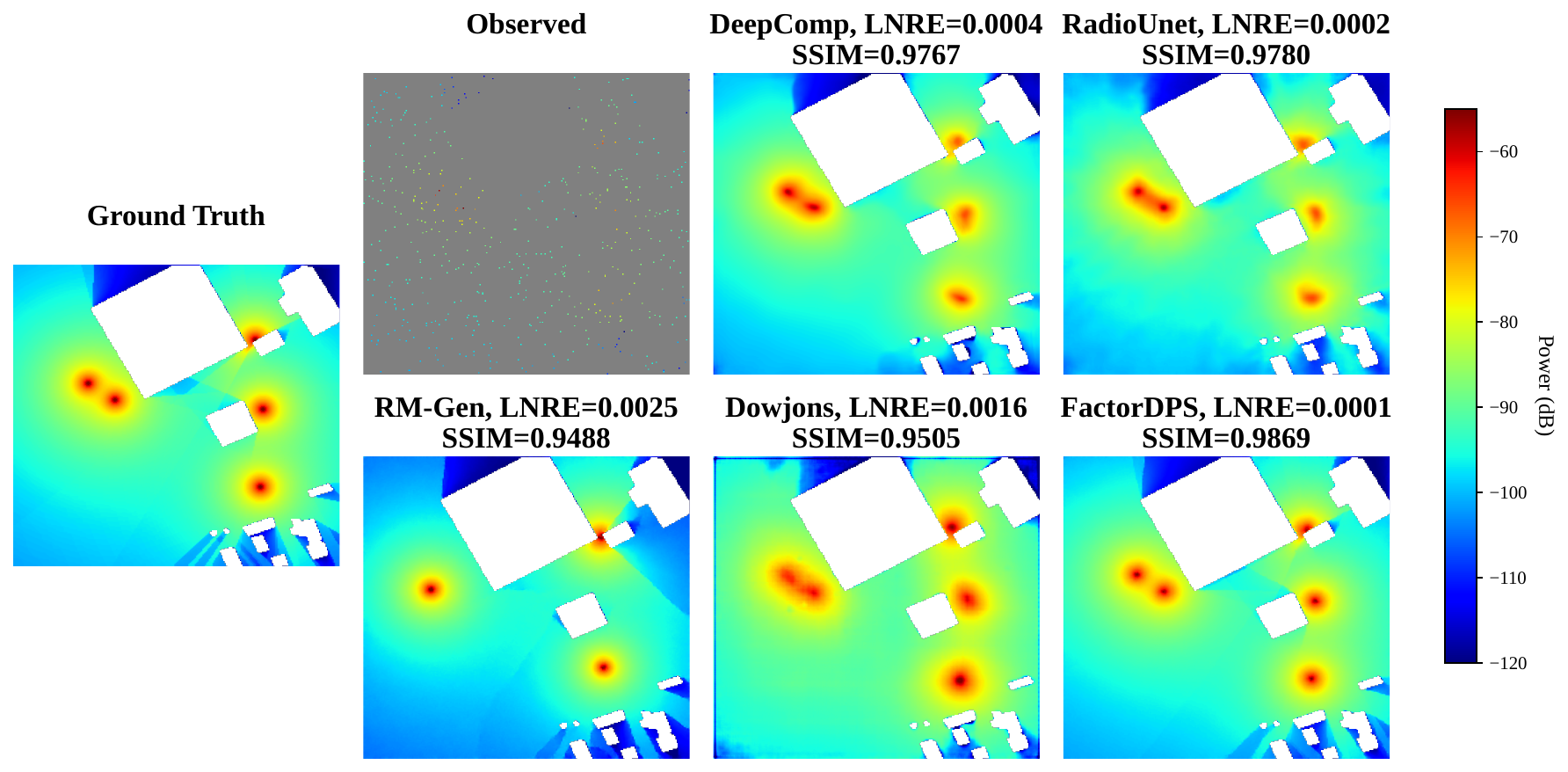}
        \caption{$R=5$}
        \label{subfig:K=1R=N}
    \end{subfigure}
    \caption{Recovered radio maps under $1\%$ sampling ratio in single-band cases ($K=1$).}
    \label{fig:K=1}
\end{figure}

\subsubsection{Metrics}
We evaluate reconstruction performance using the {log-domain normalized reconstruction error} (LNRE) \cite{Subash2024Quantized}:
\begin{align}
    \text{LNRE} = \frac{ \| 10\log(\mathbm X^{\natural}) -  10\log(\hat{\mathbm X}) \|_{\rm F}^2}
    {\| 10\log(\mathbm X^{\natural}) \|_{\rm F}^2},
\end{align}
where $\mathbm X^{\natural}$ and $\hat{\mathbm X}$ denote the ground-truth and reconstructed radio maps, respectively.
{The LNRE is} evaluated in the log domain, as radio maps {have very imbalanced pixel values}. The log-domain evaluation provides a more balanced error characterization across the entire spatial region.
In addition, the structural similarity index measure (SSIM) \cite{Wang2004SSIM} is employed to evaluate perceptual reconstruction quality ({also in the log-domain}). The SSIM score ranges from $0$ to $1$, {with higher scores indicating better performance.} For multi-band data, the mean SSIM (MSSIM) is then obtained by averaging the SSIM values across all $K$ bands.
All reported results are presented as mean $\pm$ standard deviation over 50 Monte Carlo trials. In each trial, a city scenario and the corresponding source SLFs are randomly selected to construct the test radio map.

\subsection{Results of Single-Band Cases}
\begin{table}[t!]
\caption{Performance under different sampling rate $\tau$ ($K=1,R=1$).}
\footnotesize
\centering
\centering
\scalebox{0.72}{
\begin{tabular}{ |c | c c c c c |}
\hline
 & \multicolumn{5}{c|}{LNRE ($\times 10^{-4}$)}\\
$\tau$ & \texttt{DeepComp} & \texttt{RadioUnet} & \texttt{RM-Gen} & \texttt{Dowjons} & \texttt{FactorDPS} \\
 \hline
 1\% &    3.10 $\pm$ 1.55 &    1.16 $\pm$ 0.51 &    3.91 $\pm$ 1.88 &    145.55 $\pm$ 159.81 &    \textbf{0.98} $\pm$ 0.35 \\
 2\% &    1.70 $\pm$ 0.73 &    0.78 $\pm$ 0.31 &    3.10 $\pm$ 1.41 &    57.82 $\pm$ 53.90 &    \textbf{0.73} $\pm$ 0.27 \\
 4\% &    1.75 $\pm$ 1.04 &    \textbf{0.68} $\pm$ 0.28 &    2.69 $\pm$ 1.28 &    34.66 $\pm$ 22.53 &    0.69 $\pm$ 0.27 \\
 6\% &    1.62 $\pm$ 0.90 &    0.65 $\pm$ 0.30 &    1.91 $\pm$ 0.66 &    25.33 $\pm$ 11.36 &    \textbf{0.64} $\pm$ 0.23 \\
 8\% &    1.54 $\pm$ 0.85 &    \textbf{0.57} $\pm$ 0.27 &    2.01 $\pm$ 0.80 &    22.26 $\pm$ 6.84 &    0.57 $\pm$ 0.23 \\
 10\% &    1.44 $\pm$ 0.83 &    0.55 $\pm$ 0.24 &    1.61 $\pm$ 0.52 &    21.54 $\pm$ 8.65 &    \textbf{0.54} $\pm$ 0.22 \\
\hline
\hline
 & \multicolumn{5}{c|}{MSSIM}\\
$\tau$ & \texttt{DeepComp} & \texttt{RadioUnet} & \texttt{RM-Gen} & \texttt{Dowjons} & \texttt{FactorDPS} \\
 \hline
 1\% &    0.966 $\pm$ 0.016 &    0.981 $\pm$ 0.008 &    0.963 $\pm$ 0.015 &    0.793 $\pm$ 0.074 &    \textbf{0.984} $\pm$ 0.006 \\
 2\% &    0.976 $\pm$ 0.012 &    0.986 $\pm$ 0.006 &    0.968 $\pm$ 0.015 &    0.824 $\pm$ 0.076 &    \textbf{0.988} $\pm$ 0.005 \\
 4\% &    0.977 $\pm$ 0.014 &    0.988 $\pm$ 0.006 &    0.971 $\pm$ 0.013 &    0.881 $\pm$ 0.052 &    \textbf{0.990} $\pm$ 0.004 \\
 6\% &    0.978 $\pm$ 0.011 &    0.989 $\pm$ 0.005 &    0.974 $\pm$ 0.010 &    0.876 $\pm$ 0.050 &    \textbf{0.990} $\pm$ 0.003 \\
 8\% &    0.979 $\pm$ 0.011 &    0.991 $\pm$ 0.004 &    0.973 $\pm$ 0.010 &    0.890 $\pm$ 0.053 &    \textbf{0.991} $\pm$ 0.003 \\
 10\% &    0.981 $\pm$ 0.012 &    0.992 $\pm$ 0.004 &    0.979 $\pm$ 0.008 &    0.890 $\pm$ 0.053 &    \textbf{0.992} $\pm$ 0.003 \\
\hline
\end{tabular}}
\label{table:mr_singleK_meanstd_without}
\end{table}

\begin{table}[t!]
\caption{Performance under different $R$ ($K=1$, $\tau=1\%$).}
\footnotesize
\centering
\centering
\scalebox{0.72}{
\begin{tabular}{ |c | c c c c c |}
\hline
 & \multicolumn{5}{c|}{LNRE ($\times 10^{-4}$)}\\
$R$ & \texttt{DeepComp} & \texttt{RadioUnet} & \texttt{RM-Gen} & \texttt{Dowjons} & \texttt{FactorDPS} \\
 \hline
 2 &    3.60 $\pm$ 1.81 &    5.11 $\pm$ 2.51 &    17.71 $\pm$ 8.31 &    75.86 $\pm$ 82.15 &    \textbf{1.10} $\pm$ 0.43 \\
 3 &    3.89 $\pm$ 2.23 &    4.79 $\pm$ 1.31 &    27.84 $\pm$ 11.09 &    75.62 $\pm$ 94.03 &    \textbf{1.19} $\pm$ 0.60 \\
 4 &    4.77 $\pm$ 2.40 &    3.63 $\pm$ 1.19 &    38.68 $\pm$ 13.86 &    126.39 $\pm$ 158.86 &    \textbf{1.69} $\pm$ 1.12 \\
 5 &    5.16 $\pm$ 2.71 &    3.45 $\pm$ 1.09 &    39.14 $\pm$ 14.98 &    57.29 $\pm$ 60.24 &    \textbf{1.90} $\pm$ 1.22 \\
 6 &    5.30 $\pm$ 2.93 &    3.30 $\pm$ 1.43 &    44.81 $\pm$ 13.69 &    88.13 $\pm$ 86.61 &    \textbf{1.95} $\pm$ 0.98 \\
 7 &    5.42 $\pm$ 2.92 &    3.49 $\pm$ 1.16 &    53.42 $\pm$ 18.95 &    123.21 $\pm$ 157.43 &    \textbf{2.22} $\pm$ 1.23 \\
 8 &    6.95 $\pm$ 3.89 &    3.31 $\pm$ 1.21 &    62.67 $\pm$ 18.39 &    112.48 $\pm$ 120.91 &    \textbf{2.82} $\pm$ 1.35 \\
\hline
\hline
 & \multicolumn{5}{c|}{MSSIM}\\
$R$ & \texttt{DeepComp} & \texttt{RadioUnet} & \texttt{RM-Gen} & \texttt{Dowjons} & \texttt{FactorDPS} \\
 \hline
 2 &    0.963 $\pm$ 0.018 &    0.958 $\pm$ 0.016 &    0.939 $\pm$ 0.019 &    0.824 $\pm$ 0.081 &    \textbf{0.982} $\pm$ 0.007 \\
 3 &    0.965 $\pm$ 0.019 &    0.961 $\pm$ 0.016 &    0.939 $\pm$ 0.019 &    0.841 $\pm$ 0.075 &    \textbf{0.983} $\pm$ 0.008 \\
 4 &    0.961 $\pm$ 0.021 &    0.962 $\pm$ 0.018 &    0.926 $\pm$ 0.025 &    0.821 $\pm$ 0.103 &    \textbf{0.989} $\pm$ 0.005 \\
 5 &    0.961 $\pm$ 0.020 &    0.963 $\pm$ 0.017 &    0.920 $\pm$ 0.025 &    0.869 $\pm$ 0.060 &    \textbf{0.980} $\pm$ 0.011 \\
 6 &    0.962 $\pm$ 0.021 &    0.966 $\pm$ 0.016 &    0.920 $\pm$ 0.021 &    0.854 $\pm$ 0.058 &    \textbf{0.980} $\pm$ 0.010 \\
 7 &    0.965 $\pm$ 0.020 &    0.971 $\pm$ 0.015 &    0.919 $\pm$ 0.023 &    0.800 $\pm$ 0.133 &    \textbf{0.981} $\pm$ 0.009 \\
 8 &    0.959 $\pm$ 0.021 &    0.965 $\pm$ 0.014 &    0.907 $\pm$ 0.025 &    0.825 $\pm$ 0.094 &    \textbf{0.980} $\pm$ 0.010 \\
\hline
\end{tabular}}
\label{table:mr_exp6R_meanstd_without}
\end{table}

Fig.~\ref{fig:K=1} illustrates the reconstruction results {by all the methods} under a sampling ratio of $\tau = 1\%$. 
Fig.~\ref{subfig:K=1R=1} presents the single-emitter case, where all compared methods achieve reasonably good recovery performance. 
Fig.~\ref{subfig:K=1R=N} shows {the results obtained from a} more challenging multi-emitter scenario. 
In this setting, {\tt FactorDPS} provides clearer emitter identification and sharper ray-tracing structures. 
The diffusion baseline {\tt RM-Gen} could only recover $R=3$ emitters while $R=5$, revealing generalization issues beyond its training range $R \in [4]$ (see Sec.~\ref{subsec:training}).

Table~\ref{table:mr_singleK_meanstd_without} reports the performance under various sampling ratios $\tau$, with $R=1$ and $K=1$. 
Since this is a relatively simple setting, all methods exhibit similar performance, except for {\tt Dowjons}, which performs worse, possibly due to its relatively simple network design. Overall, {\tt FactorDPS} delivers the best performance.

Table~\ref{table:mr_exp6R_meanstd_without} further summarizes the results across different $R$'s under $\tau=1\%$. 
As the number of emitters increases, the reconstruction task becomes increasingly challenging. Nonetheless, {\tt FactorDPS} outperforms the baselines across all values of $R$.
Again, the diffusion baseline {\tt RM-Gen} performs significantly worse than the proposed {\tt FactorDPS}, and its performance degrades rapidly as $R$ increases.

\subsection{Results of Multi-Band Cases}
\label{subsec:multi-band-cases}
In this subsection, we test the methods using multi-emitter, multi-band (with $K=64$ frequency bands) radio maps.

Fig.~\ref{fig:K=64Normal} presents the reconstruction results at the $48$th frequency bin under a sampling rate of $\tau = 1\%$. 
Fig.~\ref{subfig:R=3} and Fig.~\ref{subfig:R=8} correspond to scenarios with $R=3$ and $R=8$ emitters, respectively. 
In both cases, {\tt FactorDPS} achieves the clearest reconstructions, with more accurate emitter localization and ray-tracing contours.
Moreover, {\tt RM-Gen} again exhibits notable generalization issues in Fig.~\ref{subfig:R=8}, where it reconstructs only the three strongest emitters.

\begin{figure}[t]
    \centering

    \begin{subfigure}{0.98\linewidth}
        \centering
        \includegraphics[width=\linewidth]{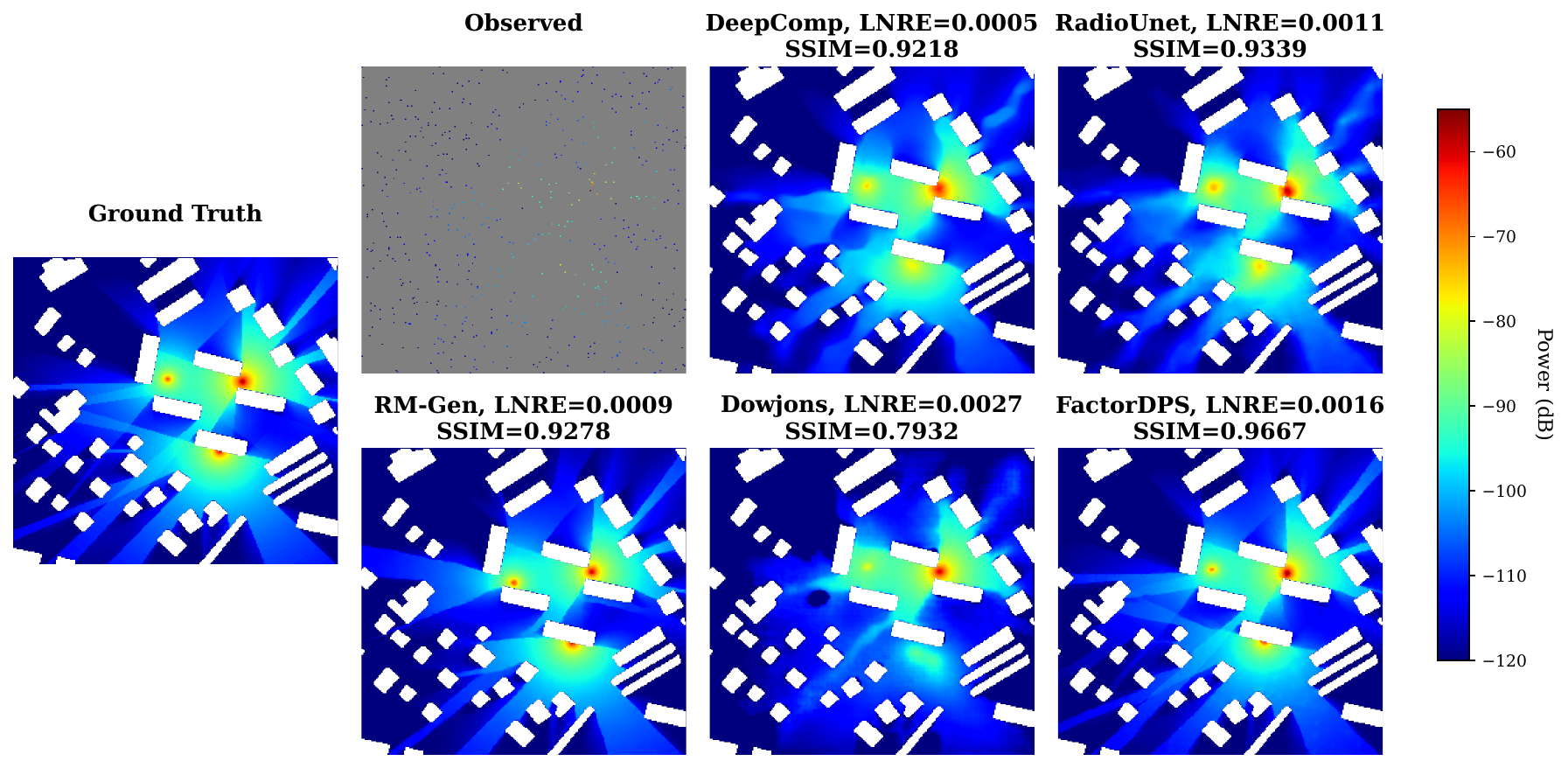}
        \caption{$R=3$.}
        \label{subfig:R=3}
    \end{subfigure}
    \hfill
    \begin{subfigure}{0.98\linewidth}
        \centering
        \includegraphics[width=\linewidth]{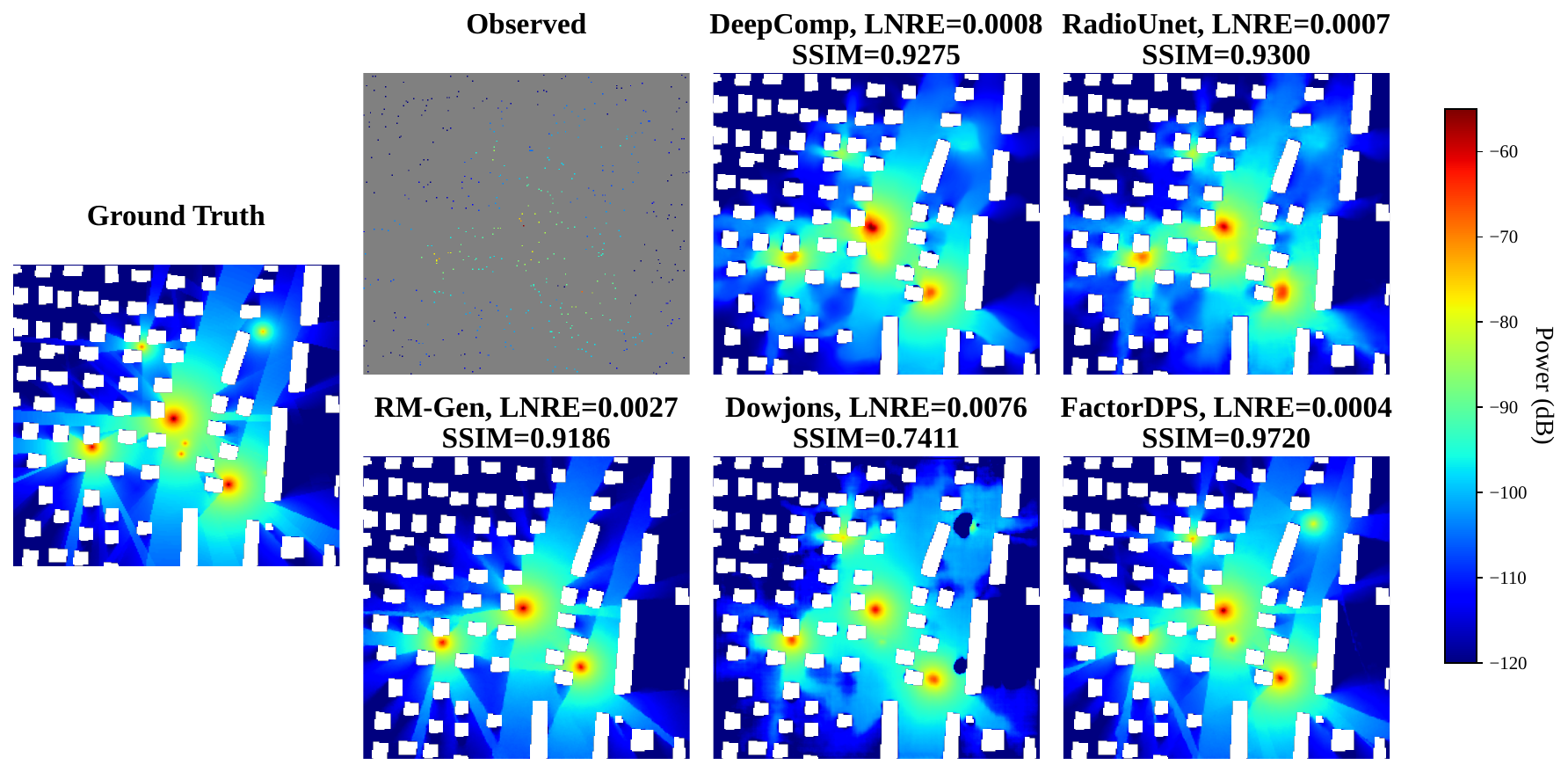}
        \caption{$R=8$.}
        \label{subfig:R=8}
    \end{subfigure}
    \caption{Recovered radio maps at the $48$th frequency bin from $1\%$ sampling ratio.}
    \label{fig:K=64Normal}
\end{figure}

\begin{table}[t!]
\caption{Performance under different sampling rate $\tau$ ($K=64$, $R=3$).}
\footnotesize
\centering
\centering
\scalebox{0.72}{
\begin{tabular}{ |c | c c c c c |}
\hline
 & \multicolumn{5}{c|}{LNRE ($\times 10^{-4}$)}\\
$\tau$ & \texttt{DeepComp} & \texttt{RadioUnet} & \texttt{RM-Gen} & \texttt{Dowjons} & \texttt{FactorDPS} \\
 \hline
 1\% &    6.37 $\pm$ 1.84 &    \textbf{5.54} $\pm$ 3.08 &    10.68 $\pm$ 0.66 &    60.67 $\pm$ 41.87 &    5.87 $\pm$ 4.41 \\
 2\% &    4.72 $\pm$ 1.19 &    4.14 $\pm$ 1.87 &    8.38 $\pm$ 0.68 &    38.81 $\pm$ 14.22 &    \textbf{3.95} $\pm$ 2.39 \\
 4\% &    4.06 $\pm$ 1.07 &    \textbf{3.44} $\pm$ 1.23 &    8.22 $\pm$ 1.32 &    29.26 $\pm$ 12.17 &    \textbf{3.74} $\pm$ 2.39 \\
 6\% &    4.03 $\pm$ 1.04 &    \textbf{3.29} $\pm$ 1.37 &    10.57 $\pm$ 4.25 &    22.59 $\pm$ 9.31 &    3.65 $\pm$ 3.60 \\
 8\% &    3.80 $\pm$ 0.94 &    3.36 $\pm$ 1.44 &    5.45 $\pm$ 0.08 &    19.68 $\pm$ 6.62 &    \textbf{3.13} $\pm$ 2.18 \\
 10\% &   3.74 $\pm$ 0.78 &    \textbf{2.99} $\pm$ 1.08 &    6.03 $\pm$ 0.88 &    15.90 $\pm$ 5.24 &    3.64 $\pm$ 2.51 \\ 
\hline
 & \multicolumn{5}{c|}{MSSIM}\\
$\tau$ & \texttt{DeepComp} & \texttt{RadioUnet} & \texttt{RM-Gen} & \texttt{Dowjons} & \texttt{FactorDPS} \\
 \hline
 1\% &    0.945 $\pm$ 0.028 &    0.916 $\pm$ 0.051 &    0.903 $\pm$ 0.024 &    0.825 $\pm$ 0.055 &    \textbf{0.950} $\pm$ 0.027 \\
 2\% &    0.956 $\pm$ 0.019 &    0.932 $\pm$ 0.038 &    0.926 $\pm$ 0.021 &    0.839 $\pm$ 0.041 &    \textbf{0.959} $\pm$ 0.018 \\
 4\% &    0.964 $\pm$ 0.017 &    0.943 $\pm$ 0.037 &    0.913 $\pm$ 0.017 &    0.871 $\pm$ 0.046 &    \textbf{0.967} $\pm$ 0.019 \\
 6\% &    0.964 $\pm$ 0.016 &    0.948 $\pm$ 0.028 &    0.936 $\pm$ 0.006 &    0.875 $\pm$ 0.045 &    \textbf{0.965} $\pm$ 0.021 \\
 8\% &    0.970 $\pm$ 0.014 &    0.944 $\pm$ 0.033 &    0.937 $\pm$ 0.008 &    0.881 $\pm$ 0.045 &    \textbf{0.971} $\pm$ 0.015 \\
 10\% &   \textbf{0.972} $\pm$ 0.013 &    0.944 $\pm$ 0.033 &    0.939 $\pm$ 0.001 &    0.901 $\pm$ 0.041 &    0.970 $\pm$ 0.018 \\
\hline
\end{tabular}}
\label{table:mr_multiK_meanstd_without}
\end{table}

\begin{table}[t!]
\caption{Performance under different emitter numbers $R$ ($K=64$, $\tau=1\%$).}
\footnotesize
\centering
\centering
\scalebox{0.72}{
\begin{tabular}{ |c | c c c c c |}
\hline
 & \multicolumn{5}{c|}{LNRE ($\times 10^{-4}$)}\\
$R$ & \texttt{DeepComp} & \texttt{RadioUnet} & \texttt{RM-Gen} & \texttt{Dowjons} & \texttt{FactorDPS} \\
 \hline
 2 &    6.11 $\pm$ 1.34 &    \textbf{5.79} $\pm$ 2.83 &    7.36 $\pm$ 1.36 &    42.48 $\pm$ 29.77 &    6.50 $\pm$ 5.19 \\
 3 &    6.23 $\pm$ 1.77 &    \textbf{5.74} $\pm$ 3.21 &    10.94 $\pm$ 1.13 &    42.89 $\pm$ 40.63 &    6.01 $\pm$ 4.30 \\
 4 &    6.00 $\pm$ 1.24 &    \textbf{4.94} $\pm$ 1.35 &    14.48 $\pm$ 0.55 &    57.95 $\pm$ 34.24 &    4.96 $\pm$ 2.83 \\
 5 &    6.43 $\pm$ 1.73 &    5.34 $\pm$ 1.51 &    15.53 $\pm$ 3.71 &    68.29 $\pm$ 37.84 &    \textbf{4.75} $\pm$ 2.94 \\
 6 &    6.67 $\pm$ 1.99 &    4.79 $\pm$ 1.48 &    19.55 $\pm$ 5.87 &    85.54 $\pm$ 49.31 &    \textbf{4.24} $\pm$ 3.17 \\
 7 &    7.30 $\pm$ 2.08 &    5.20 $\pm$ 1.26 &    21.46 $\pm$ 3.46 &    90.75 $\pm$ 47.80 &    \textbf{4.18} $\pm$ 2.42 \\
 8 &    7.61 $\pm$ 2.24 &    5.21 $\pm$ 1.48 &    25.14 $\pm$ 3.07 &    84.08 $\pm$ 60.55 &    \textbf{3.88} $\pm$ 2.30 \\
\hline
\hline
 & \multicolumn{5}{c|}{MSSIM}\\
$R$ & \texttt{DeepComp} & \texttt{RadioUnet} & \texttt{RM-Gen} & \texttt{Dowjons} & \texttt{FactorDPS} \\
 \hline
 2 &    0.939 $\pm$ 0.026 &    0.899 $\pm$ 0.055 &    0.929 $\pm$ 0.023 &    0.812 $\pm$ 0.075 &    \textbf{0.941} $\pm$ 0.032 \\
 3 &    0.943 $\pm$ 0.029 &    0.917 $\pm$ 0.049 &    0.910 $\pm$ 0.026 &    0.829 $\pm$ 0.061 &    \textbf{0.951} $\pm$ 0.028 \\
 4 &    0.954 $\pm$ 0.019 &    0.944 $\pm$ 0.025 &    0.928 $\pm$ 0.021 &    0.821 $\pm$ 0.055 &    \textbf{0.959} $\pm$ 0.019 \\
 5 &    0.954 $\pm$ 0.021 &    0.949 $\pm$ 0.023 &    0.924 $\pm$ 0.014 &    0.815 $\pm$ 0.050 &    \textbf{0.961} $\pm$ 0.019 \\
 6 &    0.954 $\pm$ 0.023 &    0.954 $\pm$ 0.021 &    0.921 $\pm$ 0.024 &    0.806 $\pm$ 0.059 &    \textbf{0.963} $\pm$ 0.023 \\
 7 &    0.952 $\pm$ 0.022 &    0.954 $\pm$ 0.018 &    0.924 $\pm$ 0.028 &    0.788 $\pm$ 0.053 &    \textbf{0.962} $\pm$ 0.022 \\
 8 &    0.953 $\pm$ 0.022 &    0.955 $\pm$ 0.020 &    0.898 $\pm$ 0.016 &    0.788 $\pm$ 0.056 &    \textbf{0.966} $\pm$ 0.018 \\
\hline
\end{tabular}}
\label{table:mr_multiK_multiR_meanstd_without}
\end{table}

\begin{figure}[!t]
    \centering

    \begin{subfigure}{0.48\linewidth}
        \centering
        \includegraphics[width=\linewidth]{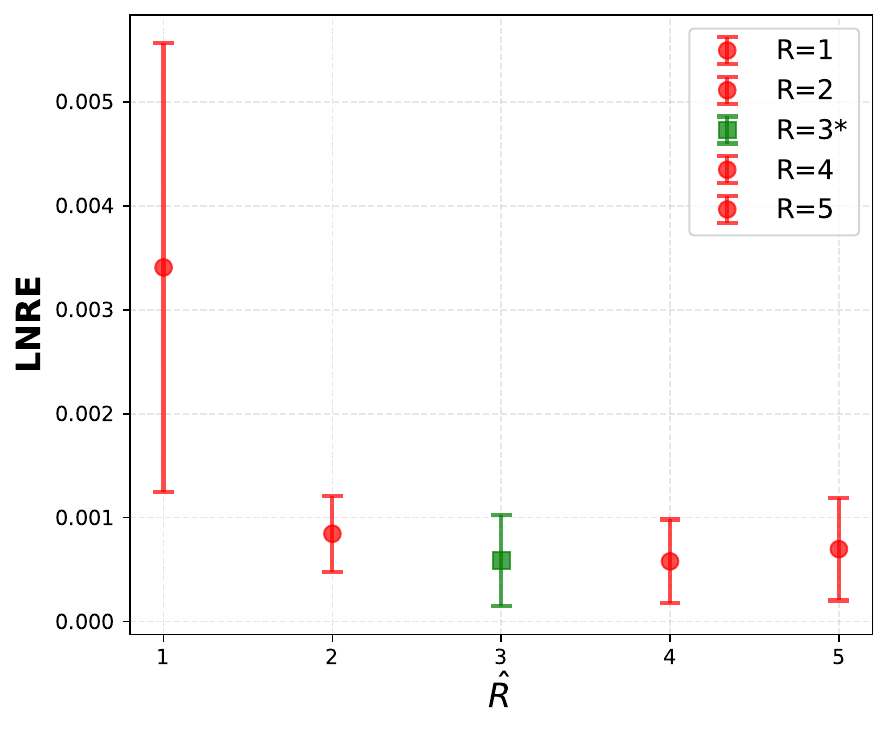}
        \caption{LNRE}
        \label{subfig:wrongR_lnre}
    \end{subfigure}
    \hfill
    \begin{subfigure}{0.48\linewidth}
        \centering
        \includegraphics[width=\linewidth]{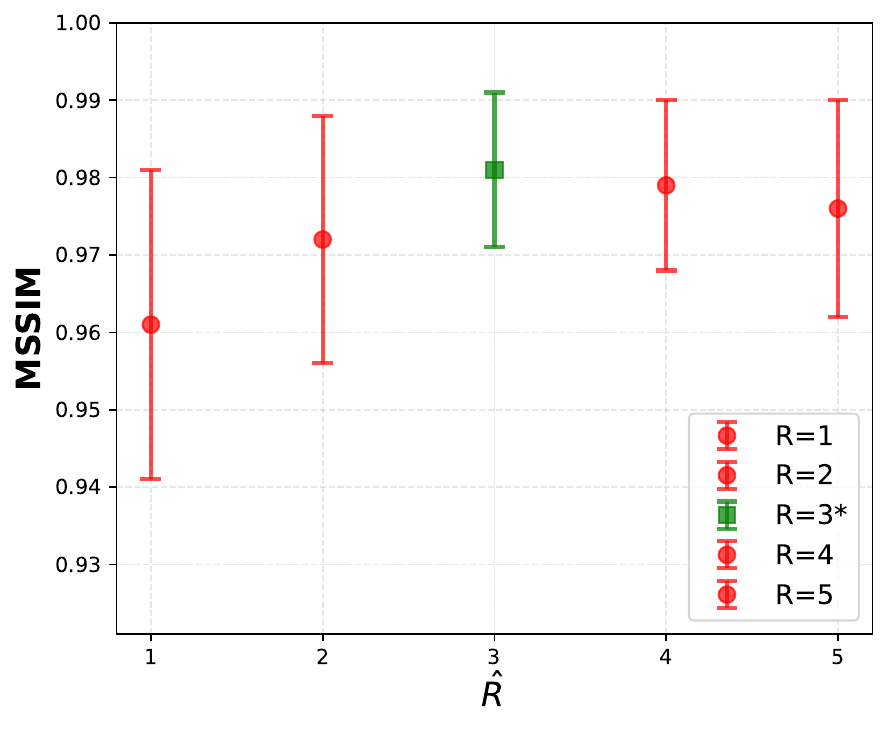}
        \caption{MSSIM}
        \label{subfig:wrongR_ssim}
    \end{subfigure}
    \caption{Performance with Different $\hat{R}$ ($R=3$, $\tau=1\%$).}
    \label{fig:wrongR}
\end{figure}

Table~\ref{table:mr_multiK_meanstd_without} summarizes the quantitative performance under various sampling rates $\tau$ with $R=3$. As shown, {\tt FactorDPS} achieves LNRE values comparable to those of the best-performing baseline, {\tt RadioUNet}, while attaining the overall highest MSSIM across all evaluated scenarios.

Table~\ref{table:mr_multiK_multiR_meanstd_without} reports the reconstruction performance under different values of $R$ at a fixed sampling rate $\tau = 1\%$. The proposed {\tt FactorDPS} achieves the best overall LNRE and consistently attains the highest MSSIM. 
As $R$ increases, the radio maps become more complex, and the performance of {\tt DeepComp}, {\tt RM-Gen}, and {\tt Dowjons} degrades noticeably. In contrast, {\tt FactorDPS} maintains improved performance as $R$ grows, demonstrating strong scalability with respect to the number of emitters.

Fig.~\ref{fig:wrongR} presents the performance of {\tt FactorDPS} under mismatched estimates of the number of emitters $\hat{R}$, with the ground truth set to $R=3$. 
When $\hat{R} < R$, the reconstruction performance degrades noticeably, and the degradation becomes more pronounced as $\hat{R}$ decreases. 
In contrast, when $\hat{R} > R$, although performance loss is also observed, it is considerably less severe, indicating the robustness of {\tt FactorDPS} to overestimation of $R$.

\begin{figure}[t]
    \centering
    \begin{subfigure}{0.98\linewidth}
        \centering
        \includegraphics[width=\linewidth]{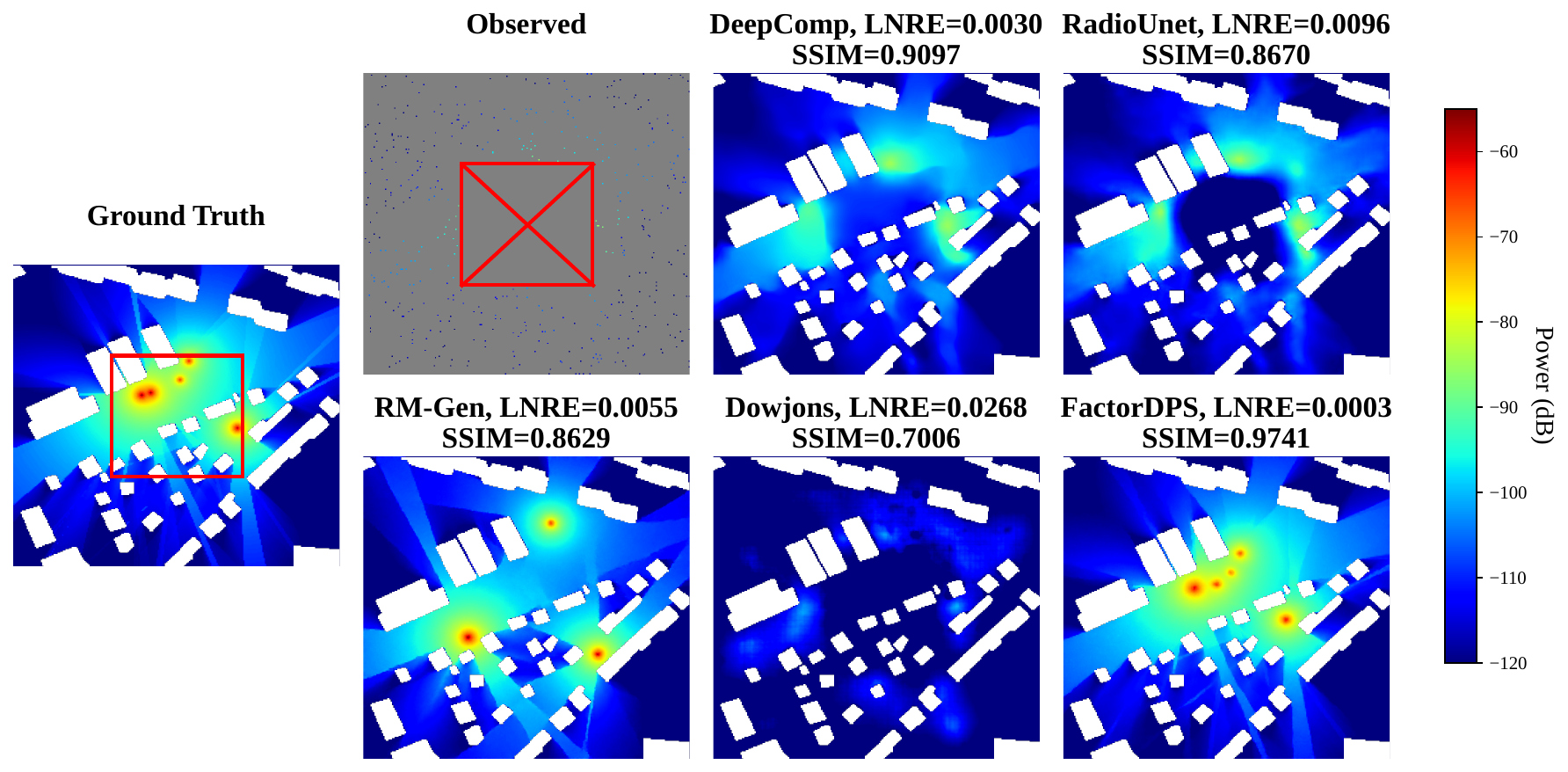}
        \caption{Central square missing region covering $40\% \times 40\%$ of the map, $R=5$}
        \label{subfig:centralmissing}
    \end{subfigure}
    \begin{subfigure}{0.98\linewidth}
        \centering
        \includegraphics[width=\linewidth]{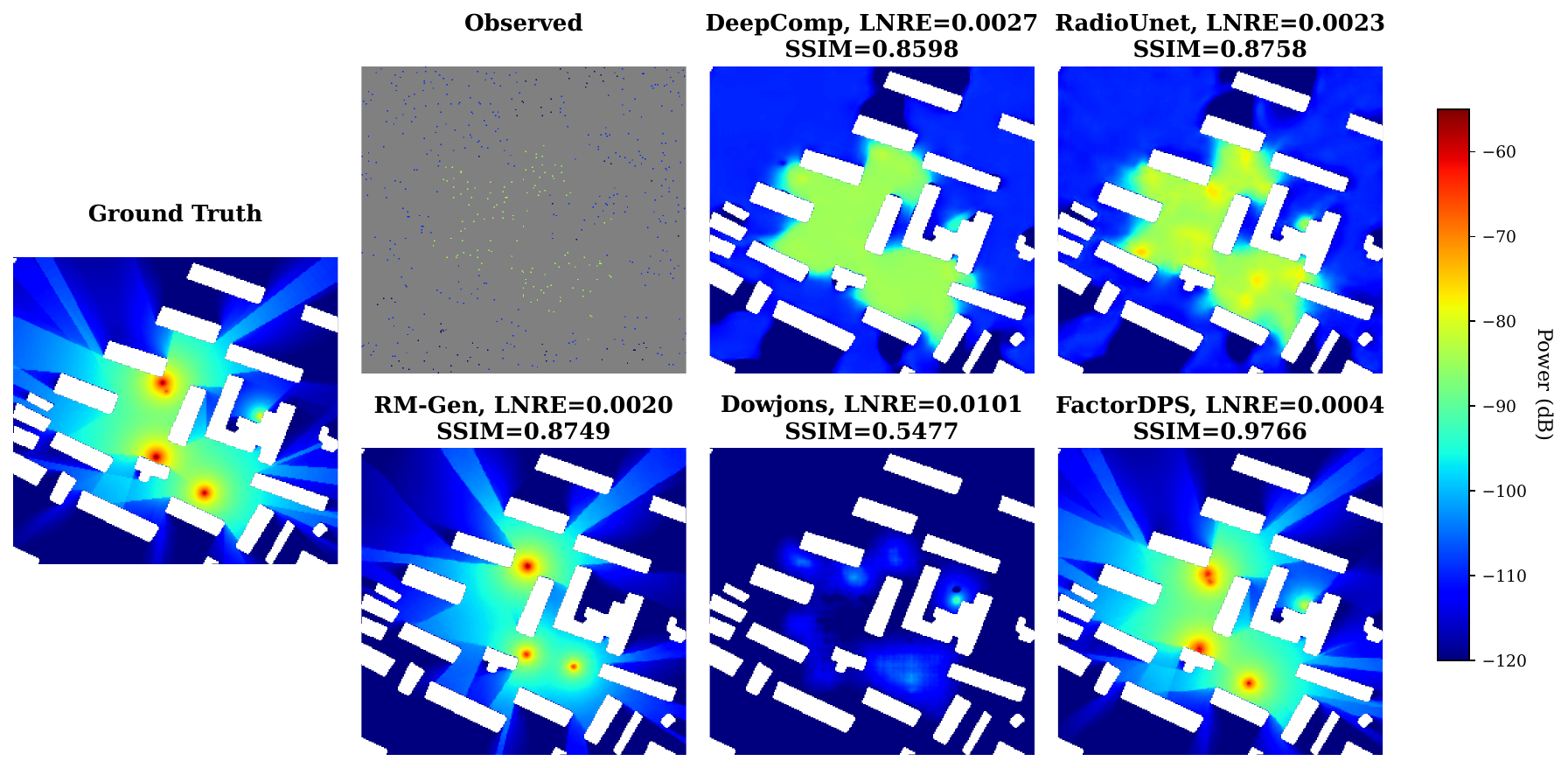}
        \caption{$2$-bit quantized measurement, $R=5$}
        \label{subfig:quantized}
    \end{subfigure}
    \caption{Recovered radio maps at the $16$th frequency bin from $1\%$ sampling ratio and different measurements, $K=64$.}
    \label{fig:distortion}
\end{figure}

Fig.~\ref{fig:distortion} presents visual reconstruction results under two challenging settings: an additional $40\% \times 40\%$ missing central region (Fig.~\ref{subfig:centralmissing}) and $2$-bit quantized measurements (Fig.~\ref{subfig:quantized}), on top of $1\%$ random sampling.
Note that most of the baselines are not designed for quantized measurements and we apply them as if the measurements are continuous. An exception is {\tt DowJons}, which uses the quantized sensing objective following the work \cite{Subash2024Quantized}. As one can see, only the proposed {\tt FactorDPS} achieves reasonable and accurate reconstructions, demonstrating its flexibility in handling changes in the observation model.

Table~\ref{table:exp9_hole_meanstd_without} and Table~\ref{table:exp10_bit_meanstd_without} further report the quantitative performance under the above settings. 
Consistent with the visual results, {\tt FactorDPS} consistently outperforms all baseline methods by a large margin in terms of both LNRE and MSSIM.

\begin{table}[t!]
\caption{Performance under different size of forbidden regions ($R=5$, $\tau=1\%$).}
\footnotesize
\centering
\scalebox{0.72}{
\begin{tabular}{ |c | c c c c c |}
\hline
 & \multicolumn{5}{c|}{LNRE ($\times 10^{-4}$)}\\
Ratio & \texttt{DeepComp} & \texttt{RadioUnet} & \texttt{RM-Gen} & \texttt{Dowjons} & \texttt{FactorDPS} \\
 \hline
 20\% &    6.24 $\pm$ 2.01 &    6.64 $\pm$ 2.51 &    17.85 $\pm$ 3.43 &    79.48 $\pm$ 34.63 &    \textbf{5.54} $\pm$ 3.68 \\
 40\% &    21.19 $\pm$ 7.48 &    85.49 $\pm$ 46.33 &    42.72 $\pm$ 9.70 &    192.00 $\pm$ 93.75 &    \textbf{10.10} $\pm$ 6.79 \\
 60\% &    49.21 $\pm$ 12.40 &    237.62 $\pm$ 97.71 &    173.40 $\pm$ 94.69 &    565.66 $\pm$ 235.17 &    \textbf{12.24} $\pm$ 6.98 \\
 80\% &    78.00 $\pm$ 16.78 &    435.52 $\pm$ 178.64 &    435.53 $\pm$ 187.52&    837.42 $\pm$ 315.12 &    \textbf{18.37} $\pm$ 11.62 \\
\hline
\hline
 & \multicolumn{5}{c|}{MSSIM}\\
Ratio & \texttt{DeepComp} & \texttt{RadioUnet} & \texttt{RM-Gen} & \texttt{Dowjons} & \texttt{FactorDPS} \\
 \hline
 20\% &    0.777 $\pm$ 0.076 &    0.768 $\pm$ 0.078 &    {0.899} $\pm$ 0.034  &    0.662 $\pm$ 0.088 &    \textbf{0.946} $\pm$ 0.023 \\
 40\% &    0.755 $\pm$ 0.079 &    0.717 $\pm$ 0.066 &    {0.860} $\pm$ 0.043 &    0.567 $\pm$ 0.094 &    \textbf{0.935} $\pm$ 0.029 \\
 60\% &    0.729 $\pm$ 0.077 &    0.602 $\pm$ 0.041 &    {0.749} $\pm$ 0.050 &    0.365 $\pm$ 0.066 &    \textbf{0.924} $\pm$ 0.033 \\
 80\% &    0.702 $\pm$ 0.091 &    0.468 $\pm$ 0.045 &    {0.531} $\pm$ 0.053 &    0.181 $\pm$ 0.035 &    \textbf{0.913} $\pm$ 0.034 \\
\hline
\end{tabular}}
\label{table:exp9_hole_meanstd_without}
\end{table}

\begin{table}[t!]
\caption{Performance under different quantization bits ($R=5$, $\tau=1\%$)}
\footnotesize
\centering
\scalebox{0.72}{
\begin{tabular}{ |c | c c c c c |}
\hline
 & \multicolumn{5}{c|}{LNRE ($\times 10^{-4}$)}\\
Bits & \texttt{DeepComp} & \texttt{RadioUnet} & \texttt{RM-Gen} & \texttt{Dowjons} & \texttt{FactorDPS} \\
 \hline
 1 &    151.57 $\pm$ 62.89 &    133.19 $\pm$ 57.52 &    70.76 $\pm$ 29.97  &    329.11 $\pm$ 141.95 &    \textbf{15.78} $\pm$ 11.75 \\
 2 &    31.10 $\pm$ 5.49 &    28.18 $\pm$ 6.56 &    20.42 $\pm$ 2.45 &    76.14 $\pm$ 29.36 &    \textbf{7.06} $\pm$ 3.67 \\
 3 &    9.12 $\pm$ 1.57 &    6.93 $\pm$ 0.77 &    16.07 $\pm$ 1.42  &    61.08 $\pm$ 22.04 &    \textbf{6.62} $\pm$ 3.53 \\
 4 &    4.83 $\pm$ 1.61 &    5.79 $\pm$ 1.18 &    16.21 $\pm$ 1.43 &    64.43 $\pm$ 24.36 &    \textbf{3.66} $\pm$ 3.12 \\
 5 &    3.86 $\pm$ 1.73 &    4.00 $\pm$ 1.87  &    16.43 $\pm$ 2.03  &    65.74 $\pm$ 28.23 &    \textbf{3.29} $\pm$ 3.34 \\
\hline
\hline
 & \multicolumn{5}{c|}{MSSIM}\\
Bits & \texttt{DeepComp} & \texttt{RadioUnet} & \texttt{RM-Gen} & \texttt{Dowjons} & \texttt{FactorDPS} \\
 \hline
 1 &    0.636 $\pm$ 0.059 &    0.642 $\pm$ 0.054 &    0.849 $\pm$ 0.040 &    0.482 $\pm$ 0.055 &    \textbf{0.919} $\pm$ 0.038 \\
 2 &    0.689 $\pm$ 0.083 &    0.705 $\pm$ 0.077 &    0.874 $\pm$ 0.028 &    0.643 $\pm$ 0.086 &    \textbf{0.934} $\pm$ 0.032 \\
 3 &    0.734 $\pm$ 0.080 &    0.755 $\pm$ 0.077 &    0.899 $\pm$ 0.031 &    0.683 $\pm$ 0.087 &    \textbf{0.946} $\pm$ 0.021 \\
 4 &    0.756 $\pm$ 0.080 &    0.771 $\pm$ 0.078 &    0.903 $\pm$ 0.030 &    0.681 $\pm$ 0.099 &    \textbf{0.945} $\pm$ 0.024 \\
 5 &    0.771 $\pm$ 0.082 &    0.771 $\pm$ 0.083 &    0.903 $\pm$ 0.030 &    0.681 $\pm$ 0.094 &    \textbf{0.947} $\pm$ 0.024 \\
\hline
\end{tabular}}
\label{table:exp10_bit_meanstd_without}
\end{table}

\subsection{Ablation Study}

\subsubsection{Comparison with Data Domain DPS}
In this subsection, we examine the impact of incorporating the decomposition model (cf. \eqref{eq:X=SC}). To this end, we compare {\tt FactorDPS} with a data-domain diffusion baseline. In the training stage, the baseline directly learns single-band multi-emitter radio maps. During inference, the baseline applies DPS independently to each frequency (i.e., a frontal slice of the radio map tensor), and the reconstructed slices are then stacked to form the full radio map. We refer to this baseline as {\tt SliceDPS}. {\tt SliceDPS} adopts the same network architecture as the SLF diffusion model in {\tt FactorDPS}, and uses the same diffusion and training settings. The main difference lies in the training data: {\tt SliceDPS} directly learns multi-emitter radio maps generated as $\bm X=\sum_{r=1}^{R}c_r\bm S_{\ell_r,n}$, where $R$ is uniformly drawn from ${1,2,3,4}$, $c_r\sim\mathcal{U}(0,1)$, and $\bm S_{\ell_r,n}\sim\mathbb S_{\rm tr}^\natural$. In each epoch, 50,000 radio map samples are generated for training.

\begin{figure}[t]
    \centering
    \includegraphics[width=\linewidth]{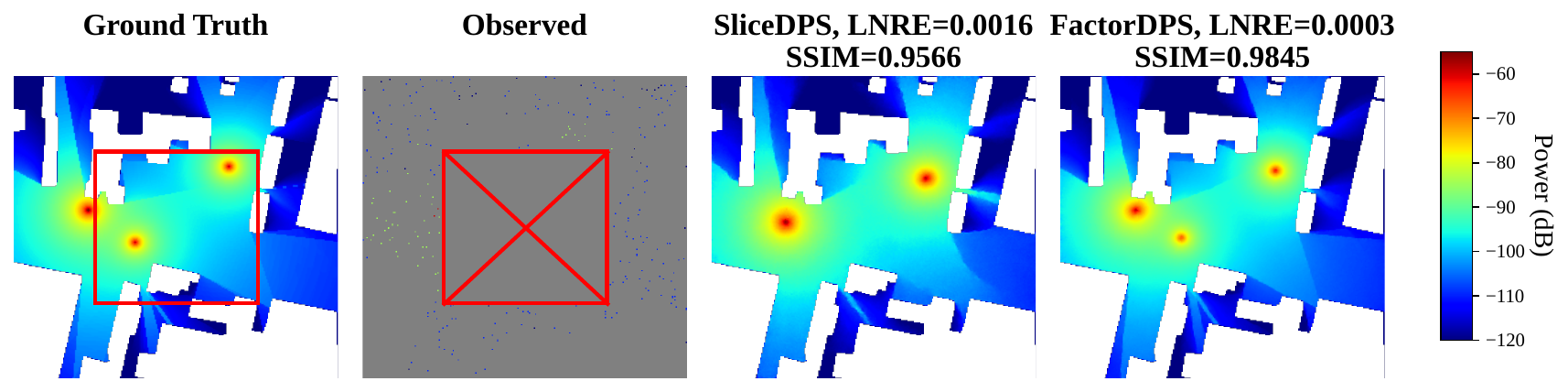}
    \caption{Recovered radio maps at the $16$th frequency bin from $1\%$ sampling ratio, 2-bits quantization, and an additional $50\%\times 50\%$ central square missing area ($R = 3$).}
    \label{fig:ablation}
\end{figure}

\begin{table}[t!]
\caption{Performance under different sampling rate $\tau$ ($R = 3$, 2-bits quantization, $50\%\times 50\%$ central square restricted region).}
\footnotesize
\centering
\centering
\scalebox{0.76}{
\begin{tabular}{ |c | c c | c c |}
\hline
 & \multicolumn{2}{c|}{LNRE $(\times 10^{-4})$} & \multicolumn{2}{c|}{MSSIM} \\
$\tau$ & \texttt{SliceDPS} & \texttt{FactorDPS} & \texttt{SliceDPS} & \texttt{FactorDPS}\\
 \hline
 1\% &    71.45 $\pm$ 34.34 &    \textbf{17.09} $\pm$ 10.76 &    0.7248 $\pm$ 0.0596 &    \textbf{0.8829} $\pm$ 0.0683 \\
 2\% &    39.63 $\pm$ 18.48 &    \textbf{14.09} $\pm$ 10.54 &    0.8280 $\pm$ 0.0841 &    \textbf{0.9083} $\pm$ 0.0548 \\
 4\% &    61.36 $\pm$ 16.55 &    \textbf{11.91} $\pm$ 7.54 &    0.7434 $\pm$ 0.0433 &    \textbf{0.9123} $\pm$ 0.0452 \\
 6\% &    62.49 $\pm$ 17.49 &    \textbf{11.37} $\pm$ 7.11 &    0.7951 $\pm$ 0.0349 &    \textbf{0.9150} $\pm$ 0.0441 \\
 8\% &    59.42 $\pm$ 21.63 &    \textbf{11.30} $\pm$ 7.44 &    0.8131 $\pm$ 0.0511 &    \textbf{0.9158} $\pm$ 0.0397 \\
 10\% &    58.12 $\pm$ 19.53 &    \textbf{13.33} $\pm$ 8.63 &    0.8472 $\pm$ 0.0741 &    \textbf{0.9151} $\pm$ 0.0515 \\
\hline
\hline
& \multicolumn{2}{c}{\texttt{SliceDPS}} & \multicolumn{2}{c|}{\texttt{FactorDPS}} \\
{Runtime (mins)} & \multicolumn{2}{c}{208.36} & \multicolumn{2}{c|}{\textbf{16.81}} \\
\hline
\end{tabular}}
\label{table:mr_multiK_bits2_hole50_meanstd_without}
\end{table}

\begin{figure}[t]
    \centering
    \includegraphics[width=0.488\linewidth]{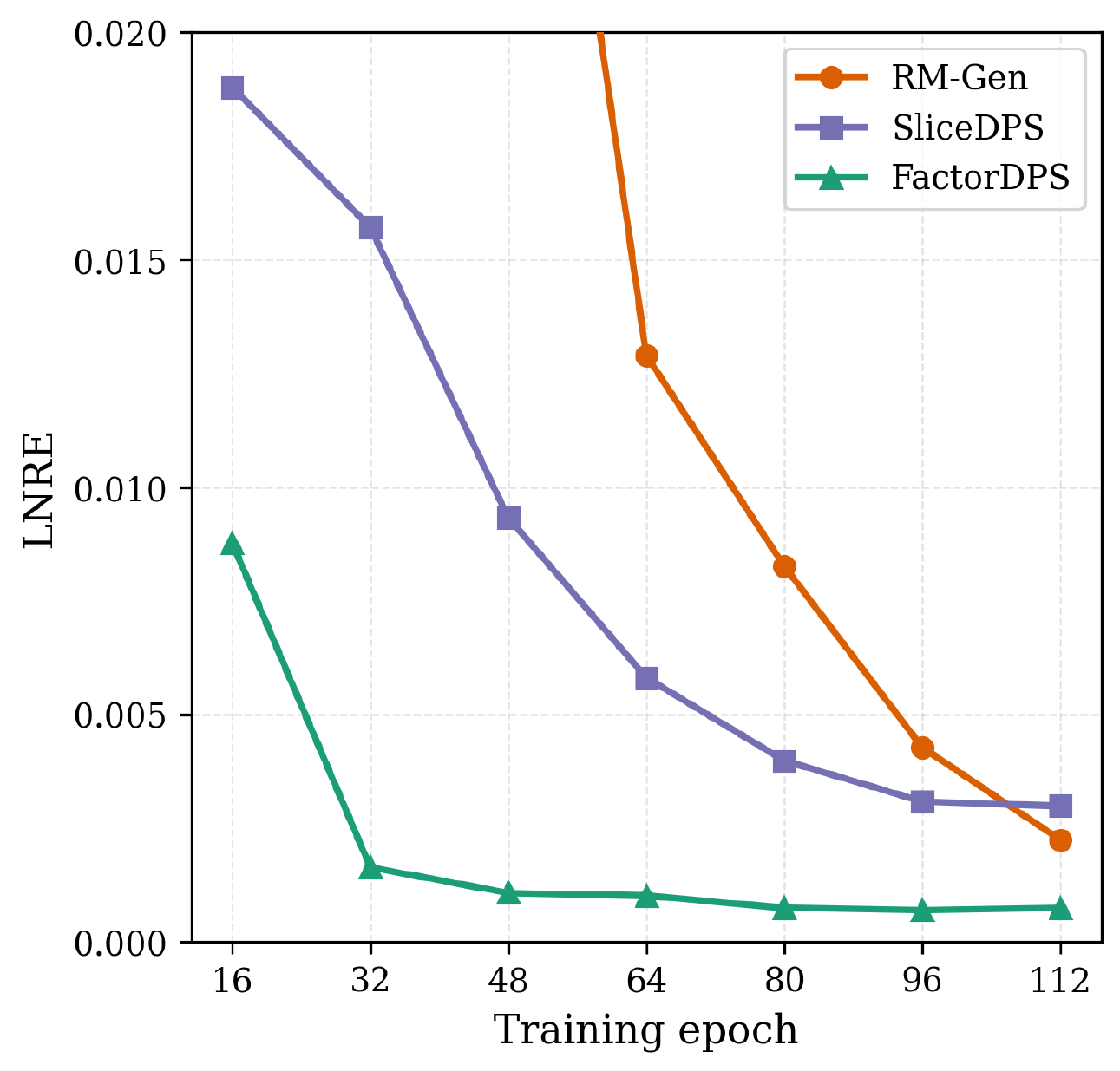}
    \hfill
    \includegraphics[width=0.47\linewidth]{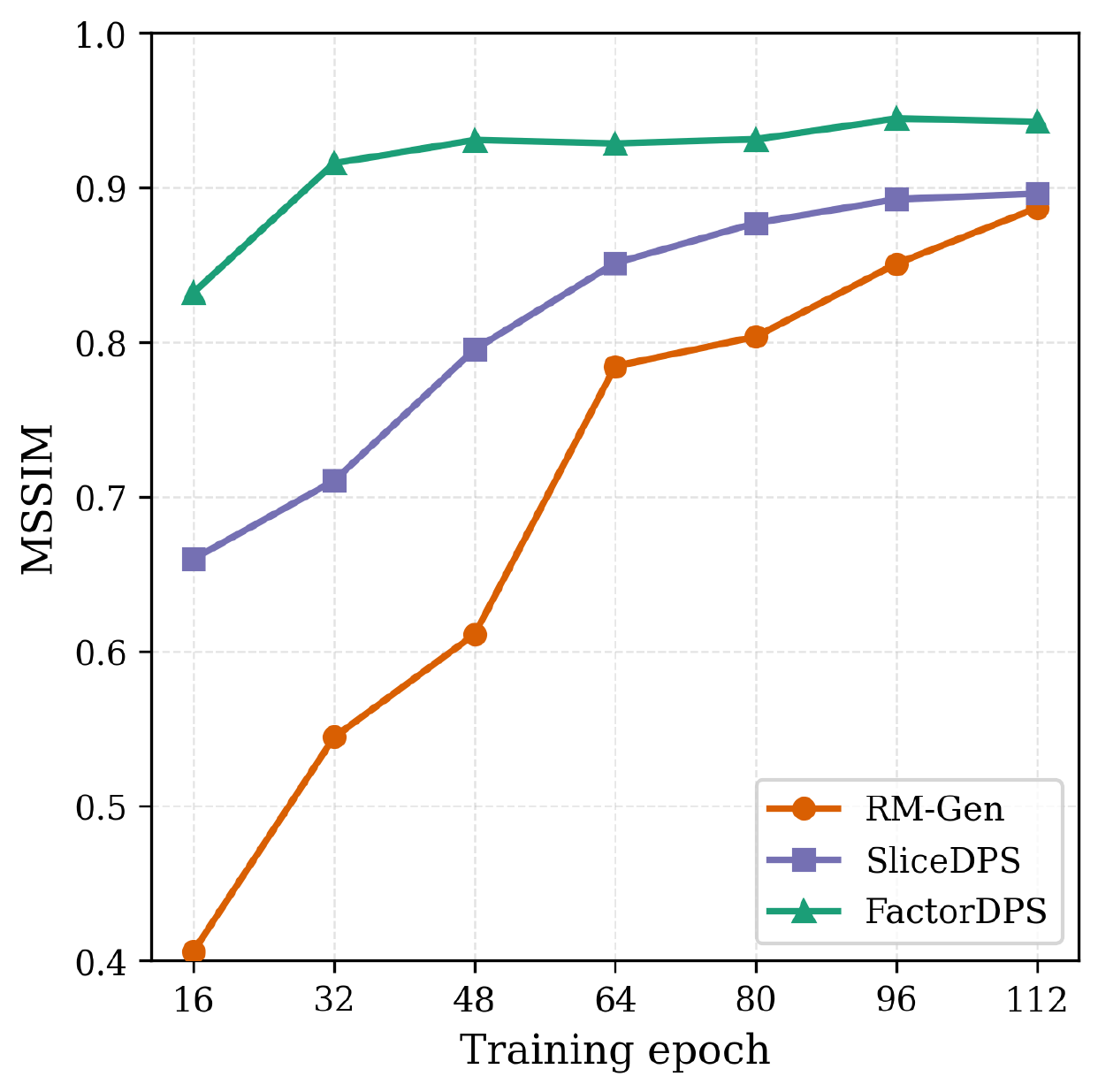}
    \caption{Performance versus the number of training epochs.}
    \label{fig:metric_epoch}
\end{figure}
\begin{figure}[t]
    \centering
    \includegraphics[width=0.98\linewidth]{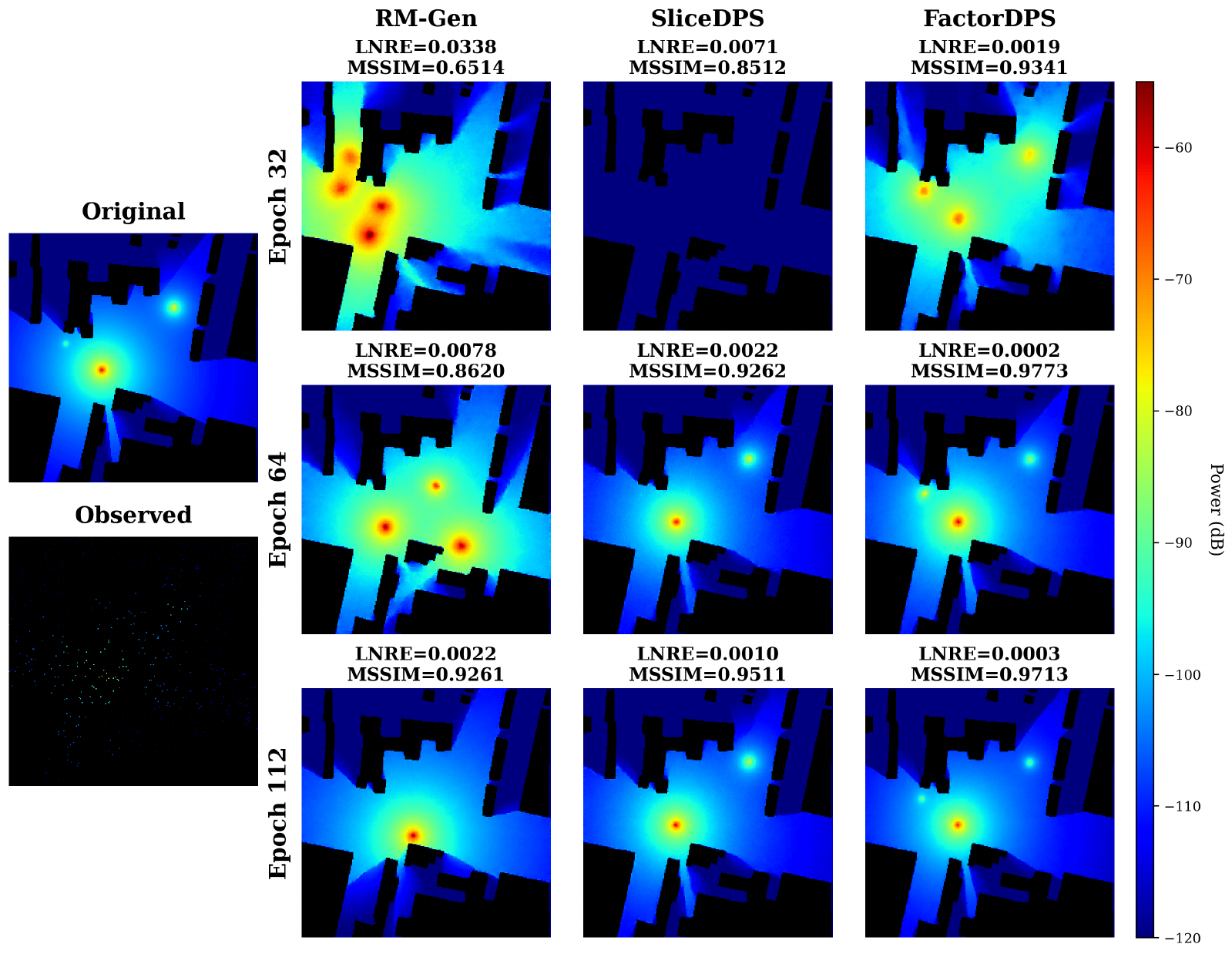}
    \caption{Recovered radio maps at the 16th frequency bin after 32, 64 and 112 training epochs, respectively.}
    \label{fig:visual_epoch}
\end{figure}

Fig.~\ref{fig:ablation} presents the reconstruction results at the $16$th frequency bin in a highly challenging scenario, including $1\%$ random sampling, an additional central missing region covering $50\% \times 50\%$ of the map, and $2$-bit quantization, with $K=64$ and $R=3$. 
In this case, {\tt FactorDPS} successfully recovers all emitter components with clearer spatial structures.
In contrast, {\tt SliceDPS} fails to reveal one of the emitters.

Table~\ref{table:mr_multiK_bits2_hole50_meanstd_without} reports the quantitative performance under the same degradation settings across different sampling ratios. 
One can see that {\tt FactorDPS} consistently outperforms {\tt SliceDPS} in terms of both LNRE and MSSIM. 
Moreover, {\tt FactorDPS} is significantly more efficient than {\tt SliceDPS}, as evidenced by the runtime comparison in Table~\ref{table:mr_multiK_bits2_hole50_meanstd_without}.
This is because {\tt FactorDPS} operates only on the decomposed spatial and spectral factors, while {\tt SliceDPS} processes all $64$ frequency bands.

\subsubsection{Performance against Training Epochs}
We next examine whether incorporating the factorization model improves training efficiency. We observe the performance of {\tt RM-Gen}, {\tt SliceDPS}, and {\tt FactorDPS} under different numbers of training epochs. Note that all three methods use essentially the same network architecture and diffusion settings, with the only difference being that {\tt RM-Gen} includes an additional input channel for the measurements. 
They also share the same training configuration, with 50,000 training samples per epoch, and are all trained on an H100 GPU.
The experimental setting follows the multi-emitter experiments in Sec.~\ref{subsec:multi-band-cases}, with $K=64$, $R=3$, and $\tau=1\%$. All reported metrics are averaged over 10 Monte Carlo runs.

Fig.~\ref{fig:metric_epoch} shows that {\tt FactorDPS} achieves strong reconstruction performance with substantially fewer training epochs, with its performance largely stabilizing after approximately 32 epochs. In contrast, {\tt RM-Gen} and {\tt SliceDPS} continue to improve with additional training, but neither reaches the performance achieved by {\tt FactorDPS} at 32 epochs within the tested range. The same trends are also visually illustrated in Fig.~\ref{fig:visual_epoch}.

\section{Conclusion}

This paper proposed a latent factor diffusion framework for radio map tensor completion, where independently trained diffusion models for spatial loss fields and spectral signatures are embedded into the latent factors to enable flexible and efficient posterior sampling–based recovery. Compared with existing data-domain diffusion approaches, the proposed method reduces training complexity and improves adaptability to varying measurement formats and emitter configurations. We further established recoverability guarantees under the latent factor diffusion model, providing theoretical insights into diffusion-based radio map estimation. 
Numerical results demonstrated strong performance of the proposed method across a wide range of scenarios: it generalizes well from single-emitter training to challenging multi-emitter cases, adapts to different measurement settings, and remains robust to environmental shifts, as demonstrated by the disjoint training and test cities.

In the current framework, the radio map is assumed to remain stationary over time, and the SLF of each emitter is assumed to be shared across frequency bands. These assumptions may become restrictive in dynamic environments or when spatial propagation varies significantly with frequency. Future work will therefore explore more general and dynamic wireless environments.
In addition, \texttt{FactorDPS} does not guarantee exact posterior sampling due to likelihood approximation. Therefore, promising directions also include improving posterior sampling accuracy with theoretical supports.

\medskip

\noindent
{\bf Source Code}:
The code of our algorithm is available at \hyperlink{https://github.com/xumaomao94/FactorDPS}{https://github.com/xumaomao94/FactorDPS}.

\medskip

\noindent
{\bf Generative AI Usage}:
Generative AI tools were used for coding assistance and language polishing. All contents have been carefully reviewed and verified by the authors.

\bibliographystyle{IEEEtran}
\bibliography{ref}


\clearpage
\appendices

{\bf Supplementary Material of ``Completing Radio Map Tensors via Factor Diffusion Posterior Sampling'' by L. Xu, R. Shrestha, and X. Fu}

\section{Proof of Fact~\ref{fact:covering_X}}
\label{apd:covering_X_proof}

From the definition of $\mathbb X$ in \eqref{eq:Xset}, every $\mathbm X\in\mathbb X$ admits latent factors satisfying
\[
|\bm S_r(i,j)|\le \|\bm S_r\|_{\rm F}\le\gamma_{\rm s},
\qquad
|\bm c_r(k)|\le \|\bm c_r\|_2\le\gamma_{\rm c}.
\]
Hence, by the triangle inequality,
\[
|\mathbm X(i,j,k)|
\le
\sum_{r=1}^R|\bm S_r(i,j)\bm c_r(k)|
\le
R\gamma_{\rm s}\gamma_{\rm c}.
\]
Since $\bm x=\vecop(\mathbm X)$, this implies
\begin{align}
\max_\ell|\bm x_\ell|
\le
R\gamma_{\rm s}\gamma_{\rm c}\triangleq\tau_{\rm x}.
\label{eq:max_x_ell}
\end{align}

The covering-number bound follows directly from existing results for factorized tensor models; see, e.g., \cite[Lemma~2]{Shrestha2022DeepSpectrum} and \cite[Lemma~3]{xu2025radio}. Applying these results to the decomposition $\mathbm X=\sum_{r=1}^R\bm S_r\circ\bm c_r$, together with the energy constraints $\|\bm S_r\|_{\rm F}\le\gamma_{\rm s}$ and $\|\bm c_r\|_2\le\gamma_{\rm c}$, yields the desired bound.

\section{Proof of Theorems~\ref{thm:perfect_case}--\ref{thm:imperfect_TV}}
\label{apd:thm_proof}
We first prove the more general imperfect-sampling result in Theorem~\ref{thm:imperfect_TV}. The perfect-sampling result in Theorem~\ref{thm:perfect_case} then follows as the special case where the learned posterior sampler exactly matches the true posterior sampler.

Our proof follows the posterior-sampling-based compressed sensing framework in~\cite{jalal2021instance,jalal2021robust}. In particular, like in~\cite{jalal2021instance}, we decompose the error probability over a finite cover of the signal class and relate each local error probability to distances between induced measurement distributions. As discussed in Remark~\ref{remark:ourproof}, our setting differs from~\cite{jalal2021instance,jalal2021robust} in several important ways, including latent-space priors, tensor fiber subsampling, and simpler characterizations of sampler imperfection.

\subsection{Main Proof}

Since the latent-domain samples are mapped to the radio-map domain through $\bm h(\cdot)$, sampling
$\bm z^\natural \sim p(\bm z)$ and $\hat{\bm z}\sim p’(\bm z\mid \bm y)$ and then forming
$\bm x^\natural=\bm h(\bm z^\natural)$ and $\hat{\bm x}=\bm h(\hat{\bm z})$
is equivalent to sampling directly from the induced radio-map distributions; i.e., 
$\bm x^\natural\sim p(\bm x)$
and
$\hat{\bm x}\sim p'(\bm x\mid \bm y)$, respectively.
Thus, we work directly in the radio-map domain.

By Fact~\ref{fact:covering_X}, every $\bm x \in \mathbb X$ satisfies $\| \bm x\|_\infty \leq \tau_{\rm x}$. If $\hat{\bm x}\in \mathbb G(\mathbb X,\kappa)$, then by the definition of
$\mathbb G(\mathbb X,\kappa)$ in \eqref{eq:G(X)}, there exists some
$\bm x\in\mathbb X$ such that
$\|\hat{\bm x}-\bm x\|_2\leq\kappa$. Therefore,
\begin{align}
    |\hat{\bm x}_\ell|
    \le |\bm x_\ell|+|\hat{\bm x}_\ell-\bm x_\ell|
    \le \tau_{\rm x}+\kappa, \quad \forall \ell\in[IJK].
    \label{eq:Gx_bound}
\end{align}
Moreover, by the data-processing inequality \cite{Raginsky2016strong} \cite[Fact~1.1]{canonne2022topics} \cite[Theorem~7.4]{polyanskiy2025information}, the TV mismatch in \eqref{eq:post_TV_mismatch_final} implies that
\begin{equation}
    \mathrm{TV}\!\left(p'(\bm x\mid\bm y),p(\bm x\mid\bm y)\right)
    \leq \mathrm{TV}\!\left(p'(\bm z\mid\bm y),p(\bm z\mid\bm y)\right)
    \leq \varepsilon_{\rm TV}.
    \label{eq:TV_x_bound_imperfect}
\end{equation}

We aim to bound
\begin{align}
    {\rm Pr} ({\|\bm{x}^\natural - \hat{\bm{x}}\|_2} \ge {\sqrt{IJK}}\eta). \label{eq:Pr_error}
\end{align}
We start with a union bound to decompose such probability into probabilities over small partitions. In particular,
let $\{\bm b_i\}_{i=1}^H$ be a $\sqrt{IJK}\epsilon$-net of $\mathbb X$, where $H\leq \mathrm N(\mathbb X,\sqrt{IJK}\epsilon)$.
Let $\{\mathbb V_i\}_{i=1}^H$ be a measurable partition of $\mathbb X$ associated with this net, defined as
\begin{align}
\mathbb V_i
\triangleq
\Big\{
\bm x\in\mathbb X:\;
\|\bm x-\bm b_i\|_2
=
\min_{j\in[H]}\|\bm x-\bm b_j\|_2
\Big\}.\label{eq:V_i}
\end{align}
Then the following holds:
\begin{align}
    \mathbb V_i &\subseteq \mathbb B_i^{\rm in}
    \triangleq
    \left\{
    \bm x\in\mathbb X:
    {\|\bm x-\bm b_i\|_2}\leq {\sqrt{IJK}}\epsilon
    \right\}, \\
    \sum_{i=1}^H& {\rm Pr}(\bm x^\natural\in\mathbb V_i)=1.\label{eq:partition_sum_to1}
\end{align}
Therefore, \eqref{eq:Pr_error} can be expressed as
\begin{align}
    &{\rm Pr}\!\left(
    {\|\bm x^\natural-\hat{\bm x}\|_2}\geq {\sqrt{IJK}}\eta
    \right)
    =
    \sum_{i=1}^H
    {\rm Pr}(\bm x^\natural\in\mathbb V_i) \nonumber \\
    & \qquad \qquad \times{\rm Pr}\!\left(
    {\|\bm x^\natural-\hat{\bm x}\|_2}\geq {\sqrt{IJK}}\eta
    \mid
    \bm x^\natural\in\mathbb V_i
    \right).
    \label{eq:failure_rate_accu}
\end{align}

Next, we follow \cite{jalal2021instance} to convert the reconstruction error in \eqref{eq:failure_rate_accu} to an ``out-of-ball'' event. 
Suppose $\bm x^\natural\in\mathbb V_i$, then $\|\bm x^\natural-\bm b_i\|_2/\sqrt{IJK}\le\epsilon$.
By triangle inequality, $\|\bm x^\natural-\hat{\bm x}\|_2\ge\sqrt{IJK}\eta$ implies
\[
\|\hat{\bm x}-\bm b_i\|_2
\ge
\|\bm x^\natural-\hat{\bm x}\|_2-\|\bm x^\natural-\bm b_i\|_2
\ge
\sqrt{IJK}(\eta-\epsilon).
\]
Consequently, by defining
\[
    \mathbb B_i^{\rm out}
    \triangleq
    \left\{
    \bm x:
    {\|\bm x-\bm b_i\|_2}\geq {\sqrt{IJK}}(\eta-\epsilon)
    \right\},
\]
\eqref{eq:failure_rate_accu} can be bounded by
\begin{align}
    &{\rm Pr}\!\left(
    {\|\bm x^\natural-\hat{\bm x}\|_2}\geq {\sqrt{IJK}}\eta
    \right) \nonumber \\
    \leq&
    \sum_{i=1}^H
    {\rm Pr}(\bm x^\natural\in\mathbb V_i)
    {\rm Pr}\!\left(
    \hat{\bm x}\in\mathbb B_i^{\rm out}
    \mid
    \bm x^\natural\in\mathbb V_i
    \right).
    \label{eq:imperfect_partition_bound}
\end{align}
Considering the observation model in \eqref{eq: IP for RM}, i.e., $\bm y = \bm A \bm x^\natural + \bm w$ with $\bm A =(\bm M_{\bm \varOmega} \otimes \bm I_K  )$, the error probability in \eqref{eq:imperfect_partition_bound} can be explicitly written as 
\begin{align}
    &{\rm Pr}\!\left(
    \hat{\bm x}\in\mathbb B_i^{\rm out}
    \mid
    \bm x^\natural\in\mathbb V_i
    \right)\nonumber\\
    =&\mathbb E_{\bm A}
    \mathbb E_{\bm w}
    \mathbb E_{\bm x^\natural \in\mathbb V_i} [{\rm Pr}(\hat{\bm x}\in\mathbb B_i^{\rm out}\mid \bm y)],
    \label{eq:imperfect_conditional_out}
\end{align}
in which, for simplicity, we use ${\rm Pr}(\hat{\bm x}\in\mathbb B_i^{\rm out}\mid \bm y)$ to denote ${\rm Pr}(\hat{\bm x}\in\mathbb B_i^{\rm out}\mid \bm y=\bm A \bm x^\natural + \bm w)$.
By \eqref{eq:TV_x_bound_imperfect} and the definition of TV distance, we further have
\begin{align}
    &\mathbb E_{\bm x^\natural \in\mathbb V_i} [{\rm Pr}(\hat{\bm x}\in\mathbb B_i^{\rm out}\mid \bm y)] \nonumber\\
    = &
    {\rm Pr}(\hat{\bm x}\in\mathbb B_i^{\rm out}\mid\bm x^\natural \in\mathbb V_i,\bm y) \nonumber\\
    \leq &
    {\rm Pr}(\tilde{\bm x}\in\mathbb B_i^{\rm out}\mid\bm x^\natural \in\mathbb V_i,\bm y)
    +
    \varepsilon_{\rm TV}\nonumber\\
    = &
    \mathbb E_{\bm x^\natural \in\mathbb V_i} [{\rm Pr}(\tilde{\bm x}\in\mathbb B_i^{\rm out}\mid \bm y)] +\varepsilon_{\rm TV}
    \label{eq:imperfect_TV_event}
\end{align}
where $\tilde{\bm x} \sim p(\bm x|\bm y)$ denotes a posterior sample from the perfect sampler.
Combining \eqref{eq:imperfect_partition_bound}-\eqref{eq:imperfect_TV_event} yields
\begin{align}
    &{\rm Pr}\!\left(
    {\|\bm x^\natural-\hat{\bm x}\|_2}\geq {\sqrt{IJK}}\eta
    \right) \leq \sum_{i=1}^H 
    \mathbb E_{\bm A}\Big[
    \nonumber\\
    &
    {\rm Pr}(\bm x^\natural\in\mathbb V_i)
    \mathbb E_{\bm w}
    \mathbb E_{\bm x^\natural \in\mathbb V_i} [{\rm Pr}(\tilde{\bm x}\in\mathbb B_i^{\rm out}\mid \bm y)]\Big] + \varepsilon_{\rm TV},
    \label{eq:imperfect_before_core_bound}
\end{align}
where $\varepsilon_{\rm TV}$ is extracted outside the summation due to \eqref{eq:partition_sum_to1}.

It remains to bound the first term in \eqref{eq:imperfect_before_core_bound}. Fixing $\bm A$, let $p_{\bm y,i}^{\rm in}$ be the distribution of $\bm y=\bm A\bm x+\bm w$ induced by $\bm x\sim p(\bm x\mid \mathbb V_i)$, and let $p_{\bm y,i}^{\rm out}$ be the distribution induced by $\bm x\sim p(\bm x\mid \mathbb B_i^{\rm out})$.
Then the following holds
\begin{align}
    &\mathbb E_{\bm w}
    \mathbb E_{\bm x^\natural \in\mathbb V_i} [{\rm Pr}(\tilde{\bm x}\in\mathbb B_i^{\rm out}\mid \bm y)] \nonumber\\
    =& \mathbb E_{\bm y\sim p_{\bm y,i}^{\rm in}} [{\rm Pr}(\tilde{\bm x}\in\mathbb B_i^{\rm out}\mid\bm y)]\nonumber\\
    =&\int  p_{{\bm y},i}^{\rm in}(\bm y) \frac{{\rm Pr}(\tilde{\bm x}\in\mathbb{B}_{i}^{\rm out}) p(\bm y\mid \tilde{\bm x}\in  \mathbb{B}_{i}^{\rm out}) }{p(\bm y)} d\bm y\nonumber \\
    =&\int  p_{{\bm y},i}^{\rm in}(\bm y) \frac{{\rm Pr}(\tilde{\bm x}\in\mathbb{B}_{i}^{\rm out}) p_{\bm y,i}^{\rm out}(\bm y) }{p(\bm y)} d\bm y.
    \label{eq:E_w_Pr}
\end{align}
The first equality follows because drawing $\bm x^\natural$ and $\bm w$ and then forming $\bm y=\bm A\bm x^\natural+\bm w$ is equivalent to drawing $\bm y\sim p_{\bm y,i}^{\rm in}$. The second equality follows from Bayes' rule, and the last equality follows from the definition of $p_{\bm y,i}^{\rm out}$.
When $\eta > 2\epsilon$, $\mathbb V_i$ and $\mathbb B_i^{\rm out}$ are disjoint. Hence
\begin{align}
    &{p(\bm y)} 
    \geq {\rm Pr}(\tilde{\bm x}\in\mathbb{V}_{i})p_{{\bm y},i}^{\rm in}(\bm y)
    +
    {\rm Pr}(\tilde{\bm x}\in\mathbb{B}_{i}^{\rm out})p_{{\bm y},i}^{\rm out}(\bm y) \nonumber \\
    &\geq \max \{ {\rm Pr}(\tilde{\bm x}\in\mathbb{V}_{i})p_{{\bm y},i}^{\rm in}(\bm y), 
    \ 
    {\rm Pr}(\tilde{\bm x}\in\mathbb{B}_{i}^{\rm out})p_{{\bm y},i}^{\rm out}(\bm y) \}.
    \label{eq:p_y_bound}
\end{align}
Combining \eqref{eq:E_w_Pr} and \eqref{eq:p_y_bound}, the term inside $\mathbb E_{\bm A}[\cdot]$ in \eqref{eq:imperfect_before_core_bound} can be bounded by
\begin{align}
    &{\rm Pr}(\bm x^\natural\in\mathbb V_i)
    \mathbb E_{\bm w}
    \mathbb E_{\bm x^\natural \in\mathbb V_i} [{\rm Pr}(\tilde{\bm x}\in\mathbb B_i^{\rm out}\mid \bm y)] \nonumber \\
    \leq& 
    \int 
    \frac{{\rm Pr}(\bm x^\natural\in\mathbb V_i)p_{{\bm y},i}^{\rm in}(\bm y)\times {\rm Pr}(\tilde{\bm x}\in\mathbb{B}_{i}^{\rm out})p_{{\bm y},i}^{\rm out}(\bm y) }{\max \{ {\rm Pr}(\tilde{\bm x}\in\mathbb{V}_{i})p_{{\bm y},i}^{\rm in}(\bm y), 
    \ 
    {\rm Pr}(\tilde{\bm x}\in\mathbb{B}_{i}^{\rm out})p_{{\bm y},i}^{\rm out}(\bm y) \}} d\bm y \nonumber \\
    = & \int {\min \{ {\rm Pr}(\tilde{\bm x}\in\mathbb{B}_{i}^{\rm out})p_{{\bm y},i}^{\rm out}(\bm y) ,\  {\rm Pr}(\tilde{\bm x}\in\mathbb{B}_{i}^{\rm in})p_{{\bm y},i}^{\rm in}(\bm y) \}} d\bm y \nonumber \\
    \leq & \int \min \{p_{{\bm y},i}^{\rm in}(\bm y), \; p_{{\bm y},i}^{\rm out}(\bm y) \} d\bm y \nonumber \\
    =& 1-{\mathrm TV}(p_{{\bm y},i}^{\rm in}, \; p_{{\bm y},i}^{\rm out}).
    \label{eq:one_ball_bound}
\end{align}
The third line holds because
$
    {\rm Pr}(\bm x^\natural\in\mathbb V_i)
    =
    {\rm Pr}(\tilde{\bm x}\in\mathbb V_i)
$, as both $\bm x^\natural$ and $\tilde{\bm x}$ are sampled from the ground-truth prior distribution.
Substituting \eqref{eq:one_ball_bound} into \eqref{eq:imperfect_before_core_bound} gives
\begin{align}
    &{\rm Pr}\!\left(
    {\|\bm x^\natural-\hat{\bm x}\|_2}\geq {\sqrt{IJK}}\eta
    \right) 
    \nonumber\\
    \leq &
    \sum_{i=1}^H 
    \mathbb E_{\bm A}\left[ 1-{\mathrm TV}(p_{{\bm y},i}^{\rm in}, \; p_{{\bm y},i}^{\rm out})\right] + \varepsilon_{\rm TV}.
    \label{eq:NxTVbound}
\end{align}

We now apply Lemma~\ref{lemma:TVsampling}. The lemma only requires the two signal distributions to be supported on the corresponding inner and outer regions with a common element-wise amplitude bound. In the present case, the inner distribution is supported on $\mathbb V_i\subseteq\mathbb B_i^{\rm in}$, while the outside distribution is supported on $\mathbb B_i^{\rm out}\cap\mathbb G(\mathbb X,\kappa)$. By \eqref{eq:Gx_bound}, every point in the enlarged set satisfies $\|\bm x\|_\infty\le \tau_{\rm x}+\kappa$. Hence Lemma~\ref{lemma:TVsampling} applies with $\gamma_{\rm x}\triangleq 2(\tau_{\rm x}+\kappa)$.
Set $\eta=\epsilon_{\rm x}$ and $\epsilon=\nicefrac{\epsilon_{\rm x}}{4}$, then the ratio between the radii of the ``out" and ``in" regions is $c= 3$, and $\nu=\nicefrac{\epsilon_{\rm x}^2}{4}$. Applying Lemma~\ref{lemma:TVsampling} gives, for any $\lambda\in(0,1)$ and $0<\epsilon_{\rm x}<\nicefrac{4}{3}(\tau_x+\kappa)$,
\begin{align}\label{eq:TV_bound_radioMap}
&\mathbb E_{\bm A}\!\left[1-\mathrm{TV}(p_{{\bm y},i}^{\rm in},p_{{\bm y},i}^{\rm out})\right]
< 4\exp\!\bigg(
-\frac{2N(\lambda\nu)^2}{\big(1-\frac{N-1}{IJ}\big)\gamma_{\rm x}^4}
\bigg)\nonumber\\
&
\qquad  +4\exp\!\left(
-\frac{NK(1-\lambda)^2\nu^2}{4\sigma^2\big(\sigma^2+2\gamma_{\rm x}^2+(1-\lambda)\nu\big)}
\right).
\end{align}
Taking $\lambda=1/2$, it follows that, for any $\delta\in(0,1)$, if
\begin{align} \label{eq:sampling_complexity_apd}
    N \geq \max&\left\{ \frac{32\gamma_{\rm x}^4}{\epsilon_{\rm x}^4},\frac{32\sigma^2(\epsilon_{\rm x}^2+8\sigma^2+16\gamma_{\rm x}^2)}{K\epsilon_{\rm x}^4} \right\} \nonumber\\
    &\times\log(\frac{8}{\delta} \mathrm{N}\big(\mathbb X,\sqrt{IJK}\nicefrac{\epsilon_{\rm x}}{4})\big), 
\end{align}
then
\begin{align} \label{eq:TV_bound_complexity}
&\mathbb E_{\bm A}\!\left[1-\mathrm{TV}(p_{{\bm y},i}^{\rm in},p_{{\bm y},i}^{\rm out})\right]
< \frac{\delta}{{\rm N}\big(\mathbb X,\sqrt{IJK}\nicefrac{\epsilon_{\rm x}}{4})}.
\end{align}
Substituting this into \eqref{eq:NxTVbound} gives
\begin{align}
    {\rm Pr}\left(\frac{\| \bm x^\natural - \hat{\bm x} \|_2 }{\sqrt{IJK}}\geq \epsilon_{\rm x}\right) < \delta + \varepsilon_{\rm TV}.
\end{align}
This completes the proof.

Theorem~\ref{thm:perfect_case} follows immediately from Theorem~\ref{thm:imperfect_TV} by setting
$\kappa=0$
and 
$\varepsilon_{\rm TV}=0$.

\subsection{Total Variation Bound under Fiber Sampling}
We provide the total-variation bound needed in \eqref{eq:NxTVbound}. For a fixed reference tensor $\underline{\bm x}$, we compare the measurement distributions induced by signals near $\underline{\bm x}$ and by signals separated from $\underline{\bm x}$. The lemma shows that these two distributions remain distinguishable after fiber sampling and additive Gaussian noise, similar to the purpose of \cite[Lemma~3.2]{jalal2021instance}.
Nonetheless,
unlike \cite[Lemma~3.2]{jalal2021instance} that considers Gaussian sensing matrices and vector signals, our result is
tailored to the structured fiber-sampling operator $\bm A=\bm M_{\bm \varOmega}\otimes \bm I_K$ and our latent factorization model.

\begin{lemma}
\label{lemma:TVsampling}
Let $\gamma>0$ and let $\underline{\bm x}\in[-\gamma_{\rm x}/2,\gamma_{\rm x}/2]^{IJK}$ be a fixed reference point. Define the local neighborhood and the separated region around $\underline{\bm x}$ as
\begin{align}
    &{\mathbb B}_{\rm in} \triangleq \{ \bm{x} \in [-\gamma_{\rm x}/2,\gamma_{\rm x}/2]^{IJK}: \frac{\| \bm x - \underline{\bm x}\|_2}{\sqrt{IJK}} \leq \epsilon \},\nonumber \\
    &{\mathbb B}_{\rm out} \triangleq \{ \bm x \in [-\gamma_{\rm x}/2,\gamma_{\rm x}/2]^{IJK}: \frac{\| \bm x - \underline{\bm x}\|_2}{\sqrt{IJK}} \geq c\epsilon \}, \nonumber
\end{align}
where $c>1$, $\epsilon>0$ and $\gamma_{\rm x} >  2c\epsilon$. Let $p(\bm x\mid \mathbb B_{\rm in})$ and $p(\bm x \mid \mathbb B_{\rm out})$ denote the distributions supported on $\mathbb B_{\rm in}$ and $\mathbb B_{\rm out}$, respectively. For a fixed realization of $\bm A=\bm M_{\bm \varOmega}\otimes \bm I_K\in\mathbb R^{NK\times IJK}$ as in \eqref{eq: IP for RM}, let $p_{\bm y}^{\rm in}$ and $p_{\bm y}^{\rm out}$ denote the distributions of $\bm y=\bm A\bm x+\bm w$, induced by $\bm x \sim p(\bm x\mid \mathbb B_{\rm in})$ and $\bm x \sim p(\bm x \mid \mathbb B_{\rm out})$, respectively. Then with $\nu=(c^2-1)\epsilon^2/2$, for any $\lambda\in(0,1)$,
\begin{align} \label{eq:TV-bound}
&\mathrm E_{\bm A}[{\rm TV}(p_{\bm y}^{\rm in},p_{\bm y}^{\rm out})]
\ge 1
-4\exp\!\bigg(
-\frac{2N(\lambda\nu)^2}{\big(1-\frac{N-1}{IJ}\big)\gamma_{\rm x}^4}
\bigg)\nonumber\\
&\qquad\qquad  - 4\exp\left(  -\frac{NK(1-\lambda)^2\nu^2}{4\sigma^2\big(\sigma^2+2\gamma_{\rm x}^2+(1-\lambda)\nu\big)} \right).
\end{align}
\end{lemma}

\begin{IEEEproof}
Let ${P}_{\bm y}^{\rm in}(\cdot)$ and ${P}_{\bm y}^{\rm out}(\cdot)$ denote probabilities under $\bm y\sim p_{\bm y}^{\rm in}$ and $\bm y\sim p_{\bm y}^{\rm out}$, respectively.
    By Definition \ref{def:tv-distance},
    \begin{align}
        {\rm TV}(p_{\bm y}^{\rm in}, p_{\bm y}^{\rm out})
        =
        \sup_{{\mathbb Y} \subset \mathbb R^{NK}}
        \big| P_{\bm y}^{\rm in}({\mathbb Y}) - P_{\bm y}^{\rm out}({\mathbb Y}) \big|.
    \end{align}
    Hence, for any measurable set $\mathbb Y\subseteq\mathbb R^{NK}$,
    \begin{align}\label{eq:Ea_bound}
        &\mathbb E_{\bm A} [{\rm TV}(p_{\bm y}^{\rm in}, p_{\bm y}^{\rm out})] 
         \geq \mathbb E_{\bm A} [ | P_{\bm y}^{\rm in}(\mathbb Y) -  P_{\bm y}^{\rm out}(\mathbb Y) | ] \nonumber \\
        & = 1 - \mathbb E_{\bm A} [  P_{\bm y}^{\rm in}(\mathbb Y^\complement)] - \mathbb E_{\bm A} [P_{\bm y}^{\rm out}(\mathbb Y) ],
    \end{align}
    in which $\mathbb Y^\complement$ denotes the complement of $\mathbb Y$.

    In the following of this proof, we will show that \eqref{eq:TV-bound} can be derived by taking $\mathbb Y = \mathbb B_{\rm y}$, which is constructed as
    \begin{align}
        \mathbb B_{\rm y} \triangleq \Big\{ \bm y \in \mathbb R^{NK}:\frac{\| \bm y - \bm A\underline{\bm x}\|_2^2}{NK}  \leq \frac{c^2+1}{2}\epsilon^2+\sigma^2\Big\}.
    \end{align}
    Since $\bm y=\bm A\bm x+\bm w$, the event $\bm y\in\mathbb B_{\rm y}$ is equivalent to $(\bm x,\bm w)\in\mathbb B_{\rm x,w}$, where
    \begin{align}
        &\mathbb B_{\rm x,w} \triangleq \Big\{ (\bm x,\bm w) \in \mathbb R^{IJK} \times \mathbb R^{NK}:\nonumber\\
        &\qquad \| \bm A \bm x - \bm A\underline{\bm x} +\bm w\|_2^2 \leq {NK\Big(\frac{c^2+1}{2}\epsilon^2}+\sigma^2\Big)\Big\}.
    \end{align}
    Denote $\bm e = \bm A \bm x - \bm A\underline{\bm x} +\bm w$ and $\bm d = \bm A \bm x - \bm A\underline{\bm x}$. Then we have
    \begin{align}
        \| \bm e \|_2^2 = \| \bm d \|_2^2 + \| \bm w\|_2^2 + 2\bm d^\top \bm w.
    \end{align}

    Before proceeding, we introduce two inequalities on the components $\| \bm d \|_2^2$ and $\| \bm w\|_2^2 + 2\bm d^\top \bm w$, respectively.

    \paragraph{Inequality on $\| \bm d\|_2^2$}
    Define the block format of $\bm x\in \mathbb R^{IJK}$ as
    \begin{align}
        \bm X = [\bm x(1:K),\ldots,\bm x(IJK-K+1:IJK)]^\top \in \mathbb R^{IJ\times K}. \nonumber
    \end{align}
    and define $\underline{\bm X}$ in the same manner. Since
    $\bm A=(\bm M_{\bm \varOmega}\otimes \bm I_K)$, $\bm d = \bm A\bm x-\bm A\underline{\bm x}$ in its block form is
    \begin{align}
        \bm D = \bm M_{\bm \varOmega} (\bm X-\underline{\bm X}) \in \mathbb R^{N\times K}.
    \end{align}
    
    Notice that
    \[
    0 \leq \| \bm X(i,:) -\underline{\bm X}(i,:) \|_{\rm F}^2/K \leq\gamma_{\rm x}^2, \quad \forall i\in[IJ],
    \]
    therefore, by Serfling's inequality for sampling without replacement, the following holds for $\| \bm d \|_2$:
    \begin{align}\label{eq:Hoeffding_block}
        &{\rm Pr}\!\left(\left|
        \frac{\|\bm d\|_2^2}{NK}
        -\frac{\|\bm x-\underline{\bm x}\|_2^2}{IJK}
        \right|
        \ge \lambda\nu\right)\nonumber\\
        &={\rm Pr} \left( \left|  \frac{\| \bm D \|_{\rm F}^2}{NK} - \frac{\| \bm X -\underline{\bm X} \|_{\rm F}^2}{IJK} \right| \geq \lambda\nu\right)\nonumber\\
        &={\rm Pr} \left( \left|  \frac{\| \bm M_{\bm \varOmega} (\bm X -\underline{\bm X})\|_{\rm F}^2/K}{N} - \frac{\| \bm X -\underline{\bm X} \|_{\rm F}^2/K}{IJ} \right| \geq \lambda\nu\right)\nonumber\\
        &\le
        2\exp\!\left(
        -\frac{2N(\lambda\nu)^2}{\big(1-\frac{N-1}{IJ}\big)\gamma_{\rm x}^4}
        \right),
    \end{align}
    where the probability is taken over the sampling pattern $\bm\varOmega$, or equivalently, over $\bm A$.

    \paragraph{{Inequality on $\| \bm w\|_2^2 + 2\bm d^\top \bm w$.}}
    
    Notice that $\bm w\sim\mathcal N(\bm 0,\sigma^2\bm I_{NK})$. Therefore, fix any $\bm d$, $(\| \bm w\|_2^2 + 2\bm d^\top \bm w + \| \bm d\|_2^2)/\sigma^2$ follows a noncentral chi-square distribution with $NK$ degrees of freedom and noncentrality parameter $\|\bm d\|_2^2/\sigma^2$ \cite{patnaik1949non}. By the concentration inequality for the noncentral chi-square distribution (Lemma 8.1 in \cite{birge2001alternative}), for every $u>0$,
    \begin{align}
        &{\rm Pr}\Big( 
        \left|\| \bm w\|_2^2 + 2\bm d^\top \bm w - NK\sigma^2 \right| \nonumber \\
        &\qquad \geq  2\sigma^2u + 2\sqrt{(NK\sigma^4+2\sigma^2\| \bm d \|_2^2) u} \ \Big | \ \bm d  \Big) \leq 2 e^{-u}, \nonumber 
    \end{align}
    where the probability is taken over the noise $\bm w$.
    Since \(\bm x_i,\underline{\bm x}_i\in[-\gamma_{\rm x}/2,\gamma_{\rm x}/2]\) for any $i\in[IJK]$, we have
    \(
    \|\bm d\|_2^2 \le NK\gamma_{\rm x}^2.
    \)
    Hence, with
    \(
    \zeta \triangleq NK\sigma^4+2NK\gamma_{\rm x}^2\sigma^2,
    \)
    it follows that 
    \begin{align}
    {\rm Pr}\!\left(
    \left|\|\bm w\|_2^2 + 2\bm d^\top \bm w - NK\sigma^2\right|
    \ge 2\sigma^2 u + 2\sqrt{\zeta\, u}
    \,\middle|\, \bm d
    \right)
    \le 2e^{-u}.\nonumber
    \end{align}
    Equivalently, for any $t>0$,
    \begin{align}\label{eq:tail_bound}
        &{\rm Pr}\left( 
        \left|\| \bm w\|_2^2 + 2\bm d^\top \bm w - NK\sigma^2 \right| \geq t \mid \bm d  \right) \nonumber\\
        &\leq 2\exp\left( - (\frac{\sqrt{2t\sigma^2+\zeta}-\sqrt{\zeta}}{2\sigma^2})^2 \right)\nonumber\\
        &=2\exp\left( - \frac{t^2}{(\sqrt{2t\sigma^2+\zeta}+\sqrt{\zeta})^2} \right)\nonumber\\
        & \leq 2\exp\left( - \frac{t^2}{4t\sigma^2+4(NK\sigma^4 + 2NK\gamma_{\rm x}^2\sigma^2)} \right).
    \end{align}

    \medskip
    With the above inequalities, we are now ready to bound the two error probabilities in \eqref{eq:Ea_bound}.

    \subsubsection{Inner case} 
    For $\mathrm E_{\bm A} [  P_{\bm y}^{\rm in}(\mathbb Y^\complement)] $ with $\mathbb Y^\complement = \mathbb B_{\rm y}^\complement$, we upper bound the equivalent set $\mathbb B_{\rm x,w}^\complement$ as 
    \begin{align} \label{eq:inner_set_bound}
        \mathbb B_{\rm x,w}^\complement \subseteq {\mathbb D^\complement} \cup  \mathbb W^\complement,
    \end{align}
    with 
    \begin{subequations}
        \begin{align}
            \mathbb D \triangleq \{& (\bm x\in \mathbb B_{\rm in},\bm w \in \mathbb R^{NK}): \nonumber\\
            &\left| \frac{\|\bm A(\bm x - \underline{\bm x})\|_2^2}{NK}
        -\frac{\|\bm x-\underline{\bm x}\|_2^2}{IJK} \right|
        \le \lambda\nu \},\\
            \mathbb W \triangleq \{& (\bm x\in \mathbb R^{IJK},\bm w \in \mathbb R^{NK}): \nonumber \\
            &\Big| \| \bm w\|_2^2 + 2\bm d^\top \bm w -NK\sigma^2 \Big| \leq NK(1-\lambda) \nu \}. \label{eq:W_set}
        \end{align}
    \end{subequations}
    To see this, note that 
    \begin{align}
        \mathbb B_{\rm x,w}^\complement = ({\mathbb D^\complement} \cap \mathbb B_{\rm x,w}^\complement) \cup ({\mathbb D} \cap \mathbb B_{\rm x,w}^\complement)  \subseteq {\mathbb D^\complement}  \cup ({\mathbb D} \cap \mathbb B_{\rm x,w}^\complement).
    \end{align}
    In addition, ${\mathbb D} \cap \mathbb B_{\rm x,w}^\complement$ implies that
    \begin{subequations}
        \begin{align}
            \| \bm d \|_2^2 &= \| \bm A(\bm x -\underline{\bm x}) \|_2^2 \nonumber\\
            &  \leq NK (\lambda \nu + \frac{\| \bm x -\underline{\bm x}\|_2^2}{IJK}) \leq NK (\lambda \nu + {\epsilon^2}),\nonumber
        \end{align}
        \begin{align}
            \| \bm d \|_2^2 + \| \bm w\|_2^2 + 2\bm d^\top \bm w &> {NK(\frac{c^2+1}{2}\epsilon^2}+\sigma^2)\nonumber\\
            & =NK(\nu+{\epsilon^2}+\sigma^2).\nonumber
        \end{align}
    \end{subequations}
    which further leads to
    \begin{align}
        \| \bm w\|_2^2 + 2\bm d^\top \bm w -NK\sigma^2 \geq NK(1-\lambda)\nu.
    \end{align}
    Therefore, ${\mathbb D} \cap \mathbb B_{\rm x,w}^\complement \subseteq \mathbb W^\complement$.
    
    With \eqref{eq:inner_set_bound}, the following can be obtained:
    \begin{align}\label{eq:inner_bound}
        &\mathbb E_{\bm A} [  P_{\bm y}^{\rm in}(\mathbb Y^\complement)] \nonumber\\
        =&\mathbb E_{\bm A} \Big[  \mathbb E_{\bm x\sim p_{\rm in}} \big[ \mathbb E_{\bm w}[\mathbb 1 (\mathbb B_{\rm x,w}^\complement)]\big] \Big] \nonumber\\
        \leq&  \mathbb E_{\bm A} \Big[  \mathbb E_{\bm x\sim p_{\rm in}} \big[ \mathbb E_{\bm w}[\mathbb 1(\mathbb D^\complement)]\big]  + \mathbb E_{\bm x\sim p_{\rm in}} \big[ \mathbb E_{\bm w}[\mathbb 1(\mathbb W^\complement)] \big]\Big]  \nonumber \\
        =& \mathbb E_{\bm x\sim p_{\rm in}} \big[\mathbb E_{\bm w} [ {\rm Pr}(\mathbb D^\complement) ] \big]+ \mathbb E_{\bm x\sim p_{\rm in}} \big[ \mathbb E_{\bm A}[{\rm Pr}(\mathbb W^\complement)]\big]\nonumber\\
        \leq& 2\exp\!\left(
        -\frac{2N(\lambda\nu)^2}{\big(1-\frac{N-1}{IJ}\big)\gamma_{\rm x}^4} 
        \right) \nonumber\\
        &
        + 2\exp\left(  -\frac{NK(1-\lambda)^2\nu^2}{4\sigma^2\big(\sigma^2+2\gamma_{\rm x}^2+(1-\lambda)\nu\big)} \right),
    \end{align}
    where $\mathbb 1(\cdot)$ is an indicator function.
    In the last inequality of \eqref{eq:inner_bound}, the first part comes directly from \eqref{eq:Hoeffding_block}, and the second part is from \eqref{eq:tail_bound} by setting $t = NK(1-\lambda)\nu$.

\subsubsection{Outer case} Similarly, for $\mathbb E_{\bm A} [P_{\bm y}^{\rm out}(\mathbb Y) ]$ with $\mathbb Y = \mathbb B_{\rm y}$, we bound the equivalent set $\mathbb B_{\rm x,w}$ as
\begin{align}
    \mathbb B_{\rm x,w} \subseteq {\mathbb F}^\complement \cup \mathbb W^\complement,
\end{align}
with
\begin{align}
     \mathbb F \triangleq \{& (\bm x\in \mathbb B_{\rm out},\bm w \in \mathbb R^{NK}):\nonumber\\
     &\left| \frac{\|\bm A(\bm x - \underline{\bm x})\|_2^2}{NK}
    -\frac{\|\bm x-\underline{\bm x}\|_2^2}{IJK} \right|
    \le \lambda\nu \},
\end{align}
and $\mathbb W$ defined as in \eqref{eq:W_set}.
Then the outer case is handled in the same way as the inner case, which yields
\begin{align}\label{eq:outer_bound}
    &\mathbb E_{\bm A} [  P_{\bm y}^{\rm out}(\mathbb Y)] 
    \leq  2\exp\!\left(
    -\frac{2N(\lambda\nu)^2}{\big(1-\frac{N-1}{IJ}\big)\gamma_{\rm x}^4} 
    \right)  \nonumber\\
    & \qquad 
    + 2\exp\left(  -\frac{NK(1-\lambda)^2\nu^2}{4\sigma^2\big(\sigma^2+2\gamma_{\rm x}^2+(1-\lambda)\nu\big)} \right).
\end{align}

Combine \eqref{eq:Ea_bound}, \eqref{eq:inner_bound} and \eqref{eq:outer_bound}, the proof is completed.

\end{IEEEproof}

\section{Proof of Corollary \ref{cor:cover_lowrank}}
\label{apd:RM_set}

Following Assumption~\ref{assump:structured_latent} and the low rank assumption on $\bm S_r$, we define the admissible sets for the spatial and spectral components as
\begin{align} 
    \mathbb S_{\rm lr} &\triangleq \{\bm S\in\mathbb R^{I\times J}:\ \mathrm{rank}(\bm S)\le r_{\rm s}, \|\bm S\|_{\rm F}\le \gamma_{\rm s}\}, \\
    \mathbb K &\triangleq \{ \bm c\in \mathbb R^K: \|\bm c \|_2 \leq \gamma_{\rm c}\}.
\end{align}
By the covering-number bound for Frobenius-bounded low-rank matrices \cite[Lemma~A.2]{wang2012stability} and the standard volumetric bound for Euclidean balls \cite[Corollary~4.2.11]{vershynin2018high}, we have
\begin{equation}
{\rm N}(\mathbb S_{\rm lr},\epsilon)
\le
\Big(\frac{9\gamma_{\rm s}}{\epsilon}\Big)^{(I+J+1)r_{\rm s}}, \quad {\rm N}(\mathbb K,\epsilon)
\le
\left(\frac{3\gamma_{\rm c}}{\epsilon}\right)^K. 
\end{equation}
Then the covering number of $\mathbb X_{\rm lr}$ satisfies \cite{Shrestha2022DeepSpectrum,xu2025radio}
\begin{align}
&\mathrm N(\mathbb X_{\rm lr},(\gamma_{\rm c} +  \gamma_{\rm s})R\epsilon')
\le
\left(
\mathrm N(\mathbb S_{\rm lr},\epsilon')\,
\mathrm N(\mathbb K,\epsilon')
\right)^R \nonumber \\
=&\left(\frac{9\gamma_{\rm s}}{\epsilon'}\right)^{(I+J+1)r_{\rm s}R} \left(\frac{3\gamma_{\rm c}}{\epsilon'}\right)^{KR} \nonumber\\
=& \left(\frac{3}{\epsilon'}\right)^{R(Ir_{\rm s}+Jr_{\rm s}+r_{\rm s}+K)} (3\gamma_{\rm s})^{Rr_{\rm s}(I+J+1)} \gamma_{\rm c}^{RK}.\nonumber
\end{align}
Setting $\epsilon=R(\gamma_{\rm c}+\gamma_{\rm s})\epsilon'$, we obtain
\begin{align}\label{eq:X_lr_covering}
    \mathrm N(\mathbb X_{\rm lr},\epsilon) \leq &\left(\frac{3(\gamma_{\rm c} +  \gamma_{\rm s})R}{\epsilon}\right)^{R(Ir_{\rm s}+Jr_{\rm s}+r_{\rm s}+K)} \nonumber\\
    &\times (3\gamma_{\rm s})^{Rr_{\rm s}(I+J+1)} \gamma_{\rm c}^{RK}.
\end{align}

Then the proof of Theorem~\ref{thm:imperfect_TV} applies exactly here, with the same parameter choices, except that the covering number is replaced by the low-rank covering bound in \eqref{eq:X_lr_covering}. This completes the proof.

\section{Network Architecture}
\label{apd:network}
Tables \ref{tab:slf_unet} and \ref{tab:psd_unet} summarize the network architectures used for SLF and PSD training, respectively. Table \ref{tab:NNmodule} further provides brief descriptions of the individual modules used in these network designs.

\begin{table}[t!]
\centering
\caption{The diffusion network for SLFs.}
\label{tab:slf_unet}
\scalebox{0.6}{
\begin{tabular}{||c||l|c|| c ||l|c||}
\hline
 & \multicolumn{2}{c||}{\textbf{Encoder}}  &  & \multicolumn{2}{c||}{\textbf{Decoder}} \\
\hline
\textbf{Resolution} & \textbf{Module} & \textbf{Channels} & \textbf{Skip conn.} & \textbf{Module} & \textbf{Channels} \\
\hline

\multirow{4}{*}{$256\times256$}
& Input & 2 & & Output & 1\\
& Conv2d($3\times3$)
& $64$
& $\rightarrow$
& ResBlock 3
& $64$ \\
& ResBlock 1
& $64$
& $\rightarrow$
& ResBlock 2
& $64$ \\
& ResBlock 2
& $64$
& $\rightarrow$
& ResBlock 1
& $64$ \\
\hline

\multirow{4}{*}{$128\times128$}
& Downsample
& $64$
& $\rightarrow$
& ResBlock 3
& $128$ \\
& ResBlock 1
& $128$
& $\rightarrow$
& ResBlock 2
& $128$ \\
& ResBlock 2
& $128$
& $\rightarrow$
& ResBlock 1
& $128$ \\
&  & &
& Upsample
& $128$ \\
\hline

\multirow{4}{*}{$64\times64$}
& Downsample
& $128$
& $\rightarrow$
& ResBlock 3
& $256$ \\
& ResBlock 1
& $256$
& $\rightarrow$
& ResBlock 2
& $256$ \\
& ResBlock 2
& $256$
& $\rightarrow$
& ResBlock 1
& $256$ \\
&  &  &
& Upsample
& $256$ \\
\hline

\multirow{4}{*}{$32\times32$}
& Downsample
& $256$
& $\rightarrow$
& ResBlock 3
& $512$ \\
& ResBlock 1
& $512$
& $\rightarrow$
& ResBlock 2
& $512$ \\
& ResBlock 2
& $512$
& $\rightarrow$
& ResBlock 1
& $512$ \\
&  & &
& Upsample
& $512$ \\
\hline

\multirow{5}{*}{$16\times16$}
& Downsample
& $512$
& $\rightarrow$
& [ResBlock+Attn] 3
& $512$ \\
& [ResBlock+Attn] 1
& $512$
& $\rightarrow$
& [ResBlock+Attn] 2
& $512$ \\
& [ResBlock+Attn] 2
& $512$
& $\rightarrow$
& [ResBlock+Attn] 1
& $512$ \\
\cline{2-6}
& \multicolumn{5}{c||}{\textbf{Bottleneck: ResBlock+Attn+ResBlock (Channels: 512)}}\\
\hline
\end{tabular}}
\end{table}

\begin{table}[t!]
\centering
\caption{The diffusion network for PSDs.}
\label{tab:psd_unet}
\scalebox{0.6}{
\begin{tabular}{||c||l|c|| c ||l|c||}
\hline
 & \multicolumn{2}{c||}{\textbf{Encoder}}  &  & \multicolumn{2}{c||}{\textbf{Decoder}} \\
\hline
\textbf{Length} & \textbf{Module} & \textbf{Channels} & \textbf{Skip conn.} & \textbf{Module} & \textbf{Channels} \\
\hline

\multirow{4}{*}{$64$}
& Input & 1 & & Output & 1\\
& Conv1d($3$)
& $64$
& $\rightarrow$
& ResBlock 3
& $64$ \\
& ResBlock 1
& $64$
& $\rightarrow$
& ResBlock 2
& $64$ \\
& ResBlock 2
& $64$
& $\rightarrow$
& ResBlock 1
& $64$ \\
\hline

\multirow{4}{*}{$32$}
& Downsample
& $64$
& $\rightarrow$
& ResBlock 3
& $128$ \\
& ResBlock 1
& $128$
& $\rightarrow$
& ResBlock 2
& $128$ \\
& ResBlock 2
& $128$
& $\rightarrow$
& ResBlock 1
& $128$ \\
&  & &
& Upsample
& $128$ \\
\hline

\multirow{4}{*}{$16$}
& Downsample
& $128$
& $\rightarrow$
& [ResBlock+Attn] 3
& $256$ \\
& [ResBlock+Attn] 1
& $256$
& $\rightarrow$
& [ResBlock+Attn] 2
& $256$ \\
& [ResBlock+Attn] 2
& $256$
& $\rightarrow$
& [ResBlock+Attn] 1
& $256$ \\
&  & &
& Upsample
& $256$ \\
\hline

\multirow{4}{*}{$8$}
& Downsample
& $256$
& $\rightarrow$
& ResBlock 3
& $512$ \\
& ResBlock 1
& $512$
& $\rightarrow$
& ResBlock 2
& $512$ \\
& ResBlock 2
& $512$
& $\rightarrow$
& ResBlock 1
& $512$ \\
&  & &
& Upsample
& $512$ \\
\hline

\multirow{5}{*}{$4$}
& Downsample
& $512$
& $\rightarrow$
& ResBlock 3
& $512$ \\
& ResBlock 1
& $512$
& $\rightarrow$
& ResBlock 2
& $512$ \\
& ResBlock 2
& $512$
& $\rightarrow$
& ResBlock 1
& $512$ \\
\cline{2-6}
& \multicolumn{5}{c||}{\textbf{Bottleneck: ResBlock+Attn+ResBlock (Channels: 512)}}\\
\hline
\end{tabular}}
\end{table}

\begin{table}[t!]
    \caption{Description on the network modules. The SLF network uses 2D operators, whereas the PSD network uses their 1D counterparts.}
    \label{tab:NNmodule}
    \centering
    \scalebox{0.6}{
    \begin{tabular}{|c|c|}
        \hline
        \textbf{Module} & \textbf{Descriptions} \\
        \hline
         \textbf{Input} & \textbf{SLF}: Noisy SLF + building map \\
         & \textbf{PSD}: Noisy PSD \\
        \hline
        \textbf{Output} & GroupNorm(32) $\rightarrow$ SiLU $\rightarrow$ Dropout(0.2) \\
        & $\rightarrow$ Conv2D/Conv1D (kernel 3, stride 1, padding 1) \\
        \hline
        \textbf{Downsample} & Conv2D/Conv1D (kernel 3, stride 2, padding 1) \\
        \hline
        \textbf{Upsample} & Neareat Neighbor Interpolation (scale factor 2) \\
        & $\rightarrow$ Conv2D/Conv1D (kernel 3, stride 1, padding 1)\\
        \hline
        \textbf{ResBlock} & [GroupNorm(32)$\rightarrow$SiLU$\rightarrow$Conv2D/Conv1D (kernel 3, stride 1, padding 1)] \\
        & $\rightarrow$ Add time embedding  \\
        & $\rightarrow$ [GroupNorm(32)$\rightarrow$SiLU$\rightarrow$Conv2D/Conv1D (kernel 3, stride 1, padding 1)] \\
        & $\rightarrow$ Add residual input\\
        \hline
        \textbf{Attn} & GroupNorm(32) $\rightarrow$ Extract Query, Key, and Value matrices \\
        &  via Conv2D/Conv1D (kernel 1, stride 1, padding 0) \\
        & $\rightarrow$ Self-attention over spatial/spectral domain \\
        & $\rightarrow$ Output projection via Conv2D/Conv1D (kernel 1, stride 1, padding 0) \\
        & $\rightarrow$ Add residual input \\
        \hline
    \end{tabular}
    }
\end{table}


\end{document}